\documentclass[11pt]{article}
\usepackage{xr-hyper}
\usepackage[colorlinks,citecolor=blue]{hyperref}
\usepackage{amsfonts,amsmath,amssymb,amsthm,booktabs,multirow}
\usepackage[pdftex]{graphicx}
\usepackage[hmargin=1in,vmargin=1in]{geometry}
\usepackage{natbib,setspace,enumerate,mathtools}
\usepackage{courier}
\usepackage{changepage}
\allowdisplaybreaks
\usepackage[toc,title,titletoc,header]{appendix}
\usepackage{etoolbox} % This package goes here to add \qed to remarks automatically.
\usepackage{tikz}
\def\qed{\rule{2mm}{2mm}}

\newtheorem{theorem}{Theorem}[section]
\newtheorem{lemma}{Lemma}[section]
\newtheorem{definition}{Definition}[section]
\newtheorem{corollary}{Corollary}[section]

\newtheorem{example}{Example}[section]
\AtEndEnvironment{example}{~\qed}
\newtheorem{remark}{Remark}[section]
\newtheorem{assumption}{Assumption}[section]
\newtheorem{algorithm}{Algorithm}[section]

\begin{document}

\title{\Large Testing homogeneity in dynamic discrete games in finite samples\thanks{%The first arXiv date:  ????????. 
We thank St\'ephane Bonhomme and two anonymous referees for comments and suggestions that have greatly improved the manuscript. We also greatly benefited from helpful comments and suggestions from participants of various seminars and conferences where this paper was presented. Of course, all errors are our own. Takuya Ura acknowledges Small Grants from UC Davis. Portions of this research were conducted with the advanced computing resources provided by Texas A\&M High Performance Research Computing.
%We thank Ivan Canay, the participants of the 2021 Bristol Econometric Study Group, the 2021 North American Summer Meetings of the Econometric Society, and the seminar participants at University of Tokyo, Boston College, Duke, Northwestern, Toulouse School of Economics, University of Washington Seattle, University College London, and Monash for helpful comments and suggestions. Takuya Ura acknowledges Small Grants from UC Davis. Portions of this research were conducted with the advanced computing resources provided by Texas A\&M High Performance Research Computing.
}% Bugni acknowledges support by TBA. TBA}
}
\author{Federico A. Bugni\\ Northwestern University\\ \url{federico.bugni@northwestern.edu}
\and
Jackson Bunting\\ Texas A\&M University\\ \url{jbunting@tamu.edu}
\and
Takuya Ura\\ University of California, Davis\\ \url{takura@ucdavis.edu}}
\maketitle
\vspace{-.5cm}
\begin{abstract}
The literature on dynamic discrete games often assumes that the conditional choice probabilities and the state transition probabilities are homogeneous across markets and over time. We refer to this as the ``homogeneity assumption'' in dynamic discrete games. This assumption enables empirical studies to estimate the game's structural parameters by pooling data from multiple markets and from many time periods. In this paper, we propose a hypothesis test to evaluate whether the homogeneity assumption holds in the data. Our hypothesis test is the result of an approximate randomization test, implemented via a Markov chain Monte Carlo (MCMC) algorithm. We show that our hypothesis test becomes valid as the (user-defined) number of MCMC draws diverges, for any fixed number of markets, time periods, and players. We apply our test to the empirical study of the U.S.\ Portland cement industry in \cite{ryan:2012}.
\begin{description}
\item[Keywords and phrases:]
Dynamic discrete choice problems, dynamic games, Markov decision problems, randomization tests, Markov chain Monte Carlo (MCMC).
\item[JEL classification:] C12, C57, C63
\end{description}
\end{abstract}
\newpage
%\tableofcontents
%-------- End Title Page -----------------------------

\linespread{1.5}
\section{Introduction}
% Explain homogeneity assumption DONE
In applications of dynamic discrete games, practitioners often assume that the conditional choice probabilities and the state transition probabilities are invariant across time and markets.\footnote{In this paper, we use ``market'' to denote a cross-sectional unit.} We refer to this as the ``homogeneity assumption'' in dynamic discrete games. This is a convenient assumption, as it allows the estimation of the model's structural parameters by pooling data from multiple markets and from many time periods. 

% Potential reason for the failure of the homogeneity assumption  DONE
Despite the widespread use of the homogeneity assumption in dynamic discrete games, it is plausible for this condition to fail in applications. We now provide a few examples. First, a game could suffer from a structural break in the model, which would invalidate the homogeneity assumption across time. Second, markets could be affected by persistent heterogeneity that is observed by the players but not by the researcher (e.g., \cite{arcidiacono/miller:2011}). This would invalidate the homogeneity assumption across markets. Third and relatedly, there may be multiplicity of equilibria, and different markets could be playing different equilibria. The literature has considered hypothesis testing for the multiplicity of equilibria in games. In particular, \cite{depaula/tang:2012} propose a test for the multiplicity of equilibria across markets in static games, while \cite{otsu/pesendorfer/takayashi:2016} do this in the context of dynamic games. 

% What we have done: DONE, and moved first 
In this paper, we propose a hypothesis test for the homogeneity assumption. That is, our test is designed to capture the various possible violations of the homogeneity assumption described in the previous paragraph (both across markets and over time). Our test is implemented via Markov chain Monte Carlo (MCMC) methods, and it is justified by the theory of randomization tests (cf.\ \citealp[Section 15.2]{lehmann/romano:2005}). While our MCMC test is not a standard randomization test, we establish its validity by coupling it with an underlying randomization test that is valid in finite samples yet computationally infeasible in practically relevant applications. Our contribution is to show that the distribution generated by our MCMC algorithm approximates the underlying randomization test as the number of (user-defined) MCMC draws diverges. In this sense, we interpret our proposed MCMC algorithm as a computationally feasible way to implement a computationally infeasible underlying randomization test. Moreover, our results are {\it valid in finite samples} in the sense that they hold for any fixed and finite number of players, markets, and time periods. This is an important aspect of our contribution, as the datasets used in empirical applications often have a small number of time periods and markets. For example, our empirical application is based on \cite{ryan:2012}, and has only $n=23$ markets and either $T=9$ or $T=10$ time periods.

% Why do we restrict to DDC games.
Our methodology is especially well suited for our testing problem in dynamic discrete games. As with standard randomization tests, the quality of our test depends on the richness of possible transformations exploited by our MCMC algorithm. There are three aspects of the typical framework in dynamic discrete games that allow us to generate a rich set of such transformations. First, dynamic discrete game models often impose independence across markets and have a Markovian structure. This provides the basis for our randomization-based methodology. Second, the data used in dynamic discrete games are usually either naturally discrete or discretized by the researcher. This is important for our test, as it relies on transformations defined by ``conditional'' permutations of data, i.e., permutations of the value of the state and the action across markets or time periods for which other related states coincide. Third, dynamic discrete game models often assume that actions are independent conditional on the state. This allows us to simplify the implementation of the transformations that generate our randomization-based methodology. It is worth pointing out that our inference methodology should apply to economic problems beyond the dynamic discrete game setup provided that these three aspects hold.

% Generality of the model DONE 
The econometric framework considered in this paper is arguably very general. It includes the single-agent dynamic discrete choice model (e.g., \cite{rust:1987,hotz/miller:1993,hotz/miller/sanders/smith:1994,aguirregabiria/mira:2002}) and the Markov equilibrium dynamic game model (e.g., \cite{pakes/ostrovsky/berry:2007,aguirregabiria/mira:2007,bajari/benkard/levin:2007,pesendorfer/schmidt-dengler:2008,pesendorfer/schmidt-dengler:2010}).  Furthermore, it includes the Markov dynamic game model of \cite{aguirregabiria/magesan:2020}, which allows some players to have biased beliefs. Importantly, our methodology does not impose functional forms on the primitive structure of the dynamic game (e.g., utility functions, state transition functions, etc.). We consider this advantageous, as functional form assumptions are hard to justify solely based on economic arguments, and are thus prone to misspecification.

% Explain Otsu
In a recent paper, \cite{otsu/pesendorfer/takayashi:2016} propose several hypothesis tests for dynamic discrete games. Two of their proposals are directly related to the problem considered in our paper.\footnote{The other two testing methodologies are less related to our paper. One test assumes that the state distribution is in its steady state. This condition is not commonly imposed in the literature, and our test does not require it. The other test they propose is based on the frequencies of states conditional on the state distribution in the first period.} Specifically, they consider a method to test the homogeneity across markets of the conditional choice probabilities and the state transition probabilities, under the maintained assumption that these functions are time-homogeneous. Their inference method is based on the bootstrap, and its validity is shown in an asymptotic framework in which the number of time periods $T$ diverges to infinity. However, $T$ is often small in applications. Besides the aforementioned application of \cite{ryan:2012} with $T=9$ or $T=10$, we can mention \cite{sweeting:2011} with $T=4$, \cite{collard-wexler:2013} with $T=24$, and \cite{dunne/klimek/roberts/xu:2013} with $T=5$. When $T$ is small, inference methods based on large-$T$ asymptotics, like the one provided by \cite{otsu/pesendorfer/takayashi:2016}, may yield inaccurate results. In contrast, our methodology is valid for any of the data dimensions, including $T$.

% technical contribution 1: DONE
The most critical step of our MCMC algorithm is based on the so-called Euler Algorithm, described in \cite{kandel/yossi/unger/winkler:1996}. In related work, \cite{besag/mondal:2013} uses the Euler Algorithm to test whether a time series of data has a time-homogeneous Markov structure. Relative to this work, our paper incorporates several essential features of dynamic Markov discrete games. First, we recognize that the dataset in a typical dynamic game has information about actions and states. Second, our construction exploits the typical economic structure imposed in dynamic games, such as the conditional independence assumption (i.e., conditional on the current state variable, the current action variable is independent of the past information). Finally, while \cite{besag/mondal:2013} mainly focuses on data from a single market (i.e., a single sequence in their setup), our MCMC algorithm exploits the possibility that the data includes observations from multiple markets.\footnote{Section 5 in \cite{besag/mondal:2013} briefly describes a few alternative ways to extend their methodology to the case of the multiple sequences of Markov chains (i.e., multiple markets). Their description does not include the formal statistical properties. In this paper, we propose a methodology that differs from theirs, describe its implementation in detail, and prove its validity by connecting it with the theory of randomization tests.} This is a valuable aspect of our contribution, as the datasets used in empirical applications usually include from data multiple markets, e.g., \cite{ryan:2012} with $n=23$, \cite{sweeting:2011} with $n=102$, \cite{collard-wexler:2013} with $n=1,600$, and \cite{dunne/klimek/roberts/xu:2013} with $n=639$.

% simulation and empirics: DONE
We explore the performance of our hypothesis test in Monte Carlo simulations. Our results show that our method provides excellent size control even in small samples, and can successfully detect relatively small deviations from the homogeneity hypothesis. In these two accounts, our test appears to work favorably in comparison with the bootstrap-based test in \cite{otsu/pesendorfer/takayashi:2016}. These favorable results appear to extend even when the discrete data have many support points, which is typical in applications. In our empirical example, we investigate the homogeneity of the decisions in the U.S.\ Portland cement industry data used in \cite{ryan:2012}. This is a key assumption in \cite{ryan:2012}, as it allows him to pool data from multiple markets to estimate the model's parameters. Unlike \cite{otsu/pesendorfer/takayashi:2016}'s test, our test finds no evidence against the homogeneity hypothesis in the data. We implement our test using the Julia package HomogeneityTestBBU, which is publicly available at the GitHub repository.\footnote{Use \texttt{Pkg.add("HomogeneityTestBBU")} to install the package in Julia. Package documentation is available in \href{https://jacksonbunting.github.io/HomogeneityTestBBU.jl/dev/}{https://jacksonbunting.github.io/HomogeneityTestBBU.jl/dev/}.}

% structure DONE
The rest of the paper is organized as follows. 
Section \ref{sec:model} describes the dynamic discrete choice model and the hypothesis test. Section \ref{sec:MCMC} specifies our hypothesis test and its implementation via our MCMC algorithm. Section \ref{sec:validity} establishes the main theoretical results of the paper. The main technical insight is that our hypothesis test is an approximate way of implementing an underlying randomization test that is finite-sample valid yet computationally infeasible. Section \ref{sec:empiric} provides the empirical application. In Section \ref{sec:MC}, we evaluate the performance of our test in finite samples via Monte Carlo simulations. Section \ref{sec:concl} concludes. The paper's appendix collects all of the proofs, auxiliary results, and computational details related to our proposed MCMC algorithm.

\section{Setup}\label{sec:model}

\subsection{The econometric model}
We begin by describing the dynamic discrete game under consideration. We observe the outcome of $n$ markets in which $J$ players choose actions over $T $ time periods. Our setup allows for $J=1$, i.e., single-agent problems, or $J>1$, i.e., multiple-agent games. This paper's inference results are valid for all finite $n$, $ T $, and $ J $. 

We consider a setup in which the observed actions and state variables are discretely distributed. This is common in the dynamic discrete choice literature, where the state and action variables are often naturally discrete or discretized by the researcher. For every market $i=1,\ldots ,n$ and period $t=1,\ldots ,T$, let $A_{i,t}$ be the random variable that specifies the actions chosen by the players in market $i$ and period $t$, and let $S_{i,t}$ be the random variable that specifies the state variable of market $i$ and period $t$. We use $\mathcal{S}$ to denote the common support of $S_{i,t}$ and $\mathcal{A}$ to denote the common support of $A_{i,t}$. We define the following $n\times T$ matrices: 
\begin{align*}
S &~\equiv~(S_{i,t}:~i=1,\ldots ,n,~t=1,\ldots ,T),\\
A &~\equiv~(A_{i,t}:~i=1,\ldots ,n,~t=1,\ldots ,T).
\end{align*}
In this notation, the data are then given by
$$
X~\equiv~(S,A).
$$
By definition, the support of $X$ is given by $\mathcal{X} \equiv\mathcal{S}^{nT} \times \mathcal{A}^{nT}$.

\begin{remark}
We assume a balanced panel (i.e., all $n$ markets are fully observed over $T$ time periods) only for the simplicity of notation and exposition. Our arguments extend immediately to the case in which each market $i=1,\dots,n$ is observed for $T_i$ time periods. 
\end{remark}

The following assumption is standard in much of the literature on dynamic discrete games.
\begin{assumption}\label{ass:M} 
The following conditions hold:
\begin{enumerate}[(a)]%[label=\it{(\alph*)}]
\item $((S_{i,t},A_{i,t}):t=1,\ldots ,T) $ and $((S_{j,t},A_{j,t}):t=1,\ldots ,T) $ are independent for all $i,j =1,\dots,n$ with $i\neq j$.
\item $(S_{i,t},A_{i,t})$ and $(S_{i,1},A_{i,1},\ldots,S_{i,t-2},A_{i,t-2})$ are conditionally independent given $(S_{i,t-1},A_{i,t-1})$ for every $i=1,\ldots ,n$ and $ t=3,\ldots, T$.
\item $A_{i,t}$ and $(S_{i,t-1},A_{i,t-1})$ are conditionally independent given $S_{i,t}$ for every $i=1,\ldots ,n$ and $t=2,\ldots, T$.
\end{enumerate}
\end{assumption}

Assumption \ref{ass:M} has three parts. Assumption \ref{ass:M}(a) imposes that markets are independently distributed. Assumption \ref{ass:M}(b) indicates that the observations of state and actions are a Markov process. Assumption \ref{ass:M}(c) imposes that the current actions are independent of past information once we condition on the current state. Assumptions \ref{ass:M}(b)-(c) are high-level restrictions that are typically imposed on the equilibrium strategies used by the players. In particular, they follow from the assumption that players use Markov strategies (e.g., \cite{maskin/tirole:2001}), as assumed in \cite{pakes/ostrovsky/berry:2007,aguirregabiria/mira:2007,bajari/benkard/levin:2007,pesendorfer/schmidt-dengler:2008}. These conditions are imposed even in models in which the players' beliefs are allowed to be out of equilibrium, i.e., do not coincide with the true equilibrium probabilities (e.g., \cite{aguirregabiria/magesan:2020}). Finally, we clarify that Assumption \ref{ass:M} refers to {\it observed} states and actions. As such, it may fail if there are state or action variables that are unobserved to the researcher and influence the distribution of their observed counterparts.

Assumption \ref{ass:M} is the only maintained assumption to study the validity of our hypothesis test. Notably, we do not impose functional forms on the primitive structure of the dynamic game, such as the utility functions, the state transition functions, or the discount factors. As explained earlier, we view this as a virtue of our methodology, as these restrictions can be hard to justify based only on economic arguments. Relatedly, while one could improve the statistical power of our test by exploiting functional form restrictions on the primitives of the dynamic game, they inevitably carry the risk of producing invalid inference when these are misspecified.

We now introduce the necessary notation to express our hypothesis of interest. We use $\sigma_{i,t}$ to denote the conditional choice probability for market $i$ and period $t$, i.e., for every $(s,a)\in\mathcal{S\times A}$,
$$
\sigma_{i,t}(a|s)~\equiv ~P(A_{i,t}=a|S_{i,t}=s).
$$
We use $f_{i,t+1}$ to denote the state transition probability from period $t$ to $t+1$ for market $i$, i.e., for every $(s,a,s^{\prime})\in\mathcal{S}\times\mathcal{A}\times\mathcal{S}$ ,
$$
f_{i,t+1}(s^{\prime}|a,s)~\equiv ~P(S_{i,t+1}=s^{\prime}|(S_{i,t},A_{i,t})=(s,a)).
$$
Finally, we use $m_{i}(s)$ to denote the marginal state distribution for market $i$ in period 1, i.e., for every $s\in\mathcal{S}$,
$$
m_{i}(s)~\equiv ~P(S_{i,1}=s).
$$
With this notation in place, we specify our hypothesis testing problem in the next section.

\subsection{The hypothesis testing problem}

Our goal is to test whether the ``homogeneity assumption'' holds in the data, i.e., whether the conditional choice probabilities and state transition probabilities are homogeneous across time and markets. That is,
\begin{equation}
H_{0}:\sigma_{i,t}(a|s)=\sigma(a|s)~~\text{and}~~f_{i,t+1}(s^{\prime}|a,s)=f(s^{\prime}|a,s). \label{eq:HT}
\end{equation}

Note that $H_{0}$ in \eqref{eq:HT} represents two types of homogeneity: time and market homogeneity, and involves two functions: conditional choice probabilities and state transition probabilities. In this sense, our hypothesis test evaluates four homogeneity conditions: time homogeneity of the conditional choice probabilities, market homogeneity of the conditional choice probabilities, time homogeneity of the state transition probabilities, and market homogeneity of the state transition probabilities. Under Assumption \ref{ass:M}, a rejection of $H_{0}$ would be indicative that one or more of these homogeneity conditions is violated, suggesting it is not appropriate to pool data across markets and time periods.

As discussed in the introduction, there may be many possible reasons for the homogeneity assumption to fail and, in general, it may be difficult to distinguish among the possible reasons. For example, given the recent literature on separately identifying equilibrium selection and market-specific permanent unobserved heterogeneity \citep{aguirregabiria/mira:2019,luo/xiao/xiao:2022}, it may be difficult to distinguish between these two possible causes for failure of the homogeneity assumption. Nevertheless, in certain applications, one may feel comfortable that some of the conditions are satisfied and should be part of our maintained assumptions. For example, in a given application, one may be confident that the conditional choice probability and state transition probability are time-homogeneous. Then, one could reinterpret $H_{0}$ as testing the market homogeneity of the conditional choice probabilities and state transition probabilities. Also, if one is confident that market and time homogeneity holds for some subsets of the time periods (e.g., before and after a policy change), then $H_{0}$ may be reinterpreted as testing homogeneity of the conditional choice probabilities and state transition probabilities across the subsets of periods.

Under Assumption \ref{ass:M} and $H_{0}$, Lemma \ref{lem:Lik1} in the appendix shows that the likelihood of the data ${X}=({S},{A})$ evaluated at any realization $\tilde{X}=(\tilde{S},\tilde{A})\in\mathcal{X}$ is as follows:
\begin{equation}
P(X=\tilde{X})~=~\prod_{i=1}^{n}\left(m_{i}(\tilde{S}_{i,1})~\left(\prod_{t=1}^{T}\sigma(\tilde{A}_{i,t}|\tilde{S}_{i,t})\right)~\left(\prod_{{t}=1}^{T-1}f(\tilde{S}_{i,{t} +1}|\tilde{S}_{i,{t}},\tilde{A}_{i,{t}})\right)\right). \label{eq:lik1}
\end{equation}
%Note that conditional choice probability function and state transition probability function have no subscript due to the imposition of $H_{0}$. 
This expression reveals that the markets are independently distributed (Assumption \ref{ass:M}(a)), but they are not necessarily identically distributed because $m_i(\cdot)$ may depend on $i$. That is, even though the conditional choice probabilities and state transition probabilities are homogeneous under $H_{0}$, markets can still be heterogeneous due to differences in their initial state values. This is a desired feature in our testing problem, as the dynamic discrete choice literature usually allows the initial state distribution to be market-specific.

We conclude the section with an observation about the type of economic models considered in this paper. From an econometric viewpoint, our goal is to evaluate the homogeneity assumption (i.e., $H_{0}$ in \eqref{eq:HT}) using discrete data that satisfies Assumption \ref{ass:M}. The discreteness of the data is necessary for our MCMC algorithm to find data transformations that can deliver non-trivial power. We motivated this problem using dynamic discrete choice games because they are an important class of models ideally suited to this econometric framework. However, it is worth highlighting that our methodology applies to any other discrete panel-data model that satisfies Assumption \ref{ass:M}.\footnote{We thank an anonymous referee for this observation.} To our knowledge, the main roadblock to applying our test beyond dynamic discrete games and single-agent problems is the requirement that the data be discrete.

\section{Our hypothesis test}\label{sec:MCMC}

In this paper, we propose to reject $H_{0}$ in \eqref{eq:HT} whenever the significance level $\alpha$ is larger than or equal to our $p$-value, which we denote by $\hat{p}_{K}$. That is,
\begin{equation}
 \phi_{K}(X)~\equiv~ 1\{\hat{p}_{K}\leq\alpha\}.
 \label{eq:testDefn}
\end{equation}
In turn, our $p$-value $\hat{p}_{K}$ is the result of constructing $K$ transformations of the data via our MCMC algorithm, which is specified in Section \ref{sec:def_algorith}. This MCMC algorithm produces $K$ sequential transformations of the data $X$, denoted by $(X^{(1)},\ldots,X^{(K)})$. Our $p$-value is then computed as follows
\begin{equation}
\hat{p}_{K}~\equiv ~\frac{1}{K}\sum_{k=1}^{K}1\{\tau(X^{(k)})\geq \tau(X)\},
\label{eq:pvalue}
\end{equation}
where $\tau:\mathcal{X} \to \mathbb{R}$ denotes the test statistic designed to detect departures from $H_{0}$ in the data. 

One notable feature of our hypothesis test is that its validity does not depend on the choice of the test statistic (see Theorem \ref{thm:sizeControl}). However, the power of our test depends on this choice. Example \ref{ex:tau_examples} below specifies several test statistics considered in the related literature and describes the types of heterogeneity they are designed to detect. In practice, the choice of the test statistic should be guided by the type of heterogeneity one is most interested in detecting.\footnote{While our test is valid for any test statistic, this does not extend to when multiple instances of our test are implemented with several test statistics. This can create a multiple-testing problem, thereby invalidating the conclusions.}

\begin{example}[Examples of test statistic $\tau(X)$]\label{ex:tau_examples}
\citet[Eqs.\ (4), (7)]{otsu/pesendorfer/takayashi:2016} propose the following test statistics:
\begin{align}
\tau_{1}(X)&~\equiv~ \sum_{i=1}^{n}\sum_{(a,s)\in\mathcal{A}\times\mathcal{S}}(\hat{\sigma}_{i}(a|s)-\hat{\sigma}(a|s))^{2}\left(\frac{\sum_{t=1}^{T}1\{S_{i,t}=s\} }{\hat{\sigma}(a|s)}\right)\notag \\
\tau_{2}(X)&~\equiv~ 2\sum_{i=1}^{n}\sum_{(a,s)\in\mathcal{A}\times\mathcal{S}}\hat{\sigma}_{i}(a|s)\ln \left(\frac{ \hat{\sigma}_{i}(a|s)}{\hat{\sigma}(a|s)}\right)\sum_{t=1}^{T}1\{S_{i,t}=s\} ,\label{eq:testStatOtsu}
\end{align}
where we interpret $0/0 =0$ and $0\times \ln(0)=0$, and, for each $i=1,\dots,n$ and ${(a,s)\in\mathcal{A}\times\mathcal{S}}$, we define
\begin{align*}
\hat{\sigma}_{i}(a|s)&~\equiv~ \frac{\sum_{t=1}^{T}1\{ A_{i,t}=a,S_{i,t}=s\} }{\sum_{t=1}^{T}1\{S_{i,t}=s\} } \\
\hat{\sigma}(a|s)&~\equiv~ \frac{\sum_{i=1}^{n}\sum_{t=1}^{T}1\{ A_{i,t}=a,S_{i,t}=s\} }{\sum_{i=1}^{n}\sum_{t=1}^{T}1\{S_{i,t}=s\} }.
\end{align*}
\citet[pages 531-2]{otsu/pesendorfer/takayashi:2016} argue that $\tau_{2}(X)$ has optimal power properties in the context of their asymptotic framework. Be this as it may, the optimality result based on asymptotic arguments may be inaccurate in empirical situations in which the number of markets $n$ and periods $T$ is small relative to other features of the problem, such as the number of players $J$, the state space $\mathcal{S}$, and the action space $\mathcal{A}$. To explore the finite sample properties of the two test statistics in \eqref{eq:testStatOtsu}, we entertain both in our Monte Carlo simulations and empirical application.

The test statistics in \eqref{eq:testStatOtsu} compare market-specific conditional choice probabilities with their pooled counterpart, so they are specifically designed to detect heterogeneity across markets. By construction, the resulting tests are ineffective in detecting the presence of time heterogeneity, such as a structural break. If one were specifically interested in detecting time heterogeneity, one might consider test statistics that compare time-specific conditional choice probabilities with their pooled counterpart, such as
\begin{align}
\tilde\tau_{1}(X)&~\equiv~ \sum_{t=1}^{T}\sum_{(a,s)\in\mathcal{A}\times\mathcal{S}}(\hat{\sigma}_{t}(a|s)-\hat{\sigma}(a|s))^{2}\left(\frac{\sum_{i=1}^{n}1\{S_{i,t}=s\} }{\hat{\sigma}(a|s)}\right)\notag \\
\tilde\tau_{2}(X)&~\equiv~ 2\sum_{t=1}^{T}\sum_{(a,s)\in\mathcal{A}\times\mathcal{S}}\hat{\sigma}_{t}(a|s)\ln \left(\frac{ \hat{\sigma}_{t}(a|s)}{\hat{\sigma}(a|s)}\right)\sum_{i=1}^{n}1\{S_{i,t}=s\} ,\label{eq:testStatOtsuB}
\end{align}
where, for each $t=1,\dots,T$, we define
\begin{align*}
\hat{\sigma}_{t}(a|s)~\equiv~ \frac{\sum_{i=1}^{n}1\{ A_{i,t}=a,S_{i,t}=s\} }{\sum_{i=1}^{n}1\{S_{i,t}=s\} }.
\end{align*}
These test statistics are specialized in detecting time heterogeneity but, for analogous reasons as before, are ineffective in detecting heterogeneity across markets. If one were interested in simultaneously detecting both sources of heterogeneity, it would be natural to combine the test statistics in \eqref{eq:testStatOtsu} and \eqref{eq:testStatOtsuB} into non-specialized statistics, such as
\begin{align}
 \tau_{1}(X)+\tilde\tau_{1}(X)~~~~~~~\text{ or }~~~~~~~\tau_{2}(X)+\tilde\tau_{2}(X).\label{eq:testStatOtsuC}
\end{align}
Naturally, the statistics specifically designed to detect a specific source of heterogeneity  (i.e., \eqref{eq:testStatOtsu} or \eqref{eq:testStatOtsuB}) will be more powerful in detecting it than the non-specialized statistics in \eqref{eq:testStatOtsuC}. As mentioned earlier, the choice of the test statistic should be guided by the type of heterogeneity one is most interested in detecting. Our hypothesis test is valid for any choice of test statistic, including all of those mentioned within this example.
\end{example}

\subsection{The MCMC algorithm}\label{sec:def_algorith}

Our MCMC algorithm requires some notation. Let $I= (I_1,I_2)$ denote an arbitrary pair of markets $I_1$ and $I_2$ in the data, i.e., $I_1,I_2 \in \{1,2,\dots,n\}$. We allow for $I_1=I_2$. We use $\mathcal{I}$ to denote the collection of all such pairs of markets, i.e., $|\mathcal{I}|=n^2$. We also define several sets.

\begin{definition}[$R_{S}(I,\breve{S})$]\label{def:RS}
For any $I = (I_1,I_2)\in\mathcal{I} $ and $ \breve{S}\in\mathcal{S}^{nT}$, $R_{S}(I,\breve{S})$ is the set of all $ \tilde{S}\in\mathcal{S}^{nT}$ satisfying the following conditions for all $s,s^{\prime}\in\mathcal{S}$:
\begin{enumerate}[(a)]
\item $\tilde{S}_{i,1}=\breve{S}_{i,1}$ for all $i=1,\ldots ,n$,
\item $\sum_{t=1}^{T-1}1\{ \tilde{S}_{i,t}=s,\tilde{S}_{i,t+1}=s^{\prime}\} =\sum_{t=1}^{T-1}1\{ \breve{S}_{i,t}=s,\breve{S}_{i,t+1}=s^{\prime}\} $ for all $i \not\in \{I_1,I_2\}$,
%\item $\tilde{S}_{i,T}=\breve{S}_{i,T}$ for all $i\in I^{c}$,
\item $\sum_{i= I_1,I_2}\sum_{t=1}^{T-1}1\{ \tilde{S}_{i,t}=s,\tilde{S}_{i,t+1}=s^{\prime}\} =\sum_{i= I_1,I_2}\sum_{t=1}^{T-1}1\{ \breve{S}_{i,t}=s,\breve{S}_{i,t+1}=s^{\prime}\} $.
%\item $\sum_{i\in I}1\{ \tilde{S}_{i,T}=s\} =\sum_{i\in I}1\{ \breve{S}_{i,T}=s\} $ for all $s\in\mathcal{S}$.
\end{enumerate}
\end{definition}

In words, $R_{S}(I,\breve{S})$ is the set of all state configurations that result from permuting the state data $\breve{S}$ subject to conditions (a)-(c), which we now interpret. First, condition (a) indicates that the initial value of the state variable must remain unchanged across markets. The reason behind this restriction is that our framework does not restrict the initial state distribution (i.e., $\{m_i(\tilde{S}_{i,1})\}_{i=1}^{n}$ in \eqref{eq:lik1}). In turn, conditions (b)-(c) imply that the aggregate state transition frequencies across all markets $i=1,\dots,n$ must remain constant. This restriction is achieved by requiring the state transition frequencies to remain invariant for each market $i \not\in \{I_1,I_2\}$ (by condition (b)) and on aggregate for markets $i \in \{I_1,I_2\}$ (by condition (c)). The main reason behind breaking an aggregate restriction into conditions (b) and (c) is computational tractability. Under Assumption \ref{ass:M} and $H_{0}$, conditions (a)-(c) imply that each state configuration in $R_{S}(I,\breve{S})$ has the same value of the likelihood function, provided that it is paired with a suitable action configuration. These suitable action configurations are precisely those in next definition.

\begin{definition}[$R_{A}(\tilde{S},(\breve{S},\breve{A}))$]\label{def:RA}
For any $\tilde{S},\breve{S}\in\mathcal{S}^{nT}$ and $ \breve{A}\in\mathcal{A}^{nT}$, $R_{A}(\tilde{S},(\breve{S},\breve{A}))$ is the set of all $\tilde{A}\in\mathcal{A}^{nT}$ satisfying the following conditions for all $s,s^{\prime}\in\mathcal{S}$ and $ a\in\mathcal{A}$:
\begin{enumerate}[(a)]
\item $\sum_{i=1}^{n}\sum_{t=1}^{T-1}1\{ \tilde{S}_{i,t}=s,\tilde{A}_{i,t}=a,\tilde{S}_{i,t+1}=s^{\prime}\} =\sum_{i=1}^{n} \sum_{t=1}^{T-1}1\{ \breve{S}_{i,t}=s,\breve{A}_{i,t}=a,\breve{S}_{i,t+1}=s^{\prime}\} $,
\item $\sum_{i=1}^{n}1\{ \tilde{S}_{i,T}=s,\tilde{A}_{i,T}=a\} =\sum_{i=1}^{n}1\{ \breve{S}_{i,T}=s,\breve{A}_{i,T}=a\} $.
\end{enumerate}
\end{definition}

By definition, $R_{A}(\tilde{S},(\breve{S},\breve{A}))$ is the set of action configurations that result from permuting the action data $\breve{A}$ subject to conditions (a)-(b), which we explain next. Condition (a) implies that the aggregate state and action transition frequencies across all markets $i=1,\dots,n$ remain constant. Condition (b) imposes an analogous requirement for the terminal period. Under Assumption \ref{ass:M} and $H_{0}$, these restrictions imply that the hypothetical data $(\breve{S},\breve{A})$ has the same likelihood as the state configuration $\tilde{S}$ paired with any action configuration in $R_{A}(\tilde{S},(\breve{S},\breve{A}))$.

Before explaining how $R_{S}(I,\breve{S})$ and $R_{A}(\tilde{S},(\breve{S},\breve{A}))$ are used in our MCMC algorithm, we illustrate their computation in a relatively simple example. While the conditions in Definitions \ref{def:RS} and \ref{def:RA} are not conceptually complicated, the example reveals that computing these sets explicitly requires thoughtful consideration, even in a relatively simple case.

\begin{example}[Computing $R_{S}(I,\breve{S})$ and $R_{A}(\tilde{S},(\breve{S},\breve{A}))$ in a simple case]\label{ex:Rsets}
    Consider a setup with $n=3$ markets, $T=4$ periods, supports $\mathcal{S} = \mathcal{A} = \{1,2,3,4\}$, and state and action data equal to
        \begin{equation}
       \breve{S} = \left(\begin{array}{cccc}
            1&2&4&3  \\
            2&1&4&3\\
            3&1&3&4  
       \end{array}\right)~~~~~~\text{and}~~~~~~  \breve{A} =  \left(\begin{array}{cccc}
            2&2&1&4  \\
            2&2&3&1\\
            1&3&3&1  
       \end{array}\right).\label{eq:exampleA_and_S}
    \end{equation}
    
    We now compute $R_{S}(I,\breve{S})$ when $I = (1,3)$ and $\breve{S}$ in \eqref{eq:exampleA_and_S}. By Definition \ref{def:RS}, $R_{S}(I,\breve{S})$ is composed of matrices of size $3\times 4$ formed by restricted permutations of $\breve{S}$. Condition (a) says that any matrix in $R_{S}(I,\breve{S})$ has a first column equal to the first column of $\breve{S}$, i.e., $(1,2,3)'$. Condition (b) implies that any matrix in $R_{S}(I,\breve{S})$ has its second row (i.e., $i \not\in \{I_1,I_2\}=\{1,3\}$) equal to the second row of $\breve{S}$, i.e., $(2,1,4,3)$. To see why, note that any other configuration of the second row would alter the value of  $\sum_{t=1}^{T-1}1\{ \tilde{S}_{2,t}=s,\tilde{S}_{2,t+1}=s^{\prime}\}$ for some $s,s' \in \mathcal{S} = \{1,2,3,4\}$. Finally, condition (c) implies two possible configurations of the first and third rows of any matrix in $R_{S}(I,\breve{S})$. Either these coincide with the corresponding rows of $\breve{S}$ (i.e., $(1,2,4,3)$  and $(3,1,3,4)$), or they are interchange information across these markets, and are equal to $(1,3,4,3)$ and $(3,1,2,4)$. Once again, any other configuration of the first and third rows would alter the value of $\sum_{i= 1,3}\sum_{t=1}^{T-1}1\{ \tilde{S}_{i,t}=s,\tilde{S}_{i,t+1}=s^{\prime}\} =\sum_{i= 1,3}\sum_{t=1}^{T-1}1\{ \breve{S}_{i,t}=s,\breve{S}_{i,t+1}=s^{\prime}\} $ for some $s,s' \in \mathcal{S} = \{1,2,3,4\}$. In conclusion, $R_{S}(I,\breve{S})$ only has two elements:
        \begin{equation*}
       \left(\begin{array}{cccc}
            1&2&4&3  \\
            2&1&4&3\\
            3&1&3&4  
       \end{array}\right)~~~~\text{and}~~~         \left(\begin{array}{cccc}
            1&3&4&3  \\
            2&1&4&3\\
            3&1&2&4  
       \end{array}\right).
    \end{equation*}
    This example also reveals that restricting condition (c) to a pair of markets (in this case, $i \in \{I_1,I_2\}=\{1,3\}$) helps to reduce the possible number of elements in $R_{S}(I,\breve{S})$. In particular, we are not considering exchanges of information between markets 1 and 2, or markets 2 and 3.
    
Next, we compute $R_{A}(\tilde{S},(\breve{S},\breve{A}))$ with $\tilde{S}=\breve{S}$ and $\breve{A}$ given by \eqref{eq:exampleA_and_S}. By Definition \ref{def:RA}, $R_{A}(\tilde{S},(\breve{S},\breve{A}))$ is composed of matrices of size $3\times 4$ formed by restricted permutations of $\breve{A}$. Condition (a) provides a rule to interchange action data within the first three columns of $\breve{A}$, i.e., $t< T=4$. The idea here is to scan the columns of $\tilde{S}$ for repetitions of pairs of consecutive states $(\tilde{S}_{i,t},\tilde{S}_{i,t+1})$ with $t\leq T=4$. It is not hard to verify that the only repetition occurs in the first two rows with $(\tilde{S}_{i,t},\tilde{S}_{i,t+1})=(4,3)$ for $i=1,2$ and $t=3$. This implies that condition (a) allows to interchange $\breve{A}_{1,3}=1$ and $\breve{A}_{2,3}=3$. In turn, condition (b) provides a rule to interchange action data within the last columns of $\breve{A}$, i.e., $T=4$. In this case, the idea is to scan repeated entries in the last column of $\breve{S}$. In this case, we have $S_{i,t}=3$ for $i=1,2$ and $t=4$. By condition (b), we can then interchange $\breve{A}_{1,4}=4$ and $\breve{A}_{2,4}=1$. Combining both possibilities, we conclude that  $R_{A}(\tilde{S},(\breve{S},\breve{A}))$ has four elements:
        \begin{equation*}
       \left(\begin{array}{cccc}
            2&2&1&4  \\
            2&2&3&1\\
            1&3&3&1
       \end{array}\right),~~         \left(\begin{array}{cccc}
            2&2&3&4  \\
            2&2&1&1\\
            1&3&3&1
       \end{array}\right),~~         \left(\begin{array}{cccc}
            2&2&1&1  \\
            2&2&3&4\\
            1&3&3&1
       \end{array}\right),~~\text{and}~~         \left(\begin{array}{cccc}
            2&2&3&1  \\
            2&2&1&4\\
            1&3&3&1
       \end{array}\right).
    \end{equation*}
Enumerating these sets is achievable with some consideration in this example, but the task becomes significantly more challenging in more complex data setups. In this respect, it is important to stress that implementing our test does not require enumerating these sets.
\end{example}

Having introduced and illustrated Definitions \ref{def:RS} and \ref{def:RA}, we now specify our MCMC algorithm.

\begin{algorithm}[MCMC algorithm]\label{alg:MCMC}
Let $(X^{(1)},\ldots,X^{(K)})$ denote the following Markov chain.
\begin{itemize}
    \item[(a)] Initiation. Initiate the chain with $X^{(1)}=X$.
    \item[(b)] Iteration. The rest of the chain is iteratively generated as follows. For any $k=2,\ldots ,K$ and given $(X^{(1)},\ldots,X^{(k-1)})$, $X^{(k)}=(S^{(k)},A^{(k)})$ is randomly generated as follows:
\begin{itemize}
\item Step 1: Draw $I^{(k)}$ uniformly from $\mathcal{I}$, independently of $(X^{(1)},\ldots,X^{(k-1)})$.
\item Step 2: Given $(X^{(k-1)},I^{(k)})$, draw $S^{(k)}$ uniformly from $ R_{S}(I^{(k)},S^{(k-1)})$.
\item Step 3: Given $(X^{(k-1)},I^{(k)},S^{(k)})$, draw $A^{(k)}$ uniformly from $ R_{A}(S^{(k)},X^{(k-1)})$.\hfill $\blacksquare$
\end{itemize}
\end{itemize}
\end{algorithm}

At each step $k=2,\dots,K$, our MCMC algorithm randomly permutes actions and states in the data. By construction, the algorithm implies the following transition probabilities for all $k=2,\dots,K$, $X^{(1)},\ldots,X^{(k-1)}\in\mathcal{X}$, $I\in\mathcal{I}$, and $\tilde{X}= (\tilde{S},\tilde{A})\in\mathcal{X}$,\begin{align}
P(I^{(k)}=I\mid X^{(1)},\ldots,X^{(k-1)})%&=P(I^{(k)}=I)
&~=~\frac{1}{|\mathcal{I}|},\label{eq:dist_0}\\
P(S^{(k)}=\tilde{S}\mid I^{(k)},X^{(1)},\ldots,X^{(k-1)})
%&=P(S^{(k)}=\tilde{S}\mid I^{(k)},S^{(k-1)})\notag \\
&~=~\frac{1\{\tilde{S}\in R_{S}(I^{(k)},{S}^{(k-1)})\}}{|R_{S}(I^{(k)},{S}^{(k-1)})|}, \label{eq:dist_1}\\
P(A^{(k)}=\tilde{A}\mid S^{(k)},I^{(k)},X^{(1)},\ldots,X^{(k-1)})
%&= P(A^{(k)}=\tilde{A}\mid S^{(k)},X^{(k-1)})\notag \\
&~=~\frac{1\{ \tilde{A}\in R_{A}(S^{(k)},X^{(k-1)})\}}{|R_{A}(S^{(k)},X^{(k-1)})|}. \label{eq:dist_2}
\end{align}
We note that \eqref{eq:dist_1} and \eqref{eq:dist_2} are well defined, as both denominators can be shown to be positive.
\subsection{Implementation of our MCMC algorithm}\label{section:imp_algorith}

Each iteration of the MCMC algorithm \ref{alg:MCMC} involves three steps. Step 1 is computationally and conceptually straightforward. Steps 2 and 3 require randomly drawing state and action configurations uniformly over the sets $ R_{S}(I^{(k)},S^{(k-1)})$ and $ R_{A}(S^{(k)},X^{(k-1)})$, respectively. As we argued in the context of Example \ref{ex:Rsets}, these sets may be difficult to enumerate even for simple data configurations. Importantly, our MCMC algorithm does not require us to enumerate these sets, but rather sample from them uniformly. The remainder of this section provides an overview of how we implement Steps 2 and 3. We defer to Section \ref{app:MCMC} for details.

Step 2 requires sampling $S^{(k)}$ uniformly from the set $ R_{S}(I^{(k)},S^{(k-1)})$. Given a pair of markets $I^{(k)}$ and state data in $S^{(k-1)}$, the restrictions considered in $ R_{S}(I^{(k)},S^{(k-1)})$ are relatively hard to implement. To construct a feasible implementation of step 2, we crucially rely on the Euler Algorithm (see \cite{kandel/yossi/unger/winkler:1996,besag/mondal:2013} for details). In particular, when both markets in $I^{(k)}$ are equal (i.e., $I^{(k)}=(i,i)$ for $i=1,2,\dots,n$), step 2 can be implemented by applying the Euler Algorithm for each market. Our marginal contribution in step 2 is to extend the Euler Algorithm to the case where the markets in $I^{(k)}$ differ. Our proposal is to concatenate the state information from both markets in $I^{(k)}$ and repeatedly apply the Euler Algorithm until two conditions hold: the initial state is the same in each market (i.e., $\tilde{S}_{i,1}=\breve{S}_{i,1}$) and the $T^{\text{th}}$ state in the first market satisfies $\tilde{S}_{I_1,T}\in\{\breve{S}_{I_1,T},\breve{S}_{I_2,T}\}$. The properties of the Euler Algorithm ensure that these two conditions imply the resulting chain belongs to $R_{S}(I^{(k)},S^{(k-1)})$. Importantly, these two conditions are far simpler to verify than the conditions in Definition \ref{def:RS}, which is one reason that our implementation is computationally feasible whereas the enumeration approach is not.\footnote{Not only are conditions simpler to verify, they may be verified without needing to complete the full $2\times T$ length chain. For example, if the candidate $\tilde{S}_{I_1,T}$ is not an element of $\{\breve{S}_{I_1,T},\breve{S}_{I_2,T}\}$, one may stop after constructing a $T$ length chain. Indeed, it is sometimes possible to verify that the second condition fails when the chain being constructed is of length $1<t<T$.} Relative to the market-by-market version of the algorithm, our modification typically generates a much larger set of data permutations, which tends to improve the power properties of our hypothesis test. We provide additional information about step 2 of the MCMC algorithm in Section \ref{sec:euler_algorithm} of the appendix, where we specify the original Euler Algorithm (Algorithm \ref{alg:Euler}) and our modification (Algorithm \ref{alg:S_k}), and we formally show that the latter exactly implements step 2 (see Lemma \ref{lem:EulerWorks}). Algorithm \ref{alg:S_k} draws $S^{(k)}$ uniformly from $ R_{S}(I^{(k)},S^{(k-1)})$, because we repeatedly sample uniformly from a superset of $R_{S}(I^{(k)},S^{(k-1)})$ until the realization belongs to $ R_{S}(I^{(k)},S^{(k-1)})$.

Step 3 requires sampling $A^{(k)}$ uniformly from the set $R_{A}(S^{(k)},X^{(k-1)})$. Given data in $X^{(k-1)}$ and state data in $S^{(k)}$, the restrictions considered in $R_{A}(S^{(k)},X^{(k-1)})$ are relatively easy to impose (compared to those in $ R_{S}(I^{(k)},S^{(k-1)})$). As a consequence, step 3 is computationally light. All we need to do is to permute the action data in $A^{(k-1)}$ subject to the simple restrictions in $R_{A}(S^{(k)},X^{(k-1)})$. Further details of step 3 are provided in Section \ref{sec:step3} of the appendix, where we specify an algorithm (Algorithm \ref{alg:A_k}) and we prove that it implements step 3 (see Lemma \ref{lem:Step3works}).

\section{Theoretical properties}\label{sec:validity}

We open this section with the main theoretical result of this paper.

\begin{theorem}\label{thm:sizeControl}
Under Assumption \ref{ass:M} and $H_{0}$ in \eqref{eq:HT}, the test in \eqref{eq:testDefn} satisfies
\begin{equation}
\underset{K\to\infty}{\lim \sup}~E[\phi_K(X)] ~~\leq~~ \alpha, \label{eq:asySizeControl}
\end{equation}
where the expectation is taken with respect to the randomness in $(X,X^{(1)},\ldots,X^{(K)})$, i.e., both in the data $X$ and in the random draws $(X^{(1)},\ldots,X^{(K)})$ generated by our MCMC algorithm.
\end{theorem}

Theorem \ref{thm:sizeControl} establishes that the proposed test in \eqref{eq:testDefn} controls size as the length of the MCMC draws $K$ diverges. We remark that $K$ is under the control of the researcher, who can increase $K$ to guarantee the convergence in \eqref{eq:asySizeControl}. Remarkably, Theorem \ref{thm:sizeControl} holds regardless of the number of markets $n$, time periods $T$, and players $J$, which remain constant in our analysis. In addition, and as promised in Section \ref{sec:MCMC}, this result also holds irrespective of the specific choice of test statistic $\tau(X)$ used in the construction of the $p$-value in \eqref{eq:pvalue}. Finally, we note that the inequality in equation \eqref{eq:asySizeControl} could be turned into equality by changing \eqref{eq:testDefn} to a random decision rule whenever $\phi_K(X)=\alpha$. We decided against this modification for the sake of simplicity.

An important practical consideration is how one should choose the number of MCMC draws $K$ in a given application. According to Theorem \ref{thm:sizeControl}, the size control of our test is guaranteed as $K$ diverges. The main drawback of increasing $K$ is the additional computation burden of implementing our test. In this sense, we recommend choosing $K$ as large as computationally possible. However, our theoretical results and practical experience can be combined to provide a more concrete recommendation regarding $K$. First, our theoretical results in later sections establish that the $p$-value in \eqref{eq:pvalue} used to implement our test converges as $K$ diverges.\footnote{In particular, see Lemma \ref{lem:MCMCconv} and the related result in \eqref{eq:Conv_pK}, which are building blocks of Theorem \ref{thm:sizeControl}.} Second, our experience from the empirical application and the Monte Carlo simulations suggests that the outcome of our test tends to become stable for a sufficiently large $K$. In conclusion, we recommend considering large values of $K$ (as large as computationally possible) and deciding on a value for which the test decision appears to become stable.

The key insight behind Theorem \ref{thm:sizeControl} is the connection between our hypothesis test and the literature on randomization tests (see \citet[Chapter 15.2]{lehmann/romano:2005}). In particular, Theorem \ref{thm:sizeControl} follows from showing that the $p$-value in \eqref{eq:pvalue} approximates the $p$-value of an underlying randomization test for $H_{0}$ in \eqref{eq:HT} that is computationally infeasible. Recall that randomization tests enjoy validity in finite samples under suitable conditions. This explains why Theorem \ref{thm:sizeControl} does not require the number of markets $n$, time periods $T$, or players $J$ to grow. 

The remainder of this section develops the connection between our hypothesis test and the underlying randomization test for $H_{0}$. It is organized as follows. Section \ref{sec:suff_stat} provides an alternative representation of the likelihood of the data under Assumption \ref{ass:M} and $H_{0}$. This result allows us to define a sufficient statistic of the data under these conditions, denoted by $U(X)$. Section \ref{sec:group} relates our MCMC algorithm to a transformation group of the data, $\mathbf{G}$, which does not change the value of the sufficient statistic $U(X)$. Section \ref{sec:exact} defines the underlying randomization test for $H_{0}$ based on the transformation group $\mathbf{G}$, and argues that it is both finite-sample valid and computationally infeasible. Finally, Section \ref{sec:approx_results} shows that our MCMC-based test in \eqref{eq:testDefn} can successfully approximate the underlying randomization test as the number of MCMC draws diverges.

\subsection{An alternative representation of the likelihood}\label{sec:suff_stat}

The next result provides an alternative representation of the likelihood of the data under Assumption \ref{ass:M} and $H_{0}$ in \eqref{eq:HT}.

\begin{lemma} \label{thm:Lik2}
Under Assumption \ref{ass:M} and $H_{0}$ in \eqref{eq:HT}, the likelihood of the data ${X}=({S},{A})$ evaluated at $\tilde{X}=(\tilde{S},\tilde{A})\in\mathcal{X}$ with $\tilde{S}=(\tilde{S}_{i,t}:i=1, \ldots ,n,t=1,\ldots ,T)\in\mathcal{S}^{nT}$ and $\tilde{A}=(\tilde{A}_{i,t}:i=1, \ldots ,n,t=1,\ldots ,T)\in\mathcal{A}^{nT}$ is 
\begin{equation}
P(X=\tilde{X})~=~P(A=\tilde{A}|S=\tilde{S})
~\times ~P(S=\tilde{S}), \label{eq:Lik2}
\end{equation}
where
\begin{align}
P(A =\tilde{A}|S=\tilde{S})&=
\prod_{(s,a,s^{\prime})\in\mathcal{S}\times\mathcal{A} \times\mathcal{S}}
 \left(\frac{\sigma(a|s)f(s^{\prime}|s,a)}{\sum_{ \tilde{a}\in\mathcal{A}}f(s^{\prime}|\tilde{a},s)\sigma(\tilde{a}|s)} \right)^{\sum_{i=1}^{n}\sum_{t=1}^{T-1}1\{\tilde{S}_{i,t}=s,\tilde{A}_{i,t}=a,\tilde{S}_{i,t+1}=s^{\prime}\}}
\notag\\&\qquad\times 
\prod_{(s,a)\in\mathcal{S}\times\mathcal{A}}
\sigma(a|s)^{\sum_{i=1}^{n}1\{\tilde{S}_{i,T}=s,\tilde{A}_{i,T}=a\}}
\label{eq:Lik2_a_s} 
\end{align}
and
\begin{align}
&P(S =\tilde{S})=
\left(
\prod_{i=1}^{n}m_{i}(\tilde{S}_{i,1})
\right)
\times
\prod_{(s,s^{\prime})\in\mathcal{S}\times\mathcal{S}}
\left(\sum_{{a}\in\mathcal{A} }f(s^{\prime}|{a},s)\sigma({a}|s)\right)^{ \sum_{i=1}^{n}\sum_{t=1}^{T-1}1\{\tilde{S}_{i,t}=s,\tilde{S}_{i,t+1}=s^{\prime}\}}
.\label{eq:Lik2_s}
\end{align}
\end{lemma}

From this result, we can deduce the following corollary.
\begin{corollary}\label{cor:Lik2}
Under Assumption \ref{ass:M} and $H_{0}$ in \eqref{eq:HT}, the sufficient statistic for $X =(S,A)$ is
\begin{equation}
 U(X)= \left(
 \begin{array}{l}
 \left(S_{i,1}:i=1\ldots ,n\right), \\
 \left(\sum_{i=1}^{n}\sum_{t=1}^{T-1}1\{S_{i,t}=s,A_{i,t}=a,S_{i,t+1}=s^{\prime}\}:(s,a,s^{\prime})\in\mathcal{S}\times\mathcal{A}\times\mathcal{S}\right),\\
 \left(\sum_{i=1}^{n}1\{S_{i,T}=s,A_{i,T}=a\}:(s,a)\in\mathcal{S}\times\mathcal{A}\right)
 \end{array}
 \right).\label{eq:SuffStatistics} 
\end{equation}
\end{corollary}

Corollary \ref{cor:Lik2} implies that, under Assumption \ref{ass:M} and $H_{0}$, a transformation of the data that maintains $U(X)$ will not change the value of the likelihood. This observation provides the basis of the underlying randomization test.

\subsection{A transformation group related to our MCMC algorithm}\label{sec:group}

In this section, we show that our MCMC algorithm is a transformation group of $\cal{X}$ that preserves the value of $U(X)$. See \citet[Appendix A.1]{lehmann/romano:2005} for the definition of the notion of a transformation group.

Our proposed MCMC algorithm can be understood as an iteration of transformations to the data $X$. In particular, $X^{(1)}=X$ is the identity transformation, $X^{(2)}$ follows from applying Steps 1-3 to $X$, $X^{(3)}$ follows from applying Steps 1-3 twice to $X^{(2)}$, and so forth. More formally, each iteration of our MCMC algorithm applies a transformation from a particular transformation group. To define this properly, we first require the following definition.

\begin{definition}[Collection of transformations ${\bf G}(I)$] \label{def:GI}
For any pair of markets $I= (I_1,I_2)\in\mathcal{I}$, let ${\bf G}(I)$ denote the set of all transformations of $\mathcal{X}$ onto itself such that, for any $g\in{\bf G}(I)$ and $(\breve{S},\breve{A})\in\mathcal{X}$, $(\tilde{S},\tilde{A})=g(\breve{S},\breve{A})$ satisfies $\tilde{S} \in R_{S}(I,\breve{S})$ and $\tilde{A} \in R_{A}(\tilde{S},(\breve{S},\breve{A}))$.
\end{definition}

Lemma \ref{lem:GI_is_group} in the appendix shows that ${\bf G}(I)$ is a transformation group. By Definition \ref{def:GI}, ${\bf G}(I)$ is the transformation group representation of Steps 2-3 of our MCMC algorithm. Given a randomly chosen pair of markets $I^{(k)}$ in step 1, steps 2-3 obtain the next element of the Markov chain $X^{(k)} = (S^{(k)},A^{(k)})$ by selecting a randomly chosen element of $\{g(X^{(k-1)})\colon{g}\in{\bf G}(I^{(k)})\}$. In this sense, Steps 2-3 of our MCMC algorithm are a specific way of choosing a particular transformation in ${\bf G}(I^{(k)})$.

By the description in the previous paragraph, our MCMC algorithm applies a randomly chosen transformation in ${\bf G}(I)$ for random pairs of markets $I$, and iteratively applies them to the data. These iterative transformations are related to the set that we define next.

\begin{definition}[Collection of transformations $\mathbf{G}$] \label{def:G}
$\mathbf{G}$ is the collection of all finitely many compositions of the elements in $\bigcup_{I\in\mathcal{I}}\mathbf{G }(I)$.
\end{definition}
See Example \ref{ex:G_explicit} for an illustration of $\mathbf{G}$. The next result states that $\mathbf{G}$ is a transformation group with desirable properties.

\begin{lemma}\label{lemma:Gdefn}
${\bf G}:\mathcal{X}\to\mathcal{X}$ is a transformation group of $\mathcal{X}$ such that, for any $g\in{\bf G}$ and $\tilde{X}\in\mathcal{X}$, 
$\tilde{X}$ and $g\tilde{X}$ have the same sufficient statistic in \eqref{eq:SuffStatistics}, i.e., $U(\tilde{X})= U(g\tilde{X})$.
\end{lemma}

The properties shown in Lemma \ref{lemma:Gdefn} imply that we can use ${\bf G}$ to define a valid randomization test. We do this in Section \ref{sec:exact}.

\subsection{The underlying randomization test}\label{sec:exact}

Following \citet[Chapter 15.2]{lehmann/romano:2005}, we can use the transformation group ${\bf G}$ to define the underlying randomization test. This test rejects $H_{0}$ in \eqref{eq:HT} whenever the significance level $\alpha$ is larger than or equal to the randomization $p$-value, which we denote by $\hat{p}$. That is,
\begin{equation}
 \phi(X)~\equiv~ 1\{\hat{p}\leq\alpha\},
 \label{eq:unfeasibleTest}
\end{equation}
where
\begin{equation}
 \hat{p}~\equiv ~\frac{1}{|{\bf G}|}\sum_{g\in{\bf G}}1\{\tau(gX)\geq \tau(X)\}.
 \label{eq:unfeasiblePvalue}
\end{equation}
By the arguments in \citet[Page 636]{lehmann/romano:2005}, the randomization test in \eqref{eq:unfeasibleTest} is finite-sample valid. We record this in the next result.

\begin{lemma} \label{lem:PvalueProperty} 
Under Assumption \ref{ass:M} and $H_{0}$ in \eqref{eq:HT}, for any $\alpha \in (0,1)$, the test in \eqref{eq:unfeasibleTest} satisfies
\begin{equation}
E[\phi(X)]~\leq~ \alpha.\label{eq:SizeControlBasedOnG}
\end{equation}
\end{lemma}

The finite-sample validity in Lemma \ref{lem:PvalueProperty} makes the randomization test in \eqref{eq:unfeasibleTest} an excellent candidate for testing $H_{0}$. Unfortunately, this randomization test is not computationally feasible in practice. This is because the test requires working with the transformation group $\mathbf{G}$, which is typically impossible to enumerate in practice. We illustrate this in Example \ref{ex:G_explicit} in the appendix, where we enumerate $\mathbf{G}$ in two simple examples with $n=2$ markets, $T=2$ time periods, and binary actions and states. Given the challenges presented even by these very simple cases, it is not hard to imagine that $\mathbf{G}$ is computationally impossible to enumerate in realistic data settings.

In the randomization testing literature, it is not uncommon to work with a huge transformation group $\mathbf{G}$. As \citet[page 636]{lehmann/romano:2005} explains, one can still implement a random version of the test in \eqref{eq:unfeasibleTest} by drawing randomly from $\mathbf{G}$ {\it in a uniform fashion}. This point is routinely exploited in standard settings to construct tests based on permutations or sign changes. However, to the best of our knowledge, there is no known feasible way of obtaining such random draws in the current context without fully enumerating $\mathbf{G}$. 

The previous paragraphs explain why the underlying randomization test in \eqref{eq:unfeasibleTest} is computationally infeasible and, thus, we cannot directly exploit its finite-sample validity. The main technical insight of our paper is that our hypothesis test in \eqref{eq:testDefn} is an approximate way of implementing the computationally infeasible underlying randomization test in \eqref{eq:unfeasibleTest}. In particular, the following section formally states that our MCMC-based $p$-value in \eqref{eq:pvalue} approximates the underlying $p$-value in \eqref{eq:unfeasiblePvalue} as the length of the MCMC diverges.

\subsection{An MCMC approximation to the underlying randomization test}\label{sec:approx_results}

Our main theoretical result is Theorem \ref{thm:sizeControl}, which shows that the test in \eqref{eq:testDefn} controls size as the number of MCMC draws $K$ diverges to infinity. The following lemma provides the fundamental ingredient to prove this result.

\begin{lemma} \label{lem:MCMCconv}
Conditional on $X$, 
$$
\sup_{t\in\mathbb{R}}~\Bigg\vert \frac{1}{K}\sum_{k=1}^{K}1\{ \tau(X^{(k)})\geq t\} ~-~\frac{1}{|\mathbf{G}|} \sum_{g\in\mathbf{G}}1\{\tau(gX)\geq t\}\Bigg\vert 
~~\overset{a.s.}{\to}0~~\text{ as }~K\to\infty. 
$$
\end{lemma}

Lemma \ref{lem:MCMCconv} shows that, as the number of MCMC draws diverges, the conditional distribution based on the MCMC algorithm converges to the conditional distribution of the computationally infeasible underlying randomization test described in Section \ref{sec:exact}.
It is worth noting that Lemma \ref{lem:MCMCconv} considers $K\to\infty$ while the complexity of the underlying randomization test, characterized by $|\mathbf{G}|$, stays constant. While we do not derive formal results, we expect that, as $|\mathbf{G}|$ increases, a larger number of MCMC draws $K$ is required to achieve a specific level of approximation. For related discussions on diagnosing convergence in MCMC algorithms, see \citet[Chapter 12]{robert/casella:2004}.

By applying Lemma \ref{lem:MCMCconv} with $t=\tau(X)$, we can deduce that the $p$-value in \eqref{eq:pvalue} approximates the $p$-value in \eqref{eq:unfeasiblePvalue} as the number of MCMC draws $K$ diverges. That is, conditional on $X$,
\begin{equation}
    \hat{p}_{K}~\overset{a.s.}{\to}~\hat{p}~~\text{ as }~K\to\infty. 
    \label{eq:Conv_pK}
\end{equation}
By combining this observation with the finite-sample validity of the underlying randomization test in \eqref{eq:unfeasibleTest} (Lemma \ref{lem:PvalueProperty}), it follows that our proposed MCMC-test becomes valid as the number of MCMC draws $K$ diverges. This argument provides the intuition behind Theorem \ref{thm:sizeControl}, and why it holds regardless of the number of markets $n$, time periods $ T $, and players $ J $. 

Our analysis in this paper focuses on the properties of our test under the null hypothesis. While analyzing our test's power properties is very desirable, we consider this to be a formidable task within our finite-sample setting. The rejection rate of our test depends on the specification of the dynamic discrete game (i.e., conditional choice probabilities, state transition probabilities, and marginal distributions specified under the alternative hypothesis). To the best of our knowledge, obtaining general results for all possible specifications of the dynamic discrete game in finite samples is impossible. An alternative to the finite-sample power analysis would be to consider asymptotic power results with a diverging number of markets $n$, time periods $T$, support points $|\mathcal{X}|$, or all three. Note that the current results in this paper are finite sample valid, and do not require such an asymptotic framework. The behavior of our test is expected to vary with the specific asymptotic framework under consideration. While such \textit{asymptotic} power results have been developed for some randomization tests in the literature (e.g., see \citet[Section 15.2.2]{lehmann/romano:2005}), these do not apply to approximate randomization tests like the one proposed in this paper. Developing this extension of our results seems out of the scope of our current contribution. As an admittedly imperfect substitute for a general power analysis, Section \ref{sec:MC} of our paper explores the power properties of our proposed test in several empirically relevant economic models. All of our simulation evidence suggests that our test has desirable power properties.
% (e.g.,  Tables \ref{tab:table2_MC} in Section \ref{sec:MCMC}).

\section{Empirical application}\label{sec:empiric}

In this section, we revisit the application in \cite{ryan:2012}, as studied in \citet[Section 5]{otsu/pesendorfer/takayashi:2016}. \cite{ryan:2012} considers a dynamic discrete game to study the welfare costs of the 1990 Amendments to the Clean Air Act on the U.S.\ Portland cement industry. He develops a dynamic oligopoly game based on \cite{ericson/pakes:1995}, and estimates it using the two-stage method developed by \cite{bajari/benkard/levin:2007}. This method's first stage is to estimate optimal entry, exit, and investment decisions as a function of production capacity, and it relies on the assumption that markets are homogeneous. Our hypothesis test can be used to investigate the validity of this assumption.

We use the same data as in \citet[Section 5]{otsu/pesendorfer/takayashi:2016}. For each year in 1980-1998 and 23 geographically separated U.S.\ markets, we observe the sum of the production capacities for all the firms in that market. Table \ref{tab:table_summary_stats} provides summary statistics of this aggregate production capacity before and after the 1990 Amendments, and Figure \ref{fig:histogram} provides the corresponding histogram. 

\begin{table}[ht]
 \centering
 \begin{tabular}{ccccc}\hline\hline
 Sample & Average & Std.\ dev. & Minimum & Maximum\\\hline
 1980-1990& 4,226.8 & 2,284.4 & 1,321.3 & 12,578.0\\
 1991-1998& 3,857.2 & 2,107.9 & 1,084.0& 9,564.8\\ \hline\hline
 \end{tabular}
\caption{\small Summary statistics for market capacity per year, measured in thousands of tons.}
  \label{tab:table_summary_stats}
\end{table}

\begin{figure}[ht]
\centering
\scalebox{0.8}{\includegraphics{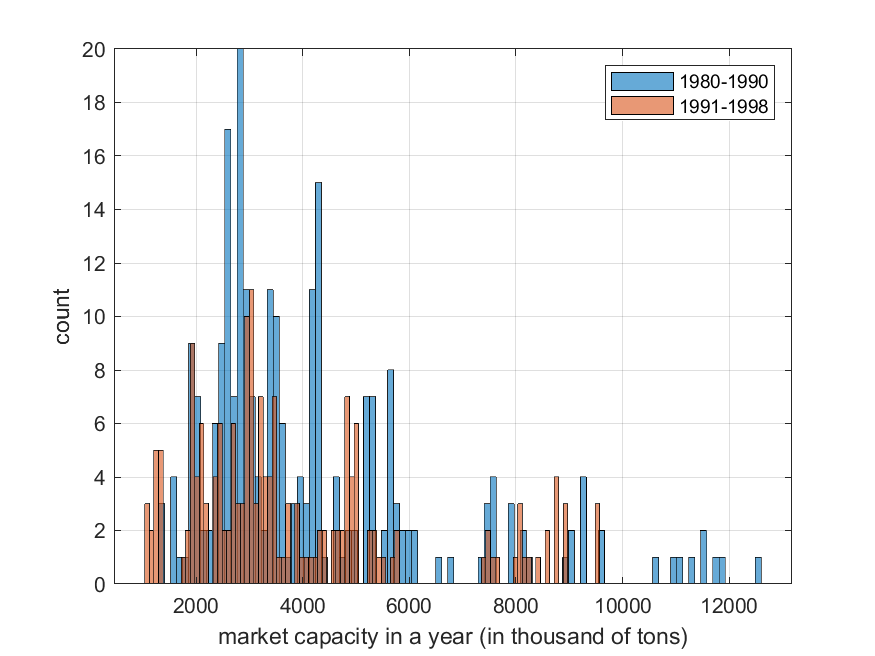}}
  \caption{\small   Histogram of market capacity per year, measured in thousands of tons.}
  \label{fig:histogram}
\end{figure}

These data represent the result of the firms' optimal entry, exit, and investment decisions in the dynamic game estimated by \cite{ryan:2012}. We follow \cite{otsu/pesendorfer/takayashi:2016} and discretize the market production capacity into 50 bins with equal intervals of 250 thousand tons each (0-250 thousand tons, 250-500 thousand tons, and so on). For each $i=1,\ldots ,n=23$ and year $t=1,\ldots ,19$, we use $A_{i,t}\in\mathcal{A}=\{ 1,\ldots ,50\} $ to denote the production capacity bin. The state variable in any market is the previous period's action, i.e.,
\begin{equation}
S_{i,t}~=~A_{i,t-1}, \label{eq:trivialstate transition probability_app}
\end{equation}
and so $S_{i,t}\in\mathcal{S}=\{ 1,\ldots ,50\} $. We note that \eqref{eq:trivialstate transition probability_app} implies that the state transition probabilities are homogeneous (given by $f_{i,t+1}(s^{\prime}|a,s)=1\{s^{\prime}=a\}$), and so $H_{0}$ in \eqref{eq:HT} is equivalent to the homogeneity of the conditional choice probabilities.

Following \cite{ryan:2012} and \cite{otsu/pesendorfer/takayashi:2016}, we allow the 1990 Amendments to affect the decision of the firms. We then test the homogeneity of the conditional choice probabilities for two subsets of data: before and after 1990. That is, we test the following hypotheses:
\begin{align}
H_{0}^{\text{before}}:\sigma_{i,t}(a|s)=\sigma(a|s)\text{ for }i=1,\ldots ,23,~t=1,\ldots ,10, \label{eq:test_app_1}\\
H_{0}^{\text{after}} :\sigma_{i,t}(a|s)=\sigma(a|s)\text{ for }i=1,\ldots ,23,~t=11,\ldots ,19.\label{eq:test_app_2}
\end{align}

We note that the before and after 1990 samples used to test the hypotheses in \eqref{eq:test_app_1} and \eqref{eq:test_app_2} have a relatively small number of time periods ($T=10$ and $T=9$ for \eqref{eq:test_app_1} and \eqref{eq:test_app_2}, respectively) and markets (in both cases, $n=23$). This represents an ideal scenario for our proposed test, as its validity does not rely on either one of these dimensions diverging. We also note that both of these dimensions are smaller than the support of the data, i.e., $|\mathcal{S}| =|\mathcal{A}|=50$.

Table \ref{tab:table_app_1} shows the results of applying our procedure to test the hypotheses in \eqref{eq:test_app_1} and \eqref{eq:test_app_2}. Following the literature, we use the test statistics in \eqref{eq:testStatOtsu}. As explained in Example \ref{ex:tau_examples}, these test statistics compare market-specific conditional choice probabilities with their pooled counterpart and are thus specifically designed to detect heterogeneity across markets. This objective seems appropriate for this empirical application, as \cite{ryan:2012}'s methodology relies on the homogeneity of the data before and after the 1990 Amendments.\footnote{It is relevant to note that these test statistics would be largely ineffective in detecting the presence of a structural break (such as the 1990 Amendments) if the break impacts equally all markets in the economy. This happens because the market-specific conditional choice probabilities would average over time and coincide with their pooled counterpart.} At a significance level of $ \alpha =5\%$, we do not reject the homogeneity of the conditional choice probabilities. Our tests were implemented with $K=50,000$, but our hypothesis testing decision (i.e., non-rejection for standard significance levels) remains invariant for any $K>10,000$. Using a standard desktop computer, our Julia package completed our test with $K=50,000$ in 4.8 minutes for the subsample before 1990 and 1.9 minutes for the subsample after 1990. The computation time increases linearly with $K$.

For contrast, Table \ref{tab:table_app_1} also shows the results of bootstrap-based test proposed by \cite{otsu/pesendorfer/takayashi:2016}. As opposed to our test, their methods reject the hypothesis of homogeneity of the conditional choice probabilities in the sample prior to 1990. Since both tests rely on the same test statistic, these differences are entirely driven by the differences in the $p$-values. Table \ref{tab:table_app_1} reveals that our test and the one proposed by \cite{otsu/pesendorfer/takayashi:2016} can produce different conclusions. This is also clearly shown in our Monte Carlo simulations in Section \ref{sec:MC}. It is natural to inquire which hypothesis test is correct about the homogeneity of the sample before 1990. Of course, this is impossible to determine with certainty in an empirical application. However, we consider that our Monte Carlo evidence in Section \ref{sec:MC_appbased} may shed light on this matter. These simulations suggest that in data settings similar to those in the empirical application (i.e., with $|\mathcal{S}|$ large relative to $n$ and $T$), our test controls size adequately while the test by \cite{otsu/pesendorfer/takayashi:2016} can suffer from overrejection.

% latex table generated in R 4.2.2 by xtable 1.8-4 package
% Thu Jan 26 10:39:59 2023
\begin{table}[ht]
% \setstretch{1.5}
\centering
\begin{tabular}{lrrrr}
  \hline
 & \multicolumn{2}{c}{Before 1990}&\multicolumn{2}{c}{After 1990} \\
 & \multicolumn{1}{l}{$\tau_1(X)$} & \multicolumn{1}{l}{$\tau_2(X)$} & \multicolumn{1}{l}{$\tau_1(X)$} & \multicolumn{1}{l}{$\tau_2(X)$}\\
 \hline 
 Test statistic & 199.48 & 159.43 &  89.44 &  90.58  \\ 
  OPT's \textit{p}-value &   0.01 &   0.01 &   0.14 &   0.07  \\ 
  Our \textit{p}-value &   0.21 &   0.12 &   0.73 &   0.68 \\ 
   \hline
\end{tabular}
\caption{\small   The test statistics $\tau_1(X)$ and $\tau_2(X)$ are defined in \eqref{eq:testStatOtsu}. Our test is computed according to \eqref{eq:testDefn} with $K=50,000$. OPT refers to \cite{otsu/pesendorfer/takayashi:2016}, whose results were copied from their Table 6.} 
\label{tab:table_app_1}
\end{table}

\section{Monte Carlo simulations}\label{sec:MC}

In this section, we explore the performance of our proposed test in Monte Carlo simulations. We consider three simulation designs. Our first design is based on the duopoly entry game in \citet[Section 7.1]{pesendorfer/schmidt-dengler:2008}. Our second design is based on our empirical application in Section \ref{sec:empiric}. Our third simulation is based on a dynamic single-agent human capital formation model in \cite{keane/wolpin:1997}. Our three designs offer a comprehensive description of the finite sample behavior of our proposed methodology, as they exhibit significant differences in crucial aspects of the dynamic discrete problem, including the number of support points, the number of players, and the nature of their strategic interaction.
% An important distinction between the two designs is the support of the state variable. In the first design, the state variable represents the combination of two firms' incumbency statuses and is thus naturally limited to having four points of support. In the second design, the state variable is constructed by discretizing a continuous aggregate production capacity. As a result, the support of the state variable is chosen by the researcher and may be quite large. 

\subsection{Simulations based on a duopoly entry game}\label{sec:MC_OPT}

This Monte Carlo design is also used by \citet[Section 4]{otsu/pesendorfer/takayashi:2016}, which follows from the duopoly entry game in \citet[Section 7.1]{pesendorfer/schmidt-dengler:2008}. 
% We refer to these papers for the details on the setup. 
The simulated data are generated by two oligopolistic firms deciding whether to enter or not into $n$ markets, and over $T$ time periods. This dynamic game has multiple equilibria, which we exploit to generate departures from the homogeneity assumption.

In each period $t=1,\dots,T$ and market $i=1,\dots,n$, there are four possible actions in this game: $A_{i,t}=1$ denotes that neither firm entered the market, $A_{i,t}=2$ denotes that only firm 2 enters, $A_{i,t}=3$ denotes that only firm 1 enters, and $A_{i,t}=4$ denotes that both firms enter. This implies that $\mathcal{A}=\{ 1,2,3,4\}$. As in the empirical application, the state variable in any market is the previous period's action (i.e., \eqref{eq:trivialstate transition probability_app} holds). This implies that $\mathcal{S}=\{ 1,2,3,4\}$ and that the state transition probabilities are homogeneous (and given by $f_{i,t+1}(s^{\prime}|a,s)~=~1\{s^{\prime}=a\}$). As a consequence, $H_{0}$ in \eqref{eq:HT} is equivalent to
\begin{equation}
H_{0}:\sigma_{i,t}(a|s)=\sigma(a|s).
\label{eq:HT2}
\end{equation}

The data produced by this game is a matrix $X=(S,A)\in\mathcal{X}$ constructed exactly as in \citet[Section 4]{otsu/pesendorfer/takayashi:2016}. We simulate data from a mixture of two data-generating processes: DGP 1 and DGP 2. They represent Markov perfect equilibria of the dynamic game, which differ in the conditional choice probabilities $\sigma(a|s)$. 
The matrices of conditional choice probabilities in DGP 1 and DGP 2 are
$$
\left(
\begin{array}{cccc}
0.19 & 0.30 & 0.12 & 0.18 \\
0.08 & 0.09 & 0.08 & 0.07 \\
0.53 & 0.48 & 0.46 & 0.53 \\
0.20 & 0.13 & 0.34 & 0.22 \\
\end{array}
\right)~~~\text{and}~~~
\left(
\begin{array}{cccc}
0.18 & 0.48 & 0.03 & 0.16 \\
0.20 & 0.21 & 0.14 & 0.23 \\
0.29 & 0.22 & 0.13 & 0.26 \\
0.33 & 0.09 & 0.70 & 0.35 \\
\end{array}
\right),
$$
respectively, where the index of the column indicates the state $s \in \mathcal{S}=\{ 1,2,3,4\}$, and the index of the row indicates the value of the action $a \in \mathcal{A}=\{ 1,2,3,4\}$.
Each market is sampled independently. Market $i=1,\dots,n$ behaves according to DGP 1 with probability $\lambda$ and DGP 2 with probability $1-\lambda$. Therefore, $\lambda\in[0,1] $ represents the average proportion of markets in DGP 1. Each market is initialized with a state equal to $1$, and we simulate the corresponding action according to the conditional choice probabilities. This, in turn, determines the next period's state according to \eqref{eq:trivialstate transition probability_app}, i.e., $S_{i,t+1}=A_{i,t}$. We then proceed iteratively until we have simulated $T+100$ periods for each market. The first 100 periods are discarded, producing a sample of $T$ periods for $n$ markets, which are then observed by the researcher.

For each simulated data, we implement our proposed test in \eqref{eq:testDefn} with $K=20,000$. We consider simulations with $n\in\{ 20,40,80,160\} $, $T\in\{ 5,10,20,40,80\} $, and $\lambda\in\{ 1,0,0.5,0.9\} $. As explained earlier, $\lambda $ represents the proportion of markets that are in DGP 1. If $\lambda =1$ or $\lambda =0$, all markets are sampled from the same distribution, and so the conditional choice probabilities are homogeneous across markets, i.e., $H_{0}$ holds. In turn, if $\lambda =0.5$ or $\lambda =0.9$, each data is composed of markets from both distributions, and so the conditional choice probabilities are not homogeneous across markets, i.e., $H_{0}$ fails. Note that $ \lambda =0.5$ generates data in which both distributions are equally represented, and so the heterogeneity in the conditional choice probabilities is more salient. On the other hand, the case with $\lambda =0.9$ produces data with a vast majority of markets in DGP 1, and so the heterogeneity in the conditional choice probabilities is harder to detect. For each simulation design, we compute rejection rates based on $2,000$ independently simulated datasets.

The results from the Monte Carlo simulation are shown in Table \ref{tab:table1_MC} for $\lambda\in\{0,1\}$ and Table \ref{tab:table2_MC} for $\lambda\in\{0.5,0.9\}$, respectively. For the sake of comparison, we also include the results from the test proposed by \cite{otsu/pesendorfer/takayashi:2016}. Their test compares the same test statistics in \eqref{eq:testStatOtsu} with critical values based on the bootstrap. As mentioned earlier, they show the validity of their test in an asymptotic framework with $T\to\infty $ and $n$ fixed. In contrast, our main result in Theorem \ref{thm:sizeControl} is valid for any finite $n$ and $T$.

Table \ref{tab:table1_MC} reveals that our test achieves relatively good size control for all values of time periods and market sizes under consideration. The table shows the result of running 80 hypothesis tests for different data configurations that satisfy $H_{0}$ (four market sizes, five time periods, two test statistics, and two distributions). Across these 80 numbers, our proposed test has an average rejection rate of 5.1\%, a standard deviation of 0.04\%, and a range of 4.1\% to 6.35\%. 
We note that Theorem \ref{thm:sizeControl} implies that our test should not produce overrejection as $K$ becomes large, but it is silent about the possibility of underrejection. Table \ref{tab:table1_MC} reveals that our test does not seem to suffer from underrejection in these simulations. For \cite{otsu/pesendorfer/takayashi:2016}'s test, the average rejection rate is also 5.1\%, but with a standard deviation is 2.2\% and a range of 0.6\% to 13.5\%. We note that the larger rejection rates occur in simulations with $T=5$, which is not unexpected for a test whose validity is proven in an asymptotic framework in which $T$ diverges.

% latex table generated in R 4.2.3 by xtable 1.8-4 package
% Tue Apr 25 21:24:14 2023
\begin{table}[ht]
\centering
\begin{tabular}{rr|rrrr|rrrr}
  \hline\hline \multicolumn{1}{c}{\multirow{3}[0]{*}{$n$}} & 
\multicolumn{1}{c|}{\multirow{3}[0]{*}{$T$}} & \multicolumn{4}{c|}{DGP 1 ($\lambda = 1$)} & 
\multicolumn
{4}{c}{DGP 2 ($\lambda = 0$)} \\
 & & \multicolumn{2}{c}{Our test} & \multicolumn{2}{c|}{OPT's test} & 
 \multicolumn{2}{c}{Our test} & \multicolumn{2}{c}{OPT's test} \\
 & & \multicolumn{1}{c}{$\tau_1(X)$} & \multicolumn{1}{c}{$\tau_2(X)$} 
& \multicolumn{1}{c}{$\tau_1(X)$} & \multicolumn{1}{c|}{$\tau_2(X)$} & 
\multicolumn{1}{c}{$\tau_1(X)$} & \multicolumn{1}{c}{$\tau_2(X)$} & 
  \multicolumn{1}{c}{$\tau_1(X)$} & \multicolumn{1}{c}{$\tau_2(X)$} \\ \hline 20 & 5 &  5.0 &  5.0 & 13.2 &  5.9 &  5.0 &  4.8 & 13.5 & 12.7 \\ 
  20 & 10 &  4.8 &  4.8 &  7.0 &  4.5 &  5.1 &  5.2 &  7.9 &  8.4 \\ 
  20 & 20 &  5.1 &  4.9 &  4.4 &  5.0 &  5.0 &  5.1 &  5.8 &  7.1 \\ 
  20 & 40 &  4.7 &  4.7 &  5.1 &  6.2 &  6.0 &  5.0 &  4.8 &  5.4 \\ 
  20 & 80 &  4.9 &  4.3 &  5.7 &  6.6 &  4.3 &  4.6 &  5.1 &  5.2 \\ 
   \hline
\hline
40 & 5 &  5.4 &  5.0 &  6.5 &  2.3 &  4.7 &  4.8 &  8.0 &  6.4 \\ 
  40 & 10 &  5.2 &  4.8 &  3.8 &  2.7 &  4.1 &  4.8 &  5.2 &  4.9 \\ 
  40 & 20 &  5.1 &  5.2 &  4.3 &  3.4 &  5.6 &  5.9 &  6.2 &  6.9 \\ 
  40 & 40 &  5.6 &  6.4 &  4.5 &  5.3 &  6.1 &  5.7 &  3.9 &  5.3 \\ 
  40 & 80 &  4.5 &  4.6 &  5.3 &  5.3 &  5.0 &  5.2 &  5.6 &  4.5 \\ 
   \hline
\hline
80 & 5 &  5.8 &  5.2 &  5.3 &  1.5 &  5.0 &  5.0 &  4.6 &  3.7 \\ 
  80 & 10 &  5.3 &  4.8 &  3.2 &  1.2 &  4.4 &  4.7 &  5.6 &  5.1 \\ 
  80 & 20 &  4.4 &  4.4 &  5.2 &  3.5 &  5.0 &  4.8 &  4.9 &  5.7 \\ 
  80 & 40 &  5.2 &  5.8 &  3.9 &  3.9 &  5.6 &  5.4 &  4.7 &  5.1 \\ 
  80 & 80 &  5.0 &  5.0 &  4.7 &  4.6 &  5.1 &  4.9 &  5.3 &  5.0 \\ 
   \hline
\hline
160 & 5 &  5.9 &  5.4 &  4.9 &  0.6 &  5.1 &  4.4 &  4.0 &  1.4 \\ 
  160 & 10 &  5.3 &  4.6 &  3.4 &  0.9 &  5.5 &  5.0 &  4.7 &  3.9 \\ 
  160 & 20 &  5.1 &  5.4 &  3.3 &  2.4 &  5.3 &  5.4 &  4.5 &  4.5 \\ 
  160 & 40 &  4.8 &  5.1 &  4.8 &  4.8 &  5.4 &  5.2 &  6.3 &  5.5 \\ 
  160 & 80 &  5.6 &  5.6 &  4.5 &  4.6 &  4.8 &  5.4 &  5.5 &  5.2 \\ 
   \hline
\hline
\end{tabular}
\caption{\small  Rejection rates of simulations under $H_{0}$ for 
                                                                 $\alpha = 5\%$ based on 2,000 i.i.d.\ simulation draws. 
                                                                 The results for $\lambda=1$ corresponds to data sampled from 
                                                                 DGP 1 and $\lambda=0$ corresponds to data sampled from DGP 2. 
                                                                 The test statistics $\tau_1(X)$ and $\tau_2(X)$ are defined in 
                                                                 \eqref{eq:testStatOtsu}. Our test is computed according to 
                                                                 \eqref{eq:testDefn} with $K=20,000$. OPT refers to 
                                                                 \cite{otsu/pesendorfer/takayashi:2016}, whose results were 
                                                                 copied from Tables 1 and 2 in their paper.} 
\label{tab:table1_MC}
\end{table}

% latex table generated in R 4.2.3 by xtable 1.8-4 package
% Tue Apr 25 21:24:23 2023
\begin{table}[ht]
\centering
\begin{tabular}{rr|rrrr|rrrr}
  \hline\hline \multicolumn{1}{c}{\multirow{3}[0]{*}{$n$}} & 
\multicolumn{1}{c|}{\multirow{3}[0]{*}{$T$}} & \multicolumn{4}{c|}{Mixture with $\lambda = 0.5$} & 
\multicolumn
{4}{c}{Mixture with $\lambda = 0.9$} \\
 & & \multicolumn{2}{c}{Our test} & \multicolumn{2}{c|}{OPT's test} & 
 \multicolumn{2}{c}{Our test} & \multicolumn{2}{c}{OPT's test} \\
 & & \multicolumn{1}{c}{$\tau_1(X)$} & \multicolumn{1}{c}{$\tau_2(X)$} 
& \multicolumn{1}{c}{$\tau_1(X)$} & \multicolumn{1}{c|}{$\tau_2(X)$} & 
\multicolumn{1}{c}{$\tau_1(X)$} & \multicolumn{1}{c}{$\tau_2(X)$} & 
  \multicolumn{1}{c}{$\tau_1(X)$} & \multicolumn{1}{c}{$\tau_2(X)$} \\ \hline 20 & 5 &   4.9 &   7.3 &  10.3 &   8.3 &   5.2 &   6.0 &  10.7 &   6.0 \\ 
  20 & 10 &   9.8 &  15.0 &   6.5 &   7.4 &   6.6 &   7.1 &   6.5 &   4.8 \\ 
  20 & 20 &  42.8 &  53.1 &  27.8 &  27.4 &  14.2 &  16.8 &  11.7 &  12.8 \\ 
  20 & 40 &  96.0 &  97.4 &  79.7 &  76.1 &  38.4 &  42.1 &  32.7 &  35.3 \\ 
  20 & 80 & 100   & 100   &  99.9 &  99.8 &  70.6 &  70.7 &  75.8 &  76.5 \\ 
   \hline
\hline
40 & 5 &   4.4 &   7.8 &   4.7 &   4.1 &   5.4 &   5.8 &   4.5 &   2.5 \\ 
  40 & 10 &  12.6 &  22.4 &   7.4 &   5.5 &   8.4 &   9.8 &   5.4 &   4.2 \\ 
  40 & 20 &  69.0 &  80.5 &  44.6 &  36.2 &  20.8 &  24.8 &  16.0 &  14.8 \\ 
  40 & 40 & 100   & 100   &  97.4 &  94.3 &  57.5 &  61.3 &  49.0 &  50.1 \\ 
  40 & 80 & 100   & 100   & 100   & 100   &  88.8 &  88.8 &  93.5 &  92.5 \\ 
   \hline
\hline
80 & 5 &   4.9 &  10.4 &   3.3 &   2.3 &   5.5 &   6.6 &   3.4 &   1.7 \\ 
  80 & 10 &  20.2 &  35.3 &  10.8 &   5.8 &   9.4 &  11.4 &   5.9 &   3.2 \\ 
  80 & 20 &  91.5 &  97.1 &  68.5 &  55.5 &  31.4 &  37.4 &  23.3 &  19.7 \\ 
  80 & 40 & 100   & 100   & 100   &  99.9 &  80.2 &  83.5 &  72.8 &  73.2 \\ 
  80 & 80 & 100   & 100   & 100   & 100   &  98.5 &  98.4 &  99.7 &  99.6 \\ 
   \hline
\hline
160 & 5 &   4.6 &  12.1 &   2.9 &   0.9 &   5.4 &   7.1 &   4.0 &   0.9 \\ 
  160 & 10 &  32.9 &  57.6 &  12.4 &   5.8 &  12.6 &  14.8 &   6.0 &   2.1 \\ 
  160 & 20 &  99.5 & 100   &  92.3 &  78.6 &  48.7 &  57.2 &  38.2 &  30.6 \\ 
  160 & 40 & 100   & 100   & 100   & 100   &  95.8 &  96.5 &  93.4 &  92.4 \\ 
  160 & 80 & 100   & 100   & 100   & 100   & 100   & 100   & 100   & 100   \\ 
   \hline
\hline
\end{tabular}
\caption{\small Rejection rates of simulations under $H_1$ for $\alpha = 5\%$ based 
          on 2,000 i.i.d.\ simulation draws. The results for $\lambda=0.5$ corresponds to data sampled from DGP 1 and DGP 2 
          in equal proportions, and the results for $\lambda=0.9$ corresponds to data sampled from DGP 1 and DGP 2 with 
          proportions 0.9 and 0.1, respectively. The test statistics $\tau_1(X)$ and $\tau_2(X)$ are defined in 
         \eqref{eq:testStatOtsu}. Our test is computed according to \eqref{eq:testDefn} with $K=20,000$. 
  OPT refers to \cite{otsu/pesendorfer/takayashi:2016}, whose results were copied from Tables 1 and 2 in their paper.} 
\label{tab:table2_MC}
\end{table}

Table \ref{tab:table2_MC} explores the performance of these tests for data configurations that do not satisfy $H_{0}$. We begin by explaining the results that are common to both hypothesis tests. First, recall that $\lambda$ denotes the proportion of the $n$ markets in the data that are in DGP 1. As $\lambda$ becomes closer to either zero or one, the data are increasingly coming from a single distribution, making the departure from the $H_{0}$ harder to detect. Second, as the number of markets $n$ grows, the inference methods gain more evidence of the presence of multiplicity, resulting in higher rejection rates. The same phenomenon occurs as the number of time periods $T$ increases. Third, we find that the hypothesis tests implemented with $\tau_2(X)$ tend to produce higher rejection rates than those implemented with $\tau_1(X)$. This finding appears consistent with the large $T$ optimality result in \citet[Proposition 2]{otsu/pesendorfer/takayashi:2016}. We now turn to compare rejection rates between the two tests. In most simulation designs, our test appears to have a higher or equal rejection rate than \cite{otsu/pesendorfer/takayashi:2016}'s test. The few exceptions occur in designs with $n=20$ and $T\in \{5,10\}$, which correspond to designs in which \cite{otsu/pesendorfer/takayashi:2016}'s test overrejects under $H_0$. This suggests that any power advantage of their test relative to ours may disappear when considering a size-corrected version.

\subsection{Simulations based on our empirical application}\label{sec:MC_appbased}

In this subsection, we explore the performance of our test in two DGPs related to the empirical application in Section \ref{sec:empiric}. The first data-generating process (DGP 1) satisfies $H_{0}$ in \eqref{eq:HT}, and the second one (DGP 2) does not. DGP 1 represents a discretized version of the pre-1990 Amendments data in the empirical application (i.e., $t\leq T_0 \equiv 9$), and is generated as follows. First, we discretize the data into $|\mathcal{S}|$ evenly spaced bins, which we denote by $\{\tilde{S}_{i,t}:i=1,\dots,n,~t=1,\dots,T\}$. As in the empirical application, the state variable in any market is the previous period's action (i.e., \eqref{eq:trivialstate transition probability_app} holds).  For each $i=1,\dots,n$, we simulate $S_{i,1}$ independently from the pre-1990 Amendments discretized distribution, i.e., for all $s \in \mathcal{S}=\{1,\dots,|\mathcal{S}|\}$,
\begin{align*}
    P(S_{i,1} = s)~=~\frac{\sum_{i=1}^{n}\sum_{t=1}^{T_0}1\{\tilde{S}_{i,t}=s\}}{n T_0 }.
\end{align*}
Second, for each $i=1,\dots,n$ and $t=1,\dots, T-1$, we simulate $A_{i,t}$ independently across markets according to the pre-1990 Amendments choice probabilities, i.e., for all $s,a \in \mathcal{S}=\{1,\dots,|\mathcal{S}|\}$,
\begin{align}
    P(A_{i,t}=a|S_{i,t} = s)~=~\frac{\sum_{i=1}^{n}\sum_{t=1}^{T_0}1\{\tilde{S}_{i,t}=s,\tilde{A}_{i,t}=a\}}
    {\sum_{i=1}^{n}\sum_{t=1}^{T_0}1\{\tilde{S}_{i,t}=s\}},\label{eq:markov_MC2}
\end{align}
where $S_{i,t}=A_{i,t-1}$ for all $i=1,\dots,n$ and $t=2,\dots, T$. Since the production capacity in each market and time period is drawn according to the market- and time-homogeneous conditional choice probabilities in \eqref{eq:markov_MC2}, DGP 1 satisfies $H_{0}$.

DGP 2 represents an economy in which half of the markets are negatively impacted by the 1990 Amendments, and is generated as follows. In the pre-Amendments periods (i.e., $t\leq T_0 \equiv 9$), DGP 2 coincides exactly with DGP 1. In the post-Amendments periods (i.e., $t> T_0$), the data is independently generated across markets in the following fashion. For markets with even index $i$ (i.e., $i=2,4,\dots,22$), the production level is distributed as in the pre-1990 Amendments periods (i.e., as in \eqref{eq:markov_MC2}). For markets with odd index $i$ (i.e., $i=1,3,\dots,23$), the production level is uniformly chosen to be weakly lower, i.e., for all $s,a \in \mathcal{S}=\{1,\dots,|\mathcal{S}|\}$,
\begin{align*}
    P(A_{i,t}=a|S_{i,t} = s)~=~\frac{1\{a\leq s\}}{s}.
\end{align*}
That is, markets with an even index $i$ are unaffected by the 1990 Amendments, while markets with an odd index $i$ are negatively affected. As in DGP 1, the state variable in any market is the previous period's action (i.e., \eqref{eq:trivialstate transition probability_app} holds). The structural change caused by the 1990 Amendments implies that DGP 2 does not satisfy $H_{0}$.

For each simulated data, we implement our proposed test in \eqref{eq:testDefn} with $K=20,000$. We consider simulations with $n=23 $, $T=19$, $T_0 =9$, and $|\mathcal{S}|\in \{5, 10, 15, \dots, 80\}$. The first three parameters are those in the empirical application, which has $|\mathcal{S}|=50$ bins. For each simulation design, we compute rejection rates based on $2{,}000$ independently simulated datasets.

The results from the Monte Carlo simulations are presented in Table \ref{tab:ryanmc}. We include results for our test and the one proposed by \cite{otsu/pesendorfer/takayashi:2016} with bootstrap-based $p$-values (see their Section 5 for details). We first describe results under DGP 1, i.e., when $H_{0}$ holds. Our test achieves good size control for all discretizations under consideration. Across the 10 hypothesis tests that satisfy $H_{0}$ (five discretizations and two test statistics), our test has an average rejection rate of 5.4\%, with a standard deviation of 0.3\%, and a range of 4.9\% to 5.9\%. These numbers also reveal that our test does not exhibit underrejection. On the other hand, \cite{otsu/pesendorfer/takayashi:2016}'s test suffers from overrejection, and this problem tends to exacerbate as $|\mathcal{S}|$ increases. For instance, when $|\mathcal{S}|$ is as in the empirical application (i.e., $|\mathcal{S}|=50$), their test has a rejection rate of 23.1\% for $\tau_1(X)$ and 26\% for $\tau_2(X)$, more than 4 times higher than the nominal size of $\alpha=5\%$. This issue may be explained by the fact that their validity result relies on $T \to \infty$, and these simulations only have $T=19$, which is smaller than $|\mathcal{S}| \in \{5, 10, 15, \dots, 80\}$. 

We now turn to the simulations under DGP 2, i.e., when $H_0$ fails. The results show that our test has nontrivial power for all values of $|\mathcal{S}| \in\{5, 10, 15, \dots, 80\}$. If particular, when $|\mathcal{S}|$ is as in the empirical application (i.e., $|\mathcal{S}|=50$), our test has a rejection rate of 36\% for $\tau_1(X)$ and 21.6\% for $\tau_2(X)$, which are considerably larger than the nominal size of $\alpha=5\%$. As one may expect, the power of our test tends to decrease with  $|\mathcal{S}|$. This is because the power of our test is based on permutations with common state values, which become increasingly rare as $|\mathcal{S}|$ grows. Also noteworthy is that, for $|\mathcal{S}|>20$, our test implemented with $\tau_1(X)$ has more power than when implemented with $\tau_2(X)$, which is an opposite pattern to that in the previous Monte Carlo simulations. Finally, we recognize that \cite{otsu/pesendorfer/takayashi:2016}'s test achieves much higher rejection rates, but these occur in the context of overrejection under the null hypothesis.

% latex table generated in R 4.3.1 by xtable 1.8-4 package
% Fri Sep 22 00:28:15 2023
\begin{table}[h]
\centering
\begin{tabular}{c|rrrr|rrrr}
  \hline\hline 
& \multicolumn{4}{c|}{DGP 1 (i.e., $H_0$ holds)} & \multicolumn{4}{c}{DGP 2 (i.e., $H_0$ fails)} \\
 & \multicolumn{2}{c}{Our test} & \multicolumn{2}{c|}{OPT's test} &
 \multicolumn{2}{c}{Our test} & \multicolumn{2}{c}{OPT's test} \\
\multicolumn{1}{c|}{$|\mathcal{S}|$}  & \multicolumn{1}{c}{$\tau_1(X)$} & \multicolumn{1}{c}{$\tau_2(X)$} & 
 \multicolumn{1}{c}{$\tau_1(X)$} & \multicolumn{1}{c|}{$\tau_2(X)$} & 
 \multicolumn{1}{c}{$\tau_1(X)$} & \multicolumn{1}{c}{$\tau_2(X)$} & 
 \multicolumn{1}{c}{$\tau_1(X)$} & \multicolumn{1}{c}{$\tau_2(X)$} \\
  \hline
\hline
5 &   4.6 &   4.2 &   8.2 &   7.4 &   5.7 &   6.8 &  15.8 &  27.0 \\ 
  10 &   5.0 &   5.1 &   9.6 &  13.7 &   7.6 &  13.8 &  36.6 &  66.5 \\ 
  15 &   5.4 &   4.9 &  17.0 &  16.2 &  26.5 &  40.1 &  91.4 &  97.9 \\ 
  20 &   5.3 &   5.9 &  15.8 &  15.4 &  36.0 &  50.7 &  94.6 &  98.2 \\ 
  25 &   5.7 &   5.3 &  18.0 &  22.0 &  27.9 &  22.6 &  94.4 &  99.5 \\ 
  30 &   5.7 &   5.7 &  20.5 &  20.9 &  31.8 &  27.7 &  97.0 &  99.0 \\ 
  35 &   5.2 &   5.1 &  23.0 &  27.5 &  34.0 &  24.8 &  99.5 & 100   \\ 
  40 &   5.6 &   5.6 &  20.8 &  24.5 &  44.0 &  31.9 &  99.8 & 100   \\ 
  45 &   4.7 &   4.2 &  40.1 &  44.5 &  34.5 &  18.1 & 100   & 100   \\ 
  50 &   5.6 &   5.8 &  21.4 &  22.9 &  36.0 &  21.6 &  99.8 & 100   \\ 
  55 &   5.1 &   4.6 &  30.1 &  38.0 &  32.9 &  22.9 & 100   & 100   \\ 
  60 &   5.1 &   5.7 &  26.8 &  36.0 &  35.0 &  22.4 & 100   & 100   \\ 
  65 &   5.6 &   5.3 &  56.0 &  71.4 &  29.8 &  20.8 & 100   & 100   \\ 
  70 &   5.0 &   5.6 &  40.0 &  47.4 &  29.8 &  16.4 & 100   & 100   \\ 
  75 &   6.2 &   5.7 &  60.2 &  71.9 &  36.9 &  25.9 & 100   & 100   \\ 
  80 &   5.4 &   4.9 &  47.4 &  65.2 &  30.8 &  18.1 & 100   & 100   \\ 
   \hline
\hline
\end{tabular}
\caption{Simulation results under DGP 1 (i.e., $H_0$ holds) and DGP 2 (i.e., $H_0$ fails). The test statistics $\tau_1(X)$ and $\tau_2(X)$ are defined in 
                                                       \eqref{eq:testStatOtsu}. Our test is computed according to \eqref{eq:testDefn} with $K=20,000$. 
                                                       OPT refers to \cite{otsu/pesendorfer/takayashi:2016}, whose results are generated via bootstrap 
                                                       following the description in their Section 5.} 
\label{tab:ryanmc}
\end{table}

\subsection{Simulations based on a human capital formation model}

We now consider Monte Carlo simulations based on the human capital formation single-agent model in \cite{keane/wolpin:1997}. In this model, individuals choose an occupation each period throughout their working life. A distinctive feature of this model is that each individual has permanent unobserved heterogeneity, representing ``innate talents'' that are unobserved by the econometrician. We use this aspect of the model to generate departures from the homogeneity assumption.

We simulate datasets with $n=100$ individuals choosing among occupations over $T=10$ time periods. The choice of $T=10$ is inspired by the data used for \cite{keane/wolpin:1997}'s structural estimation.
In each period $t=1,\dots,T$, individual $i=1,\dots,n$ chooses between home production, white-collar work, blue-collar work, schooling, and military work, which we denote as $A_{i,t}\in\mathcal{A}=\{1,2,3,4,5\}$, respectively. The state variable for individual $i$ in period $t$, denoted as $S_{i,t}$, encodes the experience vector in each occupation, i.e.,
$$S_{i,t} ~=~ \Big(\sum\nolimits_{s<t}1\{A_{i,s}=a\}:a \in \mathcal{A}\Big).$$
By definition, $S_{i,t+1}$ is a deterministic function of $A_{i,t}$ and $S_{i,t}$, and so the state transition probability is homogeneous. As a consequence, $H_{0}$ in \eqref{eq:HT} is equivalent to
\begin{equation}
H_{0}:\sigma_{i,t}(a|s)=\sigma(a|s).
\label{eq:HT3}
\end{equation}

We draw the data $X=(S,A)\in\mathcal{X}$ as a mixture of three agent types: type 1, type 2, and type 3, each of which is motivated by \cite{keane/wolpin:1997}. Type 1 represents a baseline individual with $\sigma(a|s) = \hat{\sigma}(a|s)$, where $\hat{\sigma}(a|s)$ denotes the empirical counterpart computed from the pooled NLSY79 sample across all individuals and time periods. In the pooled sample, $S_{i,t}$ has 417 support points. Type 2 represents an individual with innate talent for white-collar work, resulting in $\sigma(2|s)= \min\{1,\max\{5\hat\sigma(2|s),1/2\}\}$ and all other choice probabilities scaled appropriately, i.e., $\sigma(a|s)=\hat\sigma(a|s)/(1-\sigma(2|s))$ for $a\neq 2$. Finally, type 3 is the analog of type 2 but for blue-collar work, i.e., $\sigma(3|s)= \min\{1,\max\{5\hat\sigma(3|s),1/2\}\}$ and $\sigma(a|s)=\hat\sigma(a|s)/(1-\sigma(3|s))$ for $a\neq 3$.

We simulate independent datasets characterized by the parameter \(\lambda = (\lambda_1, \lambda_2)\). In each dataset, individuals are independently drawn, and are of type 1 with probability $1 - \lambda_1 - \lambda_2$, type 2 with probability $\lambda_{1}$, and type 3 with probability $\lambda_{2}$.  We simulate datasets from two DGPs. The first DGP uses $\lambda = (0,0)$, which produces a homogeneous sample composed of individuals of type 1, i.e., $H_0$ in \eqref{eq:HT3} holds. The second DGP uses $\lambda=(0.230, 0.556)$, which generates a sample with unobserved heterogeneity, i.e., $H_0$ in \eqref{eq:HT3} fails. The values in $\lambda=(0.230, 0.556)$ correspond to the empirical frequencies estimated in \citet[Table 9]{keane/wolpin:1997}.

The Monte Carlo results are shown in Table \ref{tab:kwmc_all}. We include results for our test and the one proposed by \cite{otsu/pesendorfer/takayashi:2016} with bootstrap-based $p$-values. Under $H_0$, our test exhibits relatively good size control, with perhaps a slight tendency to overreject. On the other hand, \cite{otsu/pesendorfer/takayashi:2016}'s test suffers from considerable overrejection, which may be explained by the fact that the current empirical setting with $T=10$ cannot be well represented by their asymptotic results as $T\to \infty$. Under $H_1$, our test exhibits small yet non-trivial power. As in our first design, our test implemented with $\tau_2(X)$ has more power than when implemented with $\tau_1(X)$. Given their results under $H_0$, we do not dwell on the performance of \cite{otsu/pesendorfer/takayashi:2016}'s test under $H_1$.

\begin{table}[ht]
\centering
\begin{tabular}{l|rr|rr}
  \hline\hline
  \multicolumn{1}{c|}{DGP}& \multicolumn{2}{c|}{Our test} & \multicolumn{2}{c}{OPT's test} \\
   & $\tau_1(X)$ & $\tau_2(X)$ & $\tau_1(X)$ & $\tau_2(X)$ \\
  \hline\hline
  $\lambda=(0,0)$, i.e., $H_0$ in \eqref{eq:HT3} holds & 6.50 & 5.90 & 46.25 & 78.70 \\
 % DGP 2 & 7.40 & 5.55 & 71.85 & 97.90 \\
 % DGP 3 & 7.15 & 6.15 & 44.10 & 32.40 \\
 $\lambda=(0.23, 0.56)$, i.e., $H_0$ in \eqref{eq:HT3}  fails & 9.05 & 9.60 & 71.85 & 74.10 \\
  \hline\hline
\end{tabular}
\caption{Simulation results for both DGPs. The test statistics $\tau_1(X)$ and $\tau_2(X)$ are defined in \eqref{eq:testStatOtsu}. Our test is computed according to \eqref{eq:testDefn} with $K=20,000$. OPT refers to \cite{otsu/pesendorfer/takayashi:2016}, whose results are generated via bootstrap following the description in their Section 5.}
\label{tab:kwmc_all}
\end{table}

\section{Conclusions}\label{sec:concl}

This paper proposes a hypothesis test for the ``homogeneity assumption'' in dynamic discrete games. Our test is implemented by an MCMC algorithm and does not rely on functional forms imposed by the researcher. We show that our test is valid as the (user-defined) number of MCMC draws diverges, regardless of the number of markets and time periods in the data. This result contrasts with that of available methods in the literature, which require the number of time periods to diverge. We establish our validity result by showing that our proposed test is an MCMC approximation to a computationally infeasible underlying randomization test, which is valid in finite samples.
Our Monte Carlo simulations reveal that our test has an excellent performance in finite samples, both in terms of size control and power.

\appendix
\setcounter{equation}{0}
\renewcommand{\theequation}{\Alph{section}-\arabic{equation}}
\begin{small}

\section{Appendix to Section \ref{sec:MCMC}}\label{app:MCMC}

To save on notation, this appendix treats ordered pairs such as $I = (I_1,I_2)$ and $I^{(k)} = (I_1,I_2)$ as the set $\{I_1,I_2\}$ whenever this does not generate confusion.

\subsection{Implementation of step 2 in our MCMC algorithm}\label{sec:euler_algorithm}

For any $k=2,\ldots ,K$, $S^{(k-1)} \in \mathcal{S}^{nT}$, and $I^{(k)} $ selected in step 1 of our MCMC algorithm, step 2 of our MCMC algorithm draws $S^{(k)}$ uniformly within $ R_{S}(I^{(k)},S^{(k-1)})$. To implement this step, we propose a modification of the Euler Algorithm. For a description of the Euler Algorithm, see \cite{kandel/yossi/unger/winkler:1996,besag/mondal:2013}.
%\footnote{We could use other algorithms to generate $\tilde{\xi}$, e.g.,\cite{wilson1996generating} and \cite{propp1998get}. All of our results in this section remain valid as long as, conditional on $\breve{\xi}$, $\tilde{\xi}$ is uniformly chosen within $R_{S0}(\breve{\xi})$.} 
We first describe the original Euler Algorithm in Algorithm \ref{alg:Euler} and then introduce our modification in Algorithm \ref{alg:S_k}. Throughout this section, we use $0$ to represent an auxiliary value for the state variable that does not belong to the observed values of the state variable, as $0 \not\in \mathcal{S} = \{1,2,\dots,|\mathcal{S}|\}$.

%%%%%%%%%%%%%%%%%%%%
\begin{algorithm}[Euler Algorithm]
\label{alg:Euler}
Given a sequence $\breve{\xi}\in(\mathcal{S}\cup \{0\})^{V}$ with $V \equiv dim(\breve{\xi}) \geq 2$, this algorithm randomly generates a sequence $\tilde{\xi}=(\tilde{\xi}_1,\ldots,\tilde{\xi}_V)$ in the following fashion:
\begin{itemize}%[leftmargin=1.75\parindent,align=left,labelwidth=\parindent,labelsep=4pt,itemsep=0ex,topsep=0ex] 
\item Step 1: 
For every $s,s^{\prime}\in\mathcal{S}\cup \{0\}$, set
\begin{equation}
  N^{(0)}(s,s^{\prime})~=~1\{(\breve{\xi}_{V},\breve{\xi}_{1})=(s,s^{\prime})\}+\sum_{v=1}^{V-1}1\{(\breve{\xi}_{v},\breve{\xi}_{v+1})=(s,s^{\prime})\}.  \label{eq:EA1}
\end{equation}
Also, set $\zeta_{1}=\breve{\xi}_{V}$ and $v=1$. 
Then, do the following:
\begin{enumerate}
\item[(a)] Given $\zeta_{v}$, generate $ \zeta_{v+1}$ according to the following distribution:
\begin{equation}
    P(\zeta_{v+1}=s\mid \zeta_{v})~=~
\frac{N^{(0)}(s,\zeta_{v})}{ \sum_{s^{\prime}\in\mathcal{S}\cup \{0\}}N^{(0)}(s^{\prime},\zeta_{v})},
\label{eq:EA2}
\end{equation}
where $\sum_{s\in\mathcal{S}\cup \{0\}}N^{(0)}(s,\zeta_{v})\geq 1$ is guaranteed by step 1.
% where, for every $s,s^{\prime}\in\mathcal{S}\cup \{0\}$, we define
% $$
% N^{(0)}(s,s^{\prime})~=~1\{(\breve{\xi}_{V},\breve{\xi}_{1})=(s,s^{\prime})\}+\sum_{v=1}^{V-1}1\{(\breve{\xi}_{v},\breve{\xi}_{v+1})=(s,s^{\prime})\}.
% $$
% (Note that, since $\sum_{s^{\prime\prime}\in\mathcal{S}\cup \{0\}}N^{(0)}(s^{\prime\prime},\zeta_{v})\geq 1$, the fraction is well-defined.)
\item[(b)]
If $(\zeta_2,\ldots,\zeta_{v+1})$ does not exhaust all the values in $\breve{\xi}$, then increase $v$ by one and go back to (a).
Otherwise, set $\bar{v}=v+1$ and go to step 2.
\end{enumerate}
\item Step 2: 
Set ${\tilde{\xi}}_{1}=\breve{\xi}_{1}$. 
Also, for every $s,s^{\prime}\in
\mathcal{S}\cup \{0\}$, set
\begin{equation}
    N^{(1)}(s,s^{\prime})~=~\sum_{v=1}^{V-1}1\{(\breve{\xi}_{v},\breve{\xi}_{v+1})=(s,s^{\prime})\}-1\{s^{\prime}=\zeta_{(\min \{v=2,\ldots ,\bar{v}:\ \zeta_{v}=s\}-1)},s\ne\breve{\xi}_{V}\}.
    \label{eq:EA3}
\end{equation}
For $v=2,\ldots ,V$, repeat the following:
\begin{itemize}
\item[(a)] 
Given $\tilde{\xi}_{v-1}$, generate $\tilde{\xi}_{v}$ according to 
\begin{align}
    P(\tilde{\xi}_{v}=s^{\prime}\mid \tilde{\xi}_{v-1}=s)~=~\left\{\begin{array}{cc}
       \dfrac{N^{(v-1)}(s,s^{\prime})}{\sum_{s^{\prime\prime}\in\mathcal{S}\cup\{0\}}N^{(v-1)}(s,s^{\prime\prime})}  & \text{ if }\sum_{s^{\prime\prime}\in\mathcal{S}\cup\{0\}}N^{(v-1)}(s,s^{\prime\prime})\geq 1, \\
       1\{s^{\prime}=\zeta_{(\min \{v=2,\ldots ,\bar{v}:\ \zeta_{v}=s\}-1)}\}  & \text{ otherwise}.
    \end{array}\right.\label{eq:EA4}
\end{align}
\item[(b)] 
For every $s,s^{\prime}\in
\mathcal{S}\cup \{0\}$, set $N^{(v)}(s,s^{\prime})~=~N^{(v-1)}(s,s^{\prime})-1\{N^{(v-1)}(s,s^{\prime})>0, (\tilde{\xi}_{v-1},\tilde{ \xi}_{v})=(s,s^{\prime})\}$ for every $s,s^{\prime}\in\mathcal{S}\cup \{0\}$.
\end{itemize}
\end{itemize}
\end{algorithm}

\begin{example}[Applying the Euler Algorithm in a simple case]
For the sake of illustration, we now apply the Euler Algorithm \ref{alg:Euler} to the sequence $\breve{\xi}=(1,1,2,1,2)$. This sequence was chosen because it generates an application of the Euler Algorithm that is both non-trivial and easy to explain. Note that $\breve{\xi}$ has $\dim(\breve{\xi})=V=5$ and only two values: 1 and 2. Also, $(0 \cup \mathcal{S}) = \{0,1,2\}$. As we now demonstrate, the Euler Algorithm applied to $\breve{\xi}$ produces a new sequence $\tilde{\xi}$ that is uniformly distributed in the set $\{ (1,1,2,1,2),(1,2,1,1,2)\}$. 

\begin{itemize}
    \item Step 1: We start by definiting $N^{(0)}(s,s')$ for $s,s' \in \{0,1,2\}$. Following \eqref{eq:EA1}, $N^{(0)}(s,s')=0$ if $s=0$ or $s'=0$, $N^{(0)}(1,1)=1$, $N^{(0)}(1,2)=2$, $N^{(0)}(2,1)=2$, and $N^{(0)}(2,2)=0$. Also, we set $\zeta_1=\breve{\xi}_{5} = 2$ and $v=1$. 

    Next, consider (a). This generates $\zeta_2$ according to \eqref{eq:EA2}, i.e., $$P(\zeta_2 = s|\zeta_1=2) ~=~ \frac{N^{(0)}(s,2)}{N^{(0)}(0,2)+N^{(0)}(1,2)+N^{(0)}(2,2)} ~=~ 1\{s=1\},$$ and so $\zeta_2=1$. 

    We next consider (b). In this case, note that $\zeta_2=1$ does not include all the values in $\breve{\xi}=(1,1,2,1,2)$, i.e., 1 and 2. Thus, we increase $v$ by one, i.e., $v=2$, and return to (a).

    In the new instance of (a), we generate $\zeta_3$ according to \eqref{eq:EA2}, i.e., $$P(\zeta_3 = s|\zeta_2=1) ~=~ \frac{N^{(0)}(s,1)}{N^{(0)}(0,1)+N^{(0)}(1,1)+N^{(0)}(2,1)} ~=~ \frac{ 1\{s=1\} + 2 \times 1\{s=2\}}{3}.$$
    That is, $\zeta_{3}=1$ with probability $1/3$ and $\zeta_{3}=2$ with probability $2/3$. 
    If $\zeta_{3}=1$ occurs, then $(\zeta_2,\zeta_3)=(1,1)$ does not include 1 and 2. Thus, we increase $v$ by one, i.e., $v=2$, and we return to (a). In turn, if $\zeta_{3}=2$ occurs, then $(\zeta_2,\zeta_3)=(1,2)$ includes 1 and 2, and we move on to step 3.
    
    By construction, the repetition of (a)-(b) will continue until $(\zeta_2,\ldots,\zeta_{v+1})$ includes all the values in $\breve{\xi}=(1,1,2,1,2)$, i.e.,  1 and 2. As a result of this, step 1 of the Euler Algorithm generates $\zeta=(2,1,1,\ldots,1,2)$, where the length of the 1's in the middle of $\zeta$ denotes the number of times (a)-(b) have been repeated. By definition, $\bar{v}$ is the length of $\zeta$. Finally, it is relevant for step 2 that $\min \{v=2,\ldots ,\bar{v}:~ \zeta_{v}=1\}=2$ and $\min \{v=2,\ldots ,\bar{v}:~ \zeta_{v}=2\}=\bar{v}$.

\item Step 2: Set ${\tilde{\xi}}_{1}=\breve{\xi}_{1}=1$. Also, we set $N^{(0)}(s,s')$ for $s,s' \in \{0,1,2\}$. Following \eqref{eq:EA3}, $N^{(1)}(1,s^{\prime})= 1$ and $N^{(1)}(2,s^{\prime})=2 \times 1\{s^{\prime}=1\}$. Then, we move on to (a).

In (a), we then generate $\tilde{\xi}_2$ according to \eqref{eq:EA4}, i.e.,
$$P(\tilde{\xi}_2 = s'|\tilde{\xi}_1=s) ~=~ \frac{N^{(1)}(1,s')}{N^{(1)}(1,0)+N^{(1)}(1,1)+N^{(1)}(1,2)} ~=~ \frac{1\{s' \in \{1,2\}\}}{2}.$$
That is, the algorithm chooses $\tilde{\xi}_2$ equal to $1$ or $2$ with equal probability. We now divide the argument into two cases.

\begin{itemize}
    \item Case 1: $\tilde{\xi}_2=1$. Then, (b) sets $N^{(2)}(1,s^{\prime})= 1\{s^{\prime}=2\}$ and $N^{(2)}(2,s^{\prime})=2\times 1\{s^{\prime}=1\}$. Then, (a) generates $\tilde{\xi}_3$ according to \eqref{eq:EA4}. Since $\tilde{\xi}_2=1$ and $\sum_{s'' \in \{0,1,2\}} N^{(2)}(1,s'')=1$, this gives
$$P(\tilde{\xi}_3 = s'|\tilde{\xi}_2=1) ~=~ \frac{N^{(2)}(1,s')}{N^{(2)}(1,0)+N^{(2)}(1,1)+N^{(2)}(1,2)} ~=~1\{s'=2\},$$
and so $\tilde{\xi}_3=2$. Then, (b) sets $N^{(3)}(1,s^{\prime})=0$ and $N^{(3)}(2,s^{\prime})=2 \times 1\{s^{\prime}=1\}$. Then, (a) generates $\tilde{\xi}_4$ according to \eqref{eq:EA4}. Since $\tilde{\xi}_3=2$ and $\sum_{s'' \in \{0,1,2\}} N^{(3)}(2,s'')=2$, this gives
$$P(\tilde{\xi}_4 = s'|\tilde{\xi}_3=2) ~=~ \frac{N^{(3)}(2,s')}{N^{(3)}(2,0)+N^{(3)}(2,1)+N^{(3)}(2,2)} ~=~1\{s'=1\},$$
and so $\tilde{\xi}_4=1$. Then, (b) sets $N^{(4)}(1,s^{\prime})=0$ and $N^{(4)}(2,s^{\prime})=1\{s^{\prime}=1\}$. Then, (a) generates $\tilde{\xi}_5$ according to \eqref{eq:EA4}. Since $\tilde{\xi}_4=1$ and $\sum_{s'' \in \{0,1,2\}} N^{(4)}(1,s'')=0$, this gives
$$P(\tilde{\xi}_5 = s'|\tilde{\xi}_4=1) ~=~ 1\{s'=\zeta_{(\min \{v=2,\ldots ,\bar{v}:~ \zeta_{v}=1\}-1)}\} ~=~1\{s'=2\},$$
and so $\tilde{\xi}_5=2$, where we have used that $\min \{v=2,\ldots ,\bar{v}:~ \zeta_{v}=1\}=2$. In conclusion, the resulting sequence is $\tilde{\xi}=(1,1,2,1,2)$.
\item Case 2: $\tilde{\xi}_2=2$. By repeating the arguments in case 1, it is not hard to see that the resulting sequence in case 2 is $\tilde{\xi} =(1,2,1,1,2)$.
\end{itemize}
\end{itemize}
Since Case 1 (i.e., $\tilde{\xi}_2=1$) and Case 2 (i.e., $\tilde{\xi}_2=2$) in step 2 are equally likely, we conclude that the Euler Algorithm generates the sequences $\tilde{\xi} =(1,1,2,1,2)$ and $\tilde{\xi} =(1,2,1,1,2)$ with equal probability, as desired.
\end{example}

Before we describe the central property of the Euler Algorithm, we first introduce a relevant definition.
%%%%%%%%%%%%%%%%%%%%
\begin{definition}
\label{def:R_S0_appendix}
For any $\breve{\xi}\in(\mathcal{S}\cup \{0\})^{V}$, let $R_{S0}(\breve{\xi})$ denote the set of all $\tilde{\xi}\in(\mathcal{S}\cup \{0\})^{V}$ that satisfy the following conditions:
\begin{enumerate}
\item[(a)] $\tilde{\xi}_{1}=\breve{\xi}_{1}$,
\item[(b)] $\sum_{v=1}^{V-1}1\{\tilde{\xi}_{v}=s,\tilde{\xi}_{v+1}=s^{\prime}\}=\sum_{v=1}^{V-1}1\{\breve{\xi}_{v}=s,\breve{\xi}_{v+1}=s^{\prime}\}$ for all $s,s^{\prime}\in\mathcal{S}\cup \{0\}$.
%\item[(c)] $\tilde{\xi}_{V}=\breve{\xi}_{V}$.
\end{enumerate}
\end{definition}
%%%%%%%%%%%%%%%%%%%%
Note that $\breve{\xi}\in R_{S0}(\breve{\xi})$, and so $R_{S0}(\breve{\xi}) \neq \emptyset$. Next, we give the main property of the Euler Algorithm.
%%%%%%%%%%%%%%%%%%%%
\begin{lemma} \label{lemma:kandal_uniformity}
For any $\breve{\xi}\in(\mathcal{S}\cup \{0\})^{V}$, the outcome of the Euler Algorithm given $\breve{\xi}$ (i.e., Algorithm \ref{alg:Euler}) is uniformly distributed over $R_{S0}(\breve{\xi})$ conditional on $\breve{\xi}$. 
\end{lemma}
%%%%%%%%%%%%%%%%%%%%
\begin{proof}
See \citet[Theorem 2]{kandel/yossi/unger/winkler:1996}.
\end{proof}
%%%%%%%%%%%%%%%%%%%%

We now introduce our modification of the Euler Algorithm to construct $S^{(k)}$ for any $k=2,\ldots ,K$. 

%%%%%%%%%%%%%%%%%%%%
\begin{algorithm}[Generation of $S^{(k)}$]
\label{alg:S_k}
For any $k=2,\ldots ,K$ and given $(X^{(1)},\ldots ,X^{(k-1)},I^{(k)})$, $S^{(k)}$ is randomly generated as follows:
\begin{itemize}%[leftmargin=1.3\parindent,align=left,labelwidth=\parindent,labelsep=4pt,itemsep=0.5ex,topsep=0.5ex] 
\item Case 1: $I_1^{(k)} \neq I_2^{(k)}$. 
\begin{itemize}%[leftmargin=0.9\parindent,align=left,labelwidth=\parindent,labelsep=4pt,itemsep=0.5ex,topsep=0.5ex] 
\item Step 1: Set $\xi ^{(k-1)}=(S_{I_1^{(k)},1}^{(k-1)},\ldots ,S_{I_1^{(k)},T}^{(k-1)},0,S_{I_2^{(k)},1}^{(k-1)},\ldots ,S_{I_2^{(k)},T}^{(k-1)},0).$
\item Step 2: Generate $\xi ^{(k)}$ as follows:
\begin{itemize}
\item[(a)] Generate a random draw of ${\xi}$ using the Euler Algorithm given $\xi ^{(k-1)}$.
\item[(b)] If ${\xi}_{T+1}=0$, set $\xi ^{(k)}=\xi$ and go to step 3. Otherwise, return to (a).
\end{itemize}
\item Step 3: Given $\xi ^{(k)}$, generate $S^{(k)}$ as follows:
\begin{itemize}
\item[(a)] For every $i\notin I^{(k)}$, generate $(S_{i,1}^{(k)},\ldots,S_{i,T}^{(k)})$ using the Euler Algorithm given $(S_{i,1}^{(k-1)},\ldots,S_{i,T}^{(k-1)})$.
\item[(b)]  $(S_{I_1^{(k)},1}^{(k)},\ldots,S_{I_1^{(k)},T}^{(k)})=(\xi ^{(k)}_{1},\dots,\xi ^{(k)}_{T})$.
\item[(c)] $(S_{I_2^{(k)},1}^{(k)},\ldots,S_{I_2^{(k)},T}^{(k)})=(\xi ^{(k)}_{T+2},\dots,\xi ^{(k)}_{2T+1})$.
\end{itemize}
\end{itemize}
\item Case 2: $I_1^{(k)} = I_2^{(k)}$. For every $i=1,\ldots,n$, generate $(S_{i,1}^{(k)},\ldots,S_{i,T}^{(k)})$ using the Euler Algorithm given $(S_{i,1}^{(k-1)},\ldots,S_{i,T}^{(k-1)})$.\hfill $\blacksquare$
\end{itemize}
\end{algorithm}
%%%%%%%%%%%%%%%%%%%%
\smallskip

Lemma \ref{lem:EulerWorks} shows that $S^{(k)}$ generated by Algorithm \ref{alg:S_k} has the desired properties.

%%%%%%%%%%%%%%%%%%%%
\begin{lemma}\label{lem:EulerWorks}
For any $k=2,\ldots ,K$, $S^{(k)}$ generated by Algorithm \ref{alg:S_k} satisfies the requirements of step 2 of our MCMC algorithm, i.e., \eqref{eq:dist_1} holds.
\end{lemma}
%%%%%%%%%%%%%%%%%%%%
\begin{proof}
We fix $k=2,\ldots ,K$, $(X^{(1)},\ldots ,X^{(k-1)})$, and a generic $\breve{S}\in\mathcal{S}^{nT}$ arbitrarily throughout this proof. We divide the proof in two cases.

Case 1: $I_{1}^{(k)}\not=I_{2}^{(k)}$. For $S^{(k-1)}$ and $S^{(k)}$ determined by $X^{(k-1)} = (S^{(k-1)},A^{(k-1)})$ and $X^{(k)} = (S^{(k)},A^{(k)})$, and for a generic $\breve{S} \in \mathcal{S}^{nT}$, we set
\begin{align*}
{\xi}^{(k-1)} &~=~({S}_{I_{1}^{(k-1)},1},\ldots ,{S} _{I_{1}^{(k-1)},T},0,{S}_{I_{2}^{(k-1)},1},\ldots ,{S} _{I_{2}^{(k-1)},T},0), \\
{\xi}^{(k)} &~=~({S}_{I_{1}^{(k)},1},\ldots ,{S} _{I_{1}^{(k)},T},0,{S}_{I_{2}^{(k)},1},\ldots ,{S} _{I_{2}^{(k)},T},0), \\
\breve{\xi} &~=~(\breve{S}_{I_{1}^{(k)},1},\ldots ,\breve{S} _{I_{1}^{(k)},T},0,\breve{S}_{I_{2}^{(k)},1},\ldots ,\breve{S} _{I_{2}^{(k)},T},0).
\end{align*}
Step 3 of Algorithm \ref{alg:S_k} implies
\begin{align}
& P(S^{(k)} =\breve{S}\mid I^{(k)},X^{(1)},\ldots ,X^{(k-1)})~=~ \notag \\
%& \overset{(1)}{=}\left\{ 
%\begin{array}{c}
%P(\xi ^{(k)}=\breve{\xi}\mid I^{(k)},X^{(1)},\ldots ,X^{(k-1)}) \\
%\times \prod_{i\not\in I^{(k)}}P((S_{i,1}^{(k)},\ldots ,S_{i,T}^{(k)})=( \breve{S}_{i,1},\ldots ,\breve{S}_{i,T})\mid I^{(k)},X^{(1)},\ldots ,X^{(k-1)})
%\end{array}
%\right\}\notag\\
&
P(\xi ^{(k)}=\breve{\xi}\mid {\xi}^{(k-1)}) \times 
\prod_{i\not\in I^{(k)}}P((S_{i,1}^{(k)},\ldots ,S_{i,T}^{(k)})=( \breve{S}_{i,1},\ldots ,\breve{S}_{i,T})\mid S_{i,1}^{(k-1)},\ldots ,S_{i,T}^{(k-1)}),\label{eq:step2_6}
\end{align}
Lemma \ref{lemma:kandal_uniformity} implies that 
\begin{equation}
P((S_{i,1}^{(k)},\ldots ,S_{i,T}^{(k)})=(\breve{S}_{i,1},\ldots ,\breve{S} _{i,T})\mid S_{i,1}^{(k-1)},\ldots ,S_{i,T}^{(k-1)})=\frac{1\{(\breve{S} _{i,1},\ldots ,\breve{S}_{i,T})\in R_{S0}(S_{i,1}^{(k-1)},\ldots ,S_{i,T}^{(k-1)})\}}{|R_{S0}(S_{i,1}^{(k-1)},\ldots ,S_{i,T}^{(k-1)})|} \label{eq:step2_7}
\end{equation}
for every $i\not\in I^{(k)}$.
In turn, Lemma \ref{lemma:kandal_uniformity2} implies that
\begin{equation}
P(\xi ^{(k)}=\breve{\xi}\mid \xi^{(k-1)})~=~\frac{1\{ \breve{\xi}\in R_{S0}(\xi ^{(k-1)}):\breve{\xi}_{T+1}=0\}}{|\{\tilde{\xi}\in R_{S0}(\xi ^{(k-1)}):\tilde{\xi}_{T+1}=0\}|}. \label{eq:step2_8}
\end{equation}
By combining \eqref{eq:step2_6}, \eqref{eq:step2_7}, and \eqref{eq:step2_8},
\begin{align}
&P(S^{(k)}=\breve{S}\mid I^{(k)},X^{(1)},\ldots ,X^{(k-1)})~=~\notag \\
&\frac{1\{\breve{\xi}\in R_{S0}(\xi ^{(k-1)}):\breve{\xi}_{T+1}=0\}\times \prod_{i\not\in I^{(k)}}1\{(\breve{S}_{i,1},\ldots ,\breve{S}_{i,T})\in R_{S0}(S_{i,1}^{(k-1)},\ldots ,S_{i,T}^{(k-1)})\}}{|\{\tilde{\xi}\in R_{S0}(\xi ^{(k-1)}):\tilde{\xi}_{T+1}=0\}|\times \prod_{i\not\in I^{(k)}}|R_{S0}(S_{i,1}^{(k-1)},\ldots ,S_{i,T}^{(k-1)})|}. \label{eq:step2_9}
\end{align}
To complete the proof, it suffices to show that the right-hand side of \eqref{eq:step2_9} is equal to the right-hand side of \eqref{eq:dist_1}. To this end, it suffices to show that
\begin{align}
&1\{\breve{\xi} \in R_{S0}(\xi ^{(k-1)}):\breve{\xi}_{T+1}=0\}\times \prod_{i\not\in I^{(k)}}1\{(\breve{S}_{i,1},\ldots ,\breve{S}_{i,T})\in R_{S0}(S_{i,1}^{(k-1)},\ldots ,S_{i,T}^{(k-1)})\}
\notag\\
&~=~1\{ \breve{S}\in R_{S}(I^{(k)},S^{(k-1)})\} \label{eq:step2_10} 
\end{align}
and
\begin{align}
|\{\tilde{\xi} &\in R_{S0}(\xi ^{(k-1)}):\tilde{\xi}_{T+1}=0\}|\times \prod_{i\not\in I^{(k)}}|R_{S0}(S_{i,1}^{(k-1)},\ldots ,S_{i,T}^{(k-1)})|~=~\vert R_{S}(I^{(k)},S^{(k-1)})\vert .\label{eq:step2_11}
\end{align}
To show \eqref{eq:step2_10}, consider the following derivation:
\begin{align*}
& 1\{\breve{\xi} \in R_{S0}(\xi ^{(k-1)}):\breve{\xi}_{T+1}=0\}\times \prod_{i\not\in I^{(k)}}1\{(\breve{S}_{i,1},\ldots ,\breve{S}_{i,T})\in R_{S0}(S_{i,1}^{(k-1)},\ldots ,S_{i,T}^{(k-1)})\} \\
& ~=~\left\{
\begin{array}{c}
1\{(\breve{S}_{I_{1}^{(k)},1},\ldots ,\breve{S}_{I_{1}^{(k)},T},0,\breve{S} _{I_{2}^{(k)},1},\ldots ,\breve{S}_{I_{2}^{(k)},T},0)\in R_{S0}(\xi ^{(k-1)})\} \\
\times \prod_{i\not\in I^{(k)}}1\{(\breve{S}_{i,1},\ldots ,\breve{S} _{i,T})\in R_{S0}(S_{i,1}^{(k-1)},\ldots ,S_{i,T}^{(k-1)})\}
\end{array}
\right\}  \\
& ~\overset{(1)}{=}~1\left\{ 
\begin{array}{c}
\breve{S}_{i,1}={S}_{i,1}^{(k-1)}~\text{for all}~ i=1,\ldots ,n, \\
\sum_{i\in I^{(k)}}\sum_{t=1}^{T-1}1\{\breve{S}_{i,t}=s,\breve{S} _{i,t+1}=s^{\prime }\}=\sum_{i\in I^{(k)}}\sum_{t=1}^{T-1}1\{S_{i,t}^{(k-1)}=s,S_{i,t+1}^{(k-1)}=s^{\prime }\}~\forall s,s^{\prime }\in \mathcal{S}, \\
\sum_{i\in I^{(k)}}1\{\breve{S}_{i,T}=s\}=\sum_{i\in I^{(k)}}1\{S_{i,T}^{(k-1)}=s\}~\text{for all}~ s\in \mathcal{S}, \\
\sum_{t=1}^{T-1}1\{\breve{S}_{i,t}=s,\breve{S}_{i,t+1}=s^{\prime }\}=\sum_{t=1}^{T-1}1\{S_{i,t}^{(k-1)}=s,S_{i,t+1}^{(k-1)}=s^{\prime }\}~\text{for all}~ s,s^{\prime }\in \mathcal{S},~i\not\in I^{(k)}
\end{array}
\right\}  \\
& ~\overset{(2)}{=}~1\left\{ 
\begin{array}{c}
\breve{S}_{i,1}={S}_{i,1}^{(k-1)}~\text{for all}~ i=1,\ldots ,n, \\
\sum_{i\in I^{(k)}}\sum_{t=1}^{T-1}1\{\breve{S}_{i,t}=s,\breve{S} _{i,t+1}=s^{\prime }\}=\sum_{i\in I^{(k)}}\sum_{t=1}^{T-1}1\{S_{i,t}^{(k-1)}=s,S_{i,t+1}^{(k-1)}=s^{\prime }\} ~\forall s,s^{\prime }\in \mathcal{S}, \\
\sum_{t=1}^{T-1}1\{\breve{S}_{i,t}=s,\breve{S}_{i,t+1}=s^{\prime }\}=\sum_{t=1}^{T-1}1\{S_{i,t}^{(k-1)}=s,S_{i,t+1}^{(k-1)}=s^{\prime }\} ~\text{for all}~ s,s^{\prime }\in \mathcal{S},~i\not\in I^{(k)}
\end{array}
\right\} \\
& ~\overset{(3)}{=}~1\{\breve{S}\in R_{S}(I^{(k)},S^{(k-1)})\},
\end{align*}
as desired, where (1) holds by $I_{1}^{(k)}\not=I_{2}^{(k)}$ and Definition \ref{def:R_S0_appendix}, (2) by Lemma \ref {lemma:RS_terminal}, and (3) by Definition \ref{def:RS}. To show \eqref{eq:step2_11}, consider the following argument.
\begin{align*}
1 &~=~\sum_{\breve{S}\in \mathcal{S}}P(S^{(k)}=\breve{S}\mid I^{(k)},X^{(1)},\ldots ,X^{(k-1)}) \\
&~\overset{(1)}{=}~\frac{\sum_{\breve{S}\in \mathcal{S}}1\{ \breve{S}\in R_{S}(I^{(k)},S^{(k-1)})\} }{|\{\tilde{\xi}\in R_{S0}(\xi ^{(k-1)}): \tilde{\xi}_{T+1}=0\}|\times \prod_{i\not\in I^{(k)}}|R_{S0}(S_{i,1}^{(k-1)},\ldots ,S_{i,T}^{(k-1)})|} \\
&~=~\frac{\vert R_{S}(I^{(k)},S^{(k-1)})\vert }{|\{\tilde{\xi}\in R_{S0}(\xi ^{(k-1)}):\tilde{\xi}_{T+1}=0\}|\times \prod_{i\not\in I^{(k)}}|R_{S0}(S_{i,1}^{(k-1)},\ldots ,S_{i,T}^{(k-1)})|},
\end{align*}
where (1) holds by combining \eqref{eq:step2_9} and \eqref{eq:step2_10}. From here, \eqref{eq:step2_11} follows.

Case 2: $I_{1}^{(k)}=I_{2}^{(k)}$. Algorithm \ref{alg:S_k} implies
\begin{equation}
P(S^{(k)}=\breve{S}\mid I^{(k)},X^{(1)},\ldots ,X^{(k-1)})=\prod_{i=1}^{n}P((S_{i,1}^{(k)},\ldots ,S_{i,T}^{(k)})=(\breve{S} _{i,1},\ldots ,\breve{S}_{i,T})\mid S_{i,1}^{(k-1)},\ldots ,S_{i,T}^{(k-1)}).
\label{eq:step2_1}
\end{equation}
Lemma \ref{lemma:kandal_uniformity} implies that for every $i=1,\ldots ,n$,
\begin{equation}
P((S_{i,1}^{(k)},\ldots ,S_{i,T}^{(k)})~=~(\breve{S}_{i,1},\ldots ,\breve{S} _{i,T})\mid S_{i,1}^{(k-1)},\ldots ,S_{i,T}^{(k-1)})=\frac{1\{(\breve{S} _{i,1},\ldots ,\breve{S}_{i,T})\in R_{S0}(S_{i,1}^{(k-1)},\ldots ,S_{i,T}^{(k-1)})\}}{|R_{S0}(S_{i,1}^{(k-1)},\ldots ,S_{i,T}^{(k-1)})|}.
\label{eq:step2_2}
\end{equation}
By combining \eqref{eq:step2_1} and \eqref{eq:step2_2}
\begin{align}
P(S^{(k)}=\breve{S}\mid I^{(k)},X^{(1)},\ldots ,X^{(k-1)})&~=~\prod_{i=1}^{n} \frac{1\{(\breve{S}_{i,1},\ldots ,\breve{S}_{i,T})\in R_{S0}(S_{i,1}^{(k-1)},\ldots ,S_{i,T}^{(k-1)})\}}{|R_{S0}(S_{i,1}^{(k-1)}, \ldots ,S_{i,T}^{(k-1)})|}\notag \\
&~=~\frac{\prod_{i=1}^{n}1\{(\breve{S}_{i,1},\ldots ,\breve{S}_{i,T})\in R_{S0}(S_{i,1}^{(k-1)},\ldots ,S_{i,T}^{(k-1)})\}}{ \prod_{i=1}^{n}|R_{S0}(S_{i,1}^{(k-1)},\ldots ,S_{i,T}^{(k-1)})|}. \label{eq:step2_3}
\end{align}
To complete the proof, it suffices to show that the right-hand side of \eqref{eq:step2_3} is equal to the right-hand side of \eqref{eq:dist_1}. To this end, it suffices to show that
\begin{align}
\prod_{i=1}^{n}1\{(\breve{S}_{i,1},\ldots ,\breve{S}_{i,T}) \in R_{S0}(S_{i,1}^{(k-1)},\ldots ,S_{i,T}^{(k-1)})\}&=1\{ \breve{S}\in R_{S}(I^{(k)},S^{(k-1)})\} \label{eq:step2_4} \\
\prod_{i=1}^{n}|R_{S0}(S_{i,1}^{(k-1)},\ldots ,S_{i,T}^{(k-1)})| &=\vert R_{S}(I^{(k)},S^{(k-1)})\vert .\label{eq:step2_5} 
\end{align}
To show \eqref{eq:step2_4}, consider the following derivation:
\begin{align*}
&\prod_{i=1}^{n}1\{(\breve{S}_{i,1},\ldots ,\breve{S}_{i,T}) \in R_{S0}(S_{i,1}^{(k-1)},\ldots ,S_{i,T}^{(k-1)})\} \\
&\overset{(1)}{=} 1\left\{
\begin{array}{c}
\breve{S}_{i,1}={S}_{i,1}^{(k-1)}\text{ for all }i =1,\dots,n,\\
\sum_{t=1}^{T-1}1\{\breve{S}_{i,t}=s,\breve{S}_{i,t+1}=s^{\prime }\}=\sum_{t=1}^{T-1}1\{S_{i,t}^{(k-1)}=s,S_{i,t+1}^{(k-1)}=s^{\prime }\}\text{ for all }s,s^{\prime }\in \mathcal{S}~\text{and} ~i =1,\dots,n
\end{array}
\right\}  \\
&\overset{(2)}{=}1\{ \breve{S}\in R_{S}(I^{(k)},S^{(k-1)})\} ,
\end{align*}
as desired, where (1) holds by $I_{1}^{(k)}=I_{2}^{(k)}$ and Definition \ref{def:R_S0_appendix}, and (2) by Definition \ref{def:RS}. Finally, \eqref{eq:step2_5} can be shown by an argument that is analogous to the one used to prove \eqref{eq:step2_11}. We omit this for brevity.
%To show \eqref{eq:step2_5}, consider the following argument.
%\begin{align*}
%1 =\sum_{\breve{S}\in \mathcal{S}}P(S^{(k)}=\breve{S}\mid I^{(k)},X^{(1)},\ldots ,X^{(k-1)})&\overset{(1)}{=}\frac{\sum_{\breve{S}\in \mathcal{S} }1\{ \breve{S}\in R_{S}(I^{(k)},S^{(k-1)})\} }{ \prod_{i=1}^{n}|R_{S0}(S_{i,1}^{(k-1)},\ldots ,S_{i,T}^{(k-1)})|} \\
%&=\frac{\vert R_{S}(I^{(k)},S^{(k-1)})\vert }{ \prod_{i=1}^{n}|R_{S0}(S_{i,1}^{(k-1)},\ldots ,S_{i,T}^{(k-1)})|},
%\end{align*}
%where (1) follows from combining \eqref{eq:step2_3} and \eqref{eq:step2_4}. From here, \eqref{eq:step2_5} follows.
\end{proof}
%%%%%%%%%%%%%%%%%%%%

%%%%%%%%%%%%%%%%%%%%
\begin{lemma} \label{lemma:kandal_uniformity2}
For any $k=2,\ldots ,K$, if $\xi ^{(k)}$ is generated by Algorithm \ref{alg:S_k}, then $\xi^{(k)}$ is uniformly distributed over the set $\{\xi\in R_{S0}(\xi ^{(k-1)}):\xi_{T+1}=0)\}$ conditional on $(I^{(k)},X^{(1)},\ldots ,X^{(k-1)})$.
\end{lemma}
%%%%%%%%%%%%%%%%%%%%
\begin{proof}
By Lemma \ref{lemma:kandal_uniformity}, $\tilde{\xi}$ in step 2(a) of Algorithm \ref{alg:S_k} follows the uniform distribution on $ R_{S0}(\xi^{(k-1)})$, conditional on $(I^{(k)},X^{(1)},\ldots ,X^{(k-1)})$. step 2(b) of Algorithm \ref{alg:S_k} truncates the variable to the set $\{\xi\in R_{S0}(\xi ^{(k-1)}):\xi_{T+1}=0\}$. The desired result then follows from the fact that a truncated version of a discrete uniform distribution is uniformly distributed on the truncated set. 
\end{proof}

\begin{lemma}
\label{lemma:RS_terminal} 
For any $I\in\mathcal{I} $, if $\breve{S},\tilde{S}\in\mathcal{S}^{nT}$ satisfy the following conditions:
\begin{enumerate}[(a)]
\item $\tilde{S}_{i,1}=\breve{S}_{i,1}$ for all $i \in I$,
\item $\sum_{i\in I}\sum_{t=1}^{T-1}1\{ \tilde{S}_{i,t}=s,\tilde{S}_{i,t+1}=s^{\prime}\} =\sum_{i\in I}\sum_{t=1}^{T-1}1\{ \breve{S}_{i,t}=s,\breve{S}_{i,t+1}=s^{\prime}\} $ for all $s,s^{\prime}\in\mathcal{S}$.
%\item $\sum_{i\in I}1\{ \tilde{S}_{i,T}=s\} =\sum_{i\in I}1\{ \breve{S}_{i,T}=s\} $ for all $s\in\mathcal{S}$.
\end{enumerate}
Then, $\sum_{i\in I}1\{ \tilde{S}_{i,T}=s\} =\sum_{i\in I}1\{ \breve{S} _{i,T}=s\} $ for all $s\in\mathcal{S}$.
\end{lemma}
%%%%%%%%%%%%%%%%%%%%
\begin{proof}
For every $i\in I$ and $s\in \mathcal{S}$, note that
\begin{align}
1\{\breve{S}_{i,T}=s\}&~=~\sum_{t=1}^{T-1}1\{\breve{S}_{i,t+1}=s\}- \sum_{t=1}^{T-1}1\{\breve{S}_{i,t}=s\}+1\{\breve{S}_{i,1}=s\}\notag\\
 & ~=~\sum_{\bar{s}\in \mathcal{S}}\sum_{t=1}^{T-1}1\{\breve{S}_{i,t}=\bar{s}, \breve{S}_{i,t+1}=s\}-\sum_{\bar{s}\in \mathcal{S} }\sum_{t=1}^{T-1}1\{\breve{S}_{i,t}=s,\breve{S}_{i,t+1}=\bar{s}\}+1\{\breve{S}_{i,1}=s\}. \label{eq:RS_terminal_key_eq1}
\end{align}
By the same argument applied to $\tilde{S}\in R_{S}(I,\breve{S})$, we have that for every $i\in I$ and $s\in \mathcal{S}$,
\begin{equation}
1\{\tilde{S}_{i,T}=s\}~=~\sum_{\bar{s}\in \mathcal{S}}\sum_{t=1}^{T-1}1\{ \tilde{S}_{i,t}=\bar{s},\tilde{S}_{i,t+1}=s\}-\sum_{\bar{s}\in \mathcal{S}}\sum_{t=1}^{T-1}1\{\tilde{S}_{i,t}=s,\tilde{S}_{i,t+1}=\bar{s} ^{\prime }\}+1\{\tilde{S}_{i,1}=s\}. \label{eq:RS_terminal_key_eq2}
\end{equation}
To show this lemma, fix $s\in \mathcal{S}$ arbitrarily and consider the following argument.
\begin{align*}
\sum_{i\in I}1\{\breve{S}_{i,T}=s\} &~\overset{(1)}{=}~\sum_{\bar{s}\in \mathcal{S}}\sum_{i\in I}\sum_{t=1}^{T-1}1\{\breve{S}_{i,t}=\bar{s},\breve{S}_{i,t+1}=s\}-\sum_{ \bar{s}\in \mathcal{S}}\sum_{i\in I}\sum_{t=1}^{T-1}1\{\breve{S} _{i,t}=s,\breve{S}_{i,t+1}=\bar{s}\}+\sum_{i\in I}1\{\breve{S} _{i,1}=s\} \\
&~\overset{(2)}{=}~\sum_{\bar{s}\in \mathcal{S}}\sum_{i\in I}\sum_{t=1}^{T-1}1\{\tilde{S}_{i,t}=\bar{s},\tilde{S}_{i,t+1}=s\}-\sum_{ \bar{s}\in \mathcal{S}}\sum_{i\in I}\sum_{t=1}^{T-1}1\{\tilde{S} _{i,t}=s,\tilde{S}_{i,t+1}=\bar{s}\}+\sum_{i\in I}1\{\tilde{S} _{i,1}=s\} \\
&~\overset{(3)}{=}~\sum_{i\in I}1\{\tilde{S}_{i,T}=s\},
\end{align*}
where (1) holds by \eqref{eq:RS_terminal_key_eq1}, (2) by conditions (a)-(b), and (3) by \eqref{eq:RS_terminal_key_eq2}.
\end{proof}
%%%%%%%%%%%%%%%%%%%%

\subsection{Implementation of step 3 in our MCMC algorithm}\label{sec:step3}

For any $k=2,\ldots ,K$, $X^{(k-1)} \in \mathcal{X}$, and $S^{(k)}\in \mathcal{S}^{nT}$, step 3 of our MCMC algorithm draws $A^{(k)}$ uniformly within $R_{A}(S^{(k)},X^{(k-1)})$. This can be implemented by the following algorithm.

%%%%%%%%%%%%%%%%%%%%
\begin{algorithm}[Generation of $A^{(k)}$]
\label{alg:A_k} 
For any $k=2,\ldots ,K$ and given $(X^{(1)},\ldots ,X^{(k-1)},I^{k},S^{(k)})$,
$A^{(k)}$ is randomly generated as follows
\begin{itemize}%[leftmargin=1.75\parindent,align=left,labelwidth=\parindent,labelsep=4pt,itemsep=0.5ex,topsep=0.5ex] 
\item Step 1: 
For every $(s,s^{\prime})\in\mathcal{S}\times\mathcal{S}$, define
\begin{align*}
\mathrm{Index}^{(k-1)}(s,s^{\prime})
&~=~
\{(i,t)\in\{1,\dots,n\}\times \{1,\dots,T-1\}:(S^{(k-1)}_{i,t},S^{(k-1)}_{i,t+1})=(s,s^{\prime})\}
\\
\mathrm{Index}^{(k)}(s,s^{\prime})
&~=~
\{(i,t)\in\{1,\dots,n\}\times \{1,\dots,T-1\}:(S^{(k)}_{i,t},S^{(k)}_{i,t+1})=(s,s^{\prime})\}
\\
\mathrm{Index}^{(k-1)}(s)
&~=~
\{(i,T):i\in\{1,\dots,n\},S^{(k-1)}_{i,T}=s\}
\\
\mathrm{Index}^{(k)}(s)
&~=~
\{(i,T):i\in\{1,\dots,n\},S^{(k)}_{i,T}=s\}.
\end{align*}
\item Step 2: For every $(s,s^{\prime})\in\mathcal{S}\times\mathcal{S}$, we generate $(A_{i,t}^{(k)}:(i,t)\in\mathrm{Index}^{(k)}(s,s^{\prime}))$ by uniformly sampling from  $(A_{i,t}^{(k-1)}:(i,t)\in\mathrm{Index}^{(k-1)}(s,s^{\prime}))$  without replacement.
\item Step 3: 
For every $s\in\mathcal{S}$, we construct $(A_{i,T}^{(k)}:(i,T)\in\mathrm{Index}^{(k)}(s))$ by uniformly sampling from the discrete set $(A_{i,T}^{(k-1)}:(i,T)\in\mathrm{Index}^{(k-1)}(s))$ without replacement. \hfill $\blacksquare$
 \end{itemize}
\end{algorithm}

Lemma \ref{lem:Step3works} shows that $A^{(k)}$ generated by Algorithm \ref{alg:A_k} has the desired properties.

%%%%%%%%%%%%%%%%%%%%
\begin{lemma} \label{lem:Step3works}
For any $k=2,\ldots ,K$, $A^{(k)}$ generated by Algorithm \ref{alg:A_k} satisfies the requirements of step 3 of our MCMC algorithm, i.e., \eqref{eq:dist_2} holds.
\end{lemma}
%%%%%%%%%%%%%%%%%%%%
\begin{proof}
This follows from noting that any element of $R_{A}(S^{(k)},X^{(k-1)})$ corresponds to a restricted set of permutations of the action data, and Algorithm \ref{alg:A_k} chooses an element uniformly within this set.
\end{proof}
%%%%%%%%%%%%%%%%%%%%

\section{Appendix to Section \ref{sec:validity}}
\begin{proof}[Proof of Theorem \ref{thm:sizeControl}]\label{appendix:proof_siz_cont}
By \eqref{eq:testDefn}, \eqref{eq:asySizeControl} is equivalent to $\lim \inf_{K\rightarrow \infty }(\alpha -P(\hat{p}_{K}\leq \alpha ))\geq 0$. In this proof, we show a stronger statement \citep[cf.][Eq. (15.6)]{lehmann/romano:2005}: 
\begin{equation*}
\underset{K\rightarrow \infty }{\lim \inf }~\inf_{u\in [0,1]}~(u-P(\hat{p}_{K}\leq u))~~\geq ~~0.
\end{equation*}

Fix $\varepsilon >0$ arbitrarily for the rest of the proof. It suffices to find $\bar{K}(\varepsilon )<\infty $ such that
\begin{equation}
\inf_{u\in [0,1]} (u-P(\hat{p}_{K}\leq u))~\geq ~-2\varepsilon ~~~~\text{for all $K\geq \bar{K} (\varepsilon )$.}  \label{eq:sizeControl_0}
\end{equation}%
%for all sufficiently large $K$. %We do this in the remainder of the proof.

For any $K\in \mathbb{N}$ and $X\in \mathcal{X}$, let
\begin{equation*}
\mathcal{E}_{K}~\equiv ~\Bigg\{\sup_{t\in \mathbb{R}}\Bigg\vert\frac{1}{K} \sum_{k=1}^{K}1\{\tau (X^{(k)})\geq t\}-\frac{1}{|\mathbf{G}|}\sum_{g\in  \mathbf{G}}1\{\tau (g(X))\geq t\}\Bigg\vert>\varepsilon \Bigg\},
\end{equation*}%
where $(X^{(k)}:k=1,\dots ,K)$ is produced by the MCMC algorithm \ref{alg:MCMC} and $\mathbf{G}$ is the transformation group defined in Definition \ref{def:G}. By Lemma \ref{lem:MCMCconv}, for all $X \in \mathcal{X}$, there is $\bar{K}(\varepsilon ,X)<\infty $ such that $P(\mathcal{E}_{K}\mid X)\leq \varepsilon $ for all $K\geq \bar{K}(\varepsilon ,X)$. Since $\mathcal{X}$ is a finite set, we can define $\bar{K}(\varepsilon )\equiv \max_{X\in \mathcal{X}}K(\varepsilon ,X)<\infty $. Therefore,  $P(\mathcal{E}_{K}\mid X)\leq \varepsilon $ for all $X \in \mathcal{X}$ and $K\geq \bar{K}(\varepsilon )$. By this and the law of total probability, 
\begin{equation}
P(\mathcal{E}_{K})~\leq ~\varepsilon ~~~~\text{for all $K\geq \bar{K}(\varepsilon )$.}  \label{eq:sizeControl_2}
\end{equation}%
%for all sufficiently large $K$.

For any $K\in \mathbb{N}$ and $u\in [0,1]$, consider the following derivation: 
\begin{align}
P(\hat{p}_{K}\leq u)& ~=~P\bigg(\frac{1}{K}\sum_{k=1}^{K}1\{\tau (X^{(k)})\geq \tau (X)\}\leq u\bigg)  \notag \\
& ~\leq ~P\bigg(\frac{1}{|\mathbf{G}|}\sum_{g\in \mathbf{G}}1\{\tau (g(X))\geq \tau (X)\}\leq u+\varepsilon \bigg)+P(\mathcal{E}_{K})  \notag \\
& ~\overset{(1)}{\leq }~u+\varepsilon +P(\mathcal{E}_{K}),
\label{eq:sizeControl_1}
\end{align}%
where (1) holds by Lemma \ref{lem:PvalueProperty}. By \eqref{eq:sizeControl_1}, we have that
\begin{align}
\inf_{u\in [0,1]} (u-P(\hat{p}_{K}\leq u))~\geq ~-\varepsilon -P(\mathcal{E}_{K})~~~~\text{for all $K \in \mathbb{N}$.}
\label{eq:sizeControl_3}
\end{align}%
By \eqref{eq:sizeControl_2} and \eqref{eq:sizeControl_3}, we conclude that \eqref{eq:sizeControl_0} holds, as desired.
\end{proof}

\begin{example}[Computing $\mathbf{G}$ in two simple cases]\label{ex:G_explicit}
For illustration, we compute $\mathbf{G}$ in purposely simple data configurations with $n=2$ markets, $T=2$ time periods, and binary actions and states. These examples illustrate that the restrictions that define the set $\mathbf{G}$ requires thoughtful consideration even in simple examples. 

First, consider state support equal to $\mathcal{S} =  \{1,2\}$ and a trivial actions, i.e., $\mathcal{A} =\{1\}$. In this case, we have $|\mathcal{X}| = 2^4 = 16$ possible data configurations, corresponding to $S_{i,t}=1$ or $S_{i,t}=2$ for each $(i,t) \in \{1,2\}^2$. In principle, this yields $16!$ possible transformations from $\mathcal{X}$ onto itself. With trivial actions, the only restrictions to the transformations on $\mathbf{G}$ are those generated by $R_S(I,\breve{S})$ in Definition \ref{def:RS}. For example, restriction (a) in $R_S(I,\breve{S})$ impedes any transformation from altering the states in the first period or within any market. By restriction (c) in $R_S(I,\breve{S})$, we can only interchange state information in the second market when their states in the first period coincide. From these restrictions, we deduce that there are only four elements in $\mathbf{G}$:
\begin{enumerate}
\item $g_{1}(X)=X$ for all $X\in\mathcal{X}$, i.e., the identity transformation. In this case, there is no interchange in the state in the second period.
\item $g_{2}(X)$ is defined as follows:
\begin{eqnarray*}
g_{2}\left(\left(
\begin{array}{cc}
1 & 1 \\
1 & 2
\end{array}
\right),\left(
\begin{array}{cc}
1 & 1 \\
1 & 1
\end{array}
\right)\right)&=&\left(\left(
\begin{array}{cc}
1 & 2 \\
1 & 1
\end{array}
\right),\left(
\begin{array}{cc}
1 & 1 \\
1 & 1
\end{array}
\right)\right), \\
g_{2}\left(\left(
\begin{array}{cc}
1 & 2 \\
1 & 1
\end{array}
\right),\left(
\begin{array}{cc}
1 & 1 \\
1 & 1
\end{array}
\right)\right)&=&\left(\left(
\begin{array}{cc}
1 & 1 \\
1 & 2
\end{array}
\right),\left(
\begin{array}{cc}
1 & 1 \\
1 & 1
\end{array}
\right)\right),
\end{eqnarray*}
and $g_{2}(X)=X$ for any other $X\in\mathcal{X}$. In this case, there is an interchange of states in the second period only when the states in the first period are equal to one.
\item $g_{3}(X)$ is defined as follows:
\begin{eqnarray*}
g_{3}\left(\left(
\begin{array}{cc}
2 & 1 \\
2 & 2
\end{array}
\right),\left(
\begin{array}{cc}
1 & 1 \\
1 & 1
\end{array}
\right)\right)&=&\left(\left(
\begin{array}{cc}
2 & 2 \\
2 & 1
\end{array}
\right),\left(
\begin{array}{cc}
1 & 1 \\
1 & 1
\end{array}
\right)\right), \\
g_{3}\left(\left(
\begin{array}{cc}
2 & 2 \\
2 & 1
\end{array}
\right),\left(
\begin{array}{cc}
1 & 1 \\
1 & 1
\end{array}
\right)\right)&=&\left(\left(
\begin{array}{cc}
2 & 1 \\
2 & 2
\end{array}
\right),\left(
\begin{array}{cc}
1 & 1 \\
1 & 1
\end{array}
\right)\right),
\end{eqnarray*}
and $g_{3}(X)=X$ for any other $X\in\mathcal{X}$.  In this case, there is an interchange of states in the second period only when the states in the first period are equal to two.
\item $g_{4}(X)$ is defined as follows:
\begin{eqnarray*}
g_{4}\left(\left(
\begin{array}{cc}
1 & 1 \\
1 & 2
\end{array}
\right),\left(
\begin{array}{cc}
1 & 1 \\
1 & 1
\end{array}
\right)\right)&=&\left(\left(
\begin{array}{cc}
1 & 2 \\
1 & 1
\end{array}
\right),\left(
\begin{array}{cc}
1 & 1 \\
1 & 1
\end{array}
\right)\right), \\
g_{4}\left(\left(
\begin{array}{cc}
1 & 2 \\
1 & 1
\end{array}
\right),\left(
\begin{array}{cc}
1 & 1 \\
1 & 1
\end{array}
\right)\right)&=&\left(\left(
\begin{array}{cc}
1 & 1 \\
1 & 2
\end{array}
\right),\left(
\begin{array}{cc}
1 & 1 \\
1 & 1
\end{array}
\right)\right), \\
g_{4}\left(\left(
\begin{array}{cc}
2 & 1 \\
2 & 2
\end{array}
\right),\left(
\begin{array}{cc}
1 & 1 \\
1 & 1
\end{array}
\right)\right)&=&\left(\left(
\begin{array}{cc}
2 & 2 \\
2 & 1
\end{array}
\right),\left(
\begin{array}{cc}
1 & 1 \\
1 & 1
\end{array}
\right)\right), \\
g_{4}\left(\left(
\begin{array}{cc}
2 & 2 \\
2 & 1
\end{array}
\right),\left(
\begin{array}{cc}
1 & 1 \\
1 & 1
\end{array}
\right)\right)&=&\left(\left(
\begin{array}{cc}
2 & 1 \\
2 & 2
\end{array}
\right),\left(
\begin{array}{cc}
1 & 1 \\
1 & 1
\end{array}
\right)\right),
\end{eqnarray*}
and $g_{4}(X)=X$ for any other $X\in\mathcal{X}$.  In this case, there is an interchange of states in the second period only when the states in the first period coincide.
\end{enumerate}
This example illustrates that the restrictions that define the set $R_S(I,\breve{S})$ in Definition \ref{def:RS} can significantly constrain the transformations in $\mathbf{G}$.

Second, we consider an action support equal to $\mathcal{A} =\{1,2\}$ and trivial states, i.e., $\mathcal{S} =  \{1\}$. As in the previous example, there are $|\mathcal{X}| = 16$ possible data configurations, now corresponding to $A_{i,t}=1$ or $A_{i,t}=2$ for each $(i,t) \in \{1,2\}^2$, which yields $16!$ possible transformations from $\mathcal{X}$ onto itself. With trivial states, the only restrictions to the transformations on $\mathbf{G}$ are those generated by $R_A(\tilde{S},(\breve{S},\breve{A}))$ in Definition \ref{def:RA}. Furthermore, since $\tilde{S}=\breve{S}$ equals a $4\times 4$ matrix of ones, any permutation of the actions is allowed. As a corollary, we conclude that any $g \in \mathbf{G}$ has to satisfy the following restrictions:
\begin{enumerate}
\item $g\left( \left( 
\begin{array}{cc}
1 & 1 \\ 
1 & 1%
\end{array}%
\right) ,\left( 
\begin{array}{cc}
1 & 1 \\ 
1 & 1%
\end{array}%
\right) \right) =\left( \left( 
\begin{array}{cc}
1 & 1 \\ 
1 & 1%
\end{array}%
\right) ,\left( 
\begin{array}{cc}
1 & 1 \\ 
1 & 1%
\end{array}%
\right) \right) $,
\item $g\left( \left( 
\begin{array}{cc}
1 & 1 \\ 
1 & 1%
\end{array}%
\right) ,\left( 
\begin{array}{cc}
2 & 2 \\ 
2 & 2%
\end{array}%
\right) \right) =\left( \left( 
\begin{array}{cc}
1 & 1 \\ 
1 & 1%
\end{array}%
\right) ,\left( 
\begin{array}{cc}
2 & 2 \\ 
2 & 2%
\end{array}%
\right) \right) $,
\item Let $\mathcal{A}_{1}=\left\{ \left( 
\begin{array}{cc}
2 & 1 \\ 
1 & 1%
\end{array}%
\right) ,\left( 
\begin{array}{cc}
1 & 2 \\ 
1 & 1%
\end{array}%
\right) ,\left( 
\begin{array}{cc}
1 & 1 \\ 
1 & 2%
\end{array}%
\right) ,\left( 
\begin{array}{cc}
1 & 1 \\ 
2 & 1%
\end{array}%
\right) \right\} $. For any $A \in \mathcal{A}_{1}$, $g\left( \left( 
\begin{array}{cc}
1 & 1 \\ 
1 & 1%
\end{array}%
\right) ,A\right) =\left( \left( 
\begin{array}{cc}
1 & 1 \\ 
1 & 1%
\end{array}%
\right) ,\tilde{A}\right) $ with $\tilde{A}\in \mathcal{A}_{1}$. Since $\mathcal{A}_{1}$ has 4 elements, this restriction produces $4!$ permutations.
\item Let $\mathcal{A}_{2}=\left\{ \left( 
\begin{array}{cc}
2 & 2 \\ 
1 & 1%
\end{array}%
\right) ,\left( 
\begin{array}{cc}
2 & 1 \\ 
2 & 1%
\end{array}%
\right) ,\left( 
\begin{array}{cc}
2 & 1 \\ 
1 & 2%
\end{array}%
\right) ,\left( 
\begin{array}{cc}
1 & 2 \\ 
2 & 1%
\end{array}%
\right) ,\left( 
\begin{array}{cc}
1 & 2 \\ 
1 & 2%
\end{array}%
\right) ,\left( 
\begin{array}{cc}
1 & 1 \\ 
2 & 2%
\end{array}%
\right) \right\} $. For any $A \in \mathcal{A}_{2}$,  $g\left( \left( 
\begin{array}{cc}
1 & 1 \\ 
1 & 1%
\end{array}%
\right) ,A\right) =\left( \left( 
\begin{array}{cc}
1 & 1 \\ 
1 & 1%
\end{array}%
\right) ,\tilde{A}\right) $ with $\tilde{A}\in \mathcal{A}_{2}$.  Since $\mathcal{A}_{2}$ has 6 elements, this restriction produces $6!$ permutations.
\item Let $\mathcal{A}_{3}=\left\{ \left( 
\begin{array}{cc}
1 & 2 \\ 
2 & 2%
\end{array}%
\right) ,\left( 
\begin{array}{cc}
2 & 1 \\ 
2 & 2%
\end{array}%
\right) ,\left( 
\begin{array}{cc}
2 & 2 \\ 
1 & 2%
\end{array}%
\right) ,\left( 
\begin{array}{cc}
2 & 2 \\ 
2 & 1%
\end{array}%
\right) \right\} $. For any $A \in \mathcal{A}_{3}$, $g\left( \left( 
\begin{array}{cc}
1 & 1 \\ 
1 & 1%
\end{array}%
\right) ,A\right) =\left( \left( 
\begin{array}{cc}
1 & 1 \\ 
1 & 1%
\end{array}%
\right) ,\tilde{A}\right) $ with $\tilde{A}\in \mathcal{A}_{3}$.  Since $\mathcal{A}_{3}$ has 4 elements, this restriction produces $4!$ permutations.
\end{enumerate}
These five configurations are mutually exclusive and exhaust all possible action data. The set $\mathbf{G}$ has  $4! \times 4! \times 6! = 414,720$ elements, and is obtained by selecting one permutation from these five configurations.

To conclude, we note that these examples feature a setup in which either the action or the state data are trivial. In situations where both of these are nontrivial, the computation of $\mathbf{G}$ can become challenging, even when the number of markets and time periods remains low. Furthermore, as either of these increases, our experience is that the computation of $\mathbf{G}$ quickly becomes unmanageable.
\end{example}

%%%%%%%%%%%%%%%%%%%%
\begin{proof}[Proof of Lemma \ref{thm:Lik2}]
Note that 
\begin{equation}\label{eq:S_lik2}
P(S=\tilde{S})~=~\prod_{i=1}^{n}\left(m_{i}(\tilde{S}_{i,1})\prod_{t=1}^{T-1}\left(\sum_{ a\in\mathcal{A}}f(\tilde{S}_{i,t+1}|a,\tilde{S}_{i,t})\sigma(a |\tilde{S}_{i,t})\right)\right).
\end{equation}
This equation follows from the next derivation:
\begin{align*}
P(S=\tilde{S})
&~\overset{(1)}{=}~
\prod_{i=1}^{n}\left(P(S_{i,1}=\tilde{S}_{i,1})\prod_{t=2}^{T}P(S_{i,t}=\tilde{S}_{i,t}|(S_{i,1},\ldots,S_{i,t-1})=(\tilde{S}_{i,1},\ldots,\tilde{S}_{i,t-1}))\right)\\
&~\overset{(2)}{=}~
\prod_{i=1}^{n}\left(P(S_{i,1}=\tilde{S}_{i,1})\prod_{t=2}^{T}P(S_{i,t}=\tilde{S}_{i,t}|S_{i,t-1}=\tilde{S}_{i,t-1})\right)\\
&~=~
\prod_{i=1}^{n}\left( P(S_{i,1}=\tilde{S}_{i,1})
     \prod_{t=2}^{T}\left(\sum_{a\in\mathcal{A}}
     \left(\begin{array}{c}
         P(S_{i,t}=\tilde{S}_{i,t}|A_{i,t-1}=a,S_{i,t-1}=\tilde{S}_{i,t-1})\\
         \times P(A_{i,t-1}=a|S_{i,t-1}=\tilde{S}_{i,t-1})
     \end{array}\right)
     \right)
\right)\\
&~\overset{(3)}{=}~
\prod_{i=1}^{n}\left(m_{i}(\tilde{S}_{i,1})\prod_{t=2}^{T}\left(\sum_{ a\in\mathcal{A}}f(\tilde{S}_{i,t+1}|a,\tilde{S}_{i,t})\sigma(a |\tilde{S}_{i,t})\right)\right),
\end{align*}
where (1) holds by Assumption \ref{ass:M}(a), (2) by Lemma \ref{lem:S_markov}, and (3) by $H_{0}$ in \eqref{eq:HT}.

To conclude the proof, it suffices to show \eqref{eq:Lik2_a_s} and \eqref{eq:Lik2_s}. 
To this end, consider the following derivation:
\begin{align}
P(X=\tilde{X})
&~\overset{(1)}{=}~
\prod_{i=1}^{n}\left(m_{i}(\tilde{S}_{i,1})\sigma(\tilde{A}_{i,T}|\tilde{S}_{i,T})\prod_{t=1}^{T-1}\left(\sigma(\tilde{A}_{i,t}|\tilde{S}_{i,t})f(\tilde{S}_{i,t+1}|\tilde{S}_{i,t},\tilde{A}_{i,t})\right)\right)\notag\\
&~\overset{(2)}{=}~
P(S =\tilde{S})\left(\sigma(\tilde{A}_{i,T}|\tilde{S}_{i,T})\left(\prod_{t=1}^{T-1}\frac{\sigma(\tilde{A}_{i,t}|\tilde{S}_{i,t})f(\tilde{S}_{i,t+1}|\tilde{S}_{i,t},\tilde{A}_{i,t})}{ \sum_{a\in\mathcal{A}}f(\tilde{S}_{i,t+1}|a,\tilde{S}_{i,t})\sigma(a|\tilde{S}_{i,t})} \right)\right), \label{eq:Lik2_1}
\end{align}
where (1) holds by \eqref{eq:lik1}, which is shown in Lemma \ref{lem:Lik1}, and (2) by \eqref{eq:S_lik2}. By combining \eqref{eq:Lik2} and \eqref{eq:Lik2_1}, we conclude that
$$
P(A=\tilde{A}|S=\tilde{S})
~=~
\prod_{i=1}^{n}\left(\sigma(\tilde{A}_{i,T}|\tilde{S}_{i,T})\left(\prod_{t=1}^{T-1}\frac{\sigma(\tilde{A}_{i,t}|\tilde{S}_{i,t})f(\tilde{S}_{i,t+1}|\tilde{S}_{i,t},\tilde{A}_{i,t})}{\sum_{a\in\mathcal{A} }f(\tilde{S}_{i,t+1}|a,\tilde{S}_{i,t})\sigma(a|\tilde{S}_{i,t})}\right)\right).
$$
By re-expressing this equation in terms of counts of $(s,a,s^{\prime})\in\mathcal{S}\times\mathcal{A}\times\mathcal{S}$, \eqref{eq:Lik2_a_s} follows.
Moreover, \eqref{eq:Lik2_s} follows from re-expressing \eqref{eq:S_lik2} in terms of individual counts for each $(s,s^{\prime})\in\mathcal{S}\times\mathcal{S}$. 
\end{proof}
%%%%%%%%%%%%%%%%%%%%

%%%%%%%%%%%%%%%%%%%%
\begin{proof}[Proof of Lemma \ref{lemma:Gdefn}]
%[Proof of Lemma \ref{lemma:Gdefn}]
We first show that $\mathbf{G}$ is a collection of transformations from $\mathcal{X}$ onto itself. Consider any $g\in \mathbf{G}$. By definition, $g$ is the composition of a finite number of transformations in $\bigcup_{I\in \mathcal{I}}\mathbf{G}(I)$, i.e., $g=g^{(K)}\circ \cdots \circ g^{(1)}$ with $(g^{(1)},\ldots ,g^{(K)})\in \mathbf{G}(I^{(1)})\times \cdots \times \mathbf{G}(I^{(K)})$ with $I^{(j)}\in \mathcal{I}$ for $j=1,\ldots ,K$. By Lemma \ref{lem:GI_is_group}, $g^{(j)}\in \mathbf{G}(I^{(j)})$ are onto transformations from $\mathcal{X}$ to itself. From this, we can conclude that $g=g^{(K)}\circ \cdots \circ g^{(1)}$ is an onto transformation from $ \mathcal{X}$ to itself, as desired.

Second, we show that $\mathbf{G}$ is a transformation group. To this end, it suffices to verify conditions (i)-(iv) in \citet[Section A.1]{lehmann/romano:2005}. To verify condition (i), consider arbitrary $g_{1},g_{2}\in \mathbf{G}$. By definition, this implies $g_{1}$ and $g_{2}$ are compositions of a finite number of transformations in $\bigcup_{I\in \mathcal{I}}\mathbf{G}(I)$. Then, $g_{2}\circ g_{1}$ is a composition of a finite number of elements in $ \bigcup_{I\in \mathcal{I}}\mathbf{G}(I)$, and so $g_{2}\circ g_{1}\in \mathbf{G}$.
Condition (ii) follows from the argument in \citet[page 693]{lehmann/romano:2005}. 
%To show this explicitly, note that for any $ \tilde{X}\in \mathcal{X}$ and $g_{1},g_{2},g_{3}\in \mathbf{G}$,
%\[
%((g_{1}\circ g_{2})\circ g_{3})\tilde{X}=g_{1}(g_{2}(g_{3}(\tilde{X} )))=(g_{1}\circ (g_{2}\circ g_{3}))\tilde{X}.
%\]
Condition (iii) follows from the fact that $\mathbf{G}(I)$ is a transformation group for any $I\in\mathcal{I} $ (shown in Lemma \ref{lem:GI_is_group}), and so it includes the identity transformation.
%Lemma \ref{lem:GI_is_group} implies that, for any $I\in\mathcal{I} $, $\mathbf{G}(I)$ is a group and so it includes the identity transformation. From here, we conclude that the identity transformation is included in $\mathbf{G}$.
To verify condition (iv), consider the following argument for any arbitrary $g\in \mathbf{G}$. By definition, $g$ is the composition of a finite number of transformations in $\bigcup_{I\in \mathcal{I}}\mathbf{G}(I)$, i.e., $g=g^{(K)}\circ \cdots \circ g^{(1)}$ with $(g^{(1)},\ldots ,g^{(K)})\in \mathbf{G}(I^{(1)})\times \cdots \times \mathbf{G}(I^{(K)})$ with $I^{(j)}\in \mathcal{I}$ for $j=1,\ldots ,K$. By Lemma \ref{lem:GI_is_group}, $\mathbf{G}(I^{(j)})$ is a transformation group for each $j=1,\ldots ,K$. From this, we can conclude that $\exists ( g^{(j)}) ^{-1}\in \mathbf{G}(I^{(j)})$ for each $j=1,\ldots ,K$. Since $g\circ \tilde{g}$ and $\tilde{g}\circ g$ are equal to the identity transformation, $\tilde{g}=g^{-1}$.  Finally, note that $g^{-1}=( g^{(1)}) ^{-1}\circ \cdots \circ ( g^{(K)}) ^{-1}$ is the compositions of a finite number of transformations in $\bigcup_{I\in \mathcal{I}}\mathbf{G}(I)$ and so $ g^{-1}\in \mathbf{G}$, as desired.

To complete the proof, it suffices to show that, for any $\tilde{X}\in \mathcal{X}$ and $g\in \mathbf{G}$, $\tilde{X}$ and $g\tilde{X}$ have the same sufficient statistics in \eqref{eq:SuffStatistics}. $g$ is the composition of a finite number of transformations in $\bigcup_{I\in \mathcal{ I}}\mathbf{G}(I)$, i.e., $g=g^{(K)}\circ \cdots \circ g^{(1)}$ with $ (g^{(1)},\ldots ,g^{(K)})\in \mathbf{G}(I^{(1)})\times \cdots \times \mathbf{ G}(I^{(K)})$ with $I^{(j)}\in \mathcal{I}$ for $j=1,\ldots ,K$. Therefore, $g \tilde{X}=g^{(K)}\circ \cdots \circ g^{(1)}\tilde{X}$. For each $j=1,\ldots ,K$, Lemma \ref{lem:GI_keeps_suff_stat}  implies that, for any $\breve{X}\in \mathcal{X}$ and $g^{( j) }\in \mathbf{G}(I^{(j)})$, $g^{( j) }\breve{X}$ and $\breve{X}$ have the same sufficient statistic in \eqref{eq:SuffStatistics}. From these observations and by finite induction, it follows that $\tilde{X}$ and $g\tilde{X}$ have the same sufficient statistics in \eqref{eq:SuffStatistics}, as desired.
\end{proof}
%%%%%%%%%%%%%%%%%%%%

%%%%%%%%%%%%%%%%%%%%
\begin{proof}[Proof of Lemma \ref{lem:PvalueProperty}]
By Lemma \ref{lemma:Gdefn}, we know that: (i) $\mathbf{G}$ is a finite group of transformations of $\mathcal{X}$ onto itself, and (ii) if Assumption \ref{ass:M} and $H_{0}$ in \eqref{eq:HT} hold, then $X$ and $gX$ have the same sufficient statistics in \eqref{eq:SuffStatistics} for any $g \in \mathbf{G}$. 
The second statement, together with Lemma \ref{thm:Lik2}, implies that the randomization hypothesis holds (see \citealt[Definition 15.2.1]{lehmann/romano:2005}), i.e., if Assumption \ref{ass:M} and $H_{0}$ hold, its distribution is invariant under the transformations in $\mathbf{G}$. 
Under these conditions, the result follows from \citet[Eq. (15.6) and Problem 15.2]{lehmann/romano:2005}.
\end{proof}
%%%%%%%%%%%%%%%%%%%%

%%%%%%%%%%%%%%%%%%%%
\begin{proof}[Proof of Lemma \ref{lem:MCMCconv}] Let $(G^{(1)},\ldots ,G^{(K)})$ be as in Definition \ref{def:couplingG}. Note that this definition implies that $X$ and $(G^{(1)},\ldots ,G^{(K)})$ are independent. We condition on $X \in \mathcal{X}$ throughout this proof.  By Lemma \ref{lem:equidist}, it suffices to show that
\begin{equation} 
\sup_{t\in\mathbb{R}}\left\vert \frac{1}{K}\sum_{k=1}^{K}1\{ \tau(G^{(k)}X)\geq t\} -\frac{1}{|\mathbf{G}|} \sum_{g\in\mathbf{G}}1\{\tau(gX)\geq t\}\right\vert \overset{a.s.}{\to }0~~\text{ as }K\to\infty. \label{eq:ASconv}
\end{equation}

For any $k=1,\ldots ,K$, Definition \ref{def:couplingG} implies that $G^{(k)}X \in\mathcal{X}$. Thus, $\tau(G^{(k)}X)$ takes values in the finite set $\{\tau(\tilde{X}):\tilde{X}\in\mathcal{X}\}$. It then suffices to show the pointwise version of \eqref{eq:ASconv}, i.e.,
$$
\frac{1}{K}\sum_{k=1}^{K}1\{\tau(G^{(k)}X)\geq t\}~~\overset{a.s.}{\to } ~~\frac{1}{|\mathbf{G}|}\sum_{g\in\mathbf{G}}1\{\tau(gX)\geq t\}~~\text{ as } K\to\infty .
$$

By Definition \ref{def:couplingG}, $(G^{(1)},\ldots,G^{(K)})$ is the result of a Markov chain with transition probability given in \eqref{eq:MarkovChain_2}. By \citet[Algorithm A-24 and pages 270-1]{robert/casella:2004}, we can equivalently interpret $(G^{(1)},\ldots,G^{(K)})$ as the outcome of a Metropolis-Hastings algorithm. 
For any $g,\breve{g}\in\mathbf{G}$, this Metropolis-Hastings algorithm has a conditional density $q(\breve{g}\mid g)\equiv P(G^{(k+1)}=\breve{g}|G^{(k)}=g)$, a target probability $p$ defined by
\begin{equation}
p(g)~\equiv~ \frac{1\{g\in\mathbf{G}\}}{|\mathbf{G}|}, \label{eq:MCMCconv_1}
\end{equation}
and Metropolis-Hastings acceptance probability equal to one. To show the latter, note that, for every $g,\breve{g}\in\mathbf{G}$,
$$
\rho(g,\breve{g})~=~\min \left\{ ~\frac{p(\breve{g})}{p(g)}\times \frac{q(g\mid \breve{g})}{q(\breve{g}\mid g)},~1~\right\} ~\overset{(1)}{=}~1,
$$
where (1) uses that $p(\breve{g})=p(g)=1/|\mathbf{G}|$ and $q(g\mid \breve{g})=q(\breve{g}\mid g)$ by \eqref{eq:MCMCconv_1} and Lemma \ref{lem:q_symmetric}, respectively. 
By this and \citet[Theorem 7.4]{robert/casella:2004}, it suffices to show that the conditional density $q(\breve{g}\mid g)$ is $p$-irreducible. By \citet[Theorem 6.15, part 
(i)]{robert/casella:2004}, this follows from showing that, for any $g,\breve{g}\in\mathbf{G}$ (and so $p(g)>0$ and $p(\breve{g})>0$), the Markov chain has a positive probability of transitioning from $g$ to $\breve{g}$ after a sufficient number of steps. We devote the rest of the proof to show this fact.

Consider any arbitrary choice of $g,\breve{g}\in\mathbf{G}$. Since $\mathbf{G}$ is the transformation group generated by finitely many compositions of elements in $\bigcup_{I\in\mathcal{I} }\mathbf{G}(I)$, there are $(g^{(1)},\ldots,g^{(K_{1}+K_{2})})\in\mathbf{G}(I^{(1)})\times\cdots\times\mathbf{G}(I^{(K_{1}+K_{2})})$ with $I^{(j)}\in\mathcal{I}$ for $j=1,\dots,K$ such that $ g=g^{(K_{1})}\circ \cdots \circ g^{(1)}$ and $\breve{g}=g^{(K_{1}+K_{2})}\circ \cdots \circ g^{(K_{1}+1)}$. By Lemma \ref{lem:GI_is_group}, $\mathbf{G}(I^{(j)})$ is a transformation group for all $j=1,\ldots ,K_{1}+K_{2}$, and so $(g^{(j)})^{-1}\in\mathbf{G}(I^{(j)})$ for every $j=1,\ldots ,K_{1}+K_{2}$. Then, note that
\begin{align}
\breve{g}
& ~=~
\breve{g}\circ g^{-1}\circ g\nonumber\\ 
& ~\overset{(1)}{=}~
g^{(K_{1}+K_{2})}\circ \cdots \circ g^{(K_{1}+1)}\circ(g^{(1)})^{-1}\circ \cdots \circ(g^{(K_{1})})^{-1}\circ g\nonumber\\
& ~\overset{(2)}{=}~
\breve{g}^{(K_{1}+K_{2})}\circ \cdots \circ \breve{g}^{(K_{1}+1)}\circ \breve{g}^{(K_{1})}\circ \cdots \circ \breve{g}^{(1)}\circ g,
\label{eq:path}
\end{align}
where (1) holds by setting $\breve{g}=g^{(K_{1}+K_{2})}\circ \cdots \circ g^{(K_{1}+1)} $ and $g^{-1}=(g^{(1)})^{-1}\circ \cdots \circ(g^{(K_{1})})^{-1}$, and (2) by defining $(\breve{g}^{(1)},\ldots,\breve{g}^{(K_1+K_2)})=((g^{(K_{1})})^{-1},\ldots,(g^{(1)})^{-1},g^{(K_{1}+1)},\ldots,g^{(K_{1}+K_2)})$.
Note that \eqref{eq:path} provides a specific path for transitioning from $g$ to $ \breve{g}$ after $K_{1}+K_{2}$ steps. We complete the proof by showing that $P(G^{(K_{1}+K_{2}+k)}=\breve{g}|G^{(k)}=g)>0$ for any positive integer $k$. 
To this end, we define $(\breve{I}^{(1)},\ldots,\breve{I}^{(K_1+K_2)})=(I^{(K_{1})},\ldots,I^{(1)},I^{(K_{1}+1)},\ldots,I^{(K_{1}+K_2)})$, and consider the following argument:
\begin{align}
P(G^{(K_{1}+K_{2}+k)}=\breve{g}|G^{(k)}=g) ~&\overset{(1)}{\geq }~
q(\breve{ g}^{(1)}\circ g|g)\prod\limits_{k=2}^{K}q(\breve{g}^{(k)}\circ \cdots \circ \breve{g}^{(1)}\circ g\mid \breve{g}^{(k-1)}\circ \cdots \circ \breve{g} ^{(1)}\circ g)\nonumber\\
 ~&\overset{(2)}{\geq }~
\prod_{j=1}^{K_{1}+K_{2}}\frac{1}{|\mathcal{I}| |\mathbf{G}(\breve{I}^{(j)})|}
~\overset{(3)}{>}
~0,\nonumber
\end{align}
where (1) holds by the fact that the conditional distribution of $G^{(j+1)}$ given $G^{(j)}$ is $q$ for all $j=1,\dots ,K_{1}+K_{2}$, (2) by \eqref{eq:MarkovChain_2} and $\breve{g}^{(j)}\in\mathbf{G}(\breve{I}^{(j)})$ for all $j=1,\dots ,K_{1}+K_{2}$, and (3) by $\breve{I}^{(j)}\in\mathcal{I}$ for all $j=1,\dots ,K_{1}+K_{2}$.
\end{proof}
%%%%%%%%%%%%%%%%%%%%

% THE ONLINE APPENDIX STARTS HERE
\subsection{Auxiliary results}\label{sec:app-lemmas}
%\label{appendix:aux_results}

%%%%%%%%%%%%%%%%%%%%
\begin{lemma} \label{lem:Lik1}
Under Assumptions \ref{ass:M} and $H_{0}$ in \eqref{eq:HT}, \eqref{eq:lik1} holds.
\end{lemma}
%%%%%%%%%%%%%%%%%%%%
\begin{proof}
Consider the following derivation:
\begin{align*}
P(X=\tilde{X})
%& =
%P((S_{i,t},A_{i,t})=(\tilde{S}_{i,t},\tilde{A}_{i,t}):i=1,\ldots ,n,t=1,\ldots ,T)\\
& ~\overset{(1)}{=}~
\prod_{i=1}^{n}P((S_{i,t},A_{i,t})=(\tilde{S}_{i,t},\tilde{A}_{i,t}):t=1,\ldots ,T)\\
%& =
%\prod_{i=1}^{n}P(S_{i,1}=\tilde{S}_{i,1},A_{i,1}=\tilde{A}_{i,1})\prod_{t=2}^{T}P((S_{i,t},A_{i,t})=(\tilde{S}_{i,t},\tilde{A}_{i,t})|(S_{i,1},A_{i,1},\ldots,S_{i,t-1},A_{i,t-1})=(\tilde{S}_{i,1},\tilde{A}_{i,1},\ldots,\tilde{S}_{i,t-1},\tilde{A}_{i,t-1}))\\
& ~\overset{(2)}{=}~\prod_{i=1}^{n}
\left[\begin{array}{c}
P(S_{i,1}=\tilde{S}_{i,1},A_{i,1}=\tilde{A}_{i,1})\times\\
\prod_{t=2}^{T}P((S_{i,t},A_{i,t})=(\tilde{S}_{i,t},\tilde{A}_{i,t})|(S_{i,t-1},A_{i,t-1})=(\tilde{S}_{i,t-1},\tilde{A}_{i,t-1}))
\end{array}\right]\\
%& =
%\prod_{i=1}^{n}P(S_{i,1}=\tilde{S}_{i,1})P(A_{i,1}=\tilde{A}_{i,1}|S_{i,1}=\tilde{S}_{i,1})\prod_{t=2}^{T}P(A_{i,t}=\tilde{A}_{i,t}|S_{i,t}=\tilde{S}_{i,t},S_{i,t-1}=\tilde{S}_{i,t-1},A_{i,t-1}=\tilde{A}_{i,t-1})P(S_{i,t}=\tilde{S}_{i,t}|S_{i,t-1}=\tilde{S}_{i,t-1},A_{i,t-1}=\tilde{A}_{i,t-1})\\
& ~\overset{(3)}{=}~
\prod_{i=1}^{n}\left[ 
\begin{array}{c}
     P(S_{i,1}=\tilde{S}_{i,1})\left(\prod_{t=1}^{T}P(A_{i,t}=\tilde{A}_{i,t}|S_{i,t}=\tilde{S}_{i,t})\right)\\
     \times \left(\prod_{{t} =1}^{T-1}P(S_{i,t+1}=\tilde{S}_{i,t+1}|S_{i,t}=\tilde{S}_{i,t},A_{i,t}=\tilde{A}_{i,t})\right) 
\end{array}
\right] \\
& ~\overset{(4)}{=}~
\prod_{i=1}^{n}\Big[ m_{i}(\tilde{S}_{i,1})\Big(\prod_{t=1}^{T}\sigma(\tilde{A}_{i,t}|\tilde{S}_{i,t})\Big)\Big(\prod_{{t} =1}^{T-1}f(\tilde{S}_{i,{t}+1}|\tilde{S}_{i,{t}},\tilde{A}_{i,{t}})\Big)\Big] ,
\end{align*}
where (1) holds by Assumption \ref{ass:M}(a), (2) by Assumption \ref{ass:M}(b), (3) by Assumption \ref{ass:M}(c), and (4) by $ H_{0}$ in \eqref{eq:HT}.
\end{proof}
%%%%%%%%%%%%%%%%%%%%

\begin{lemma} \label{lem:S_markov}
Under Assumptions \ref{ass:M}(b)-(c), the state variable is Markovian, i.e., for every $i=1,\ldots ,n$ and $t=2,\ldots T$ and every $\tilde{S}\in\mathcal{S}^{nT}$,
\begin{equation}
P(S_{i,t}=\tilde{S}_{i,t}|S_{i,t-1}=\tilde{S}_{i,t-1})=P(S_{i,t}=\tilde{S}_{i,t}|(S_{i,1},\ldots,S_{i,t-1})=(\tilde{S}_{i,1},\ldots,\tilde{S}_{i,t-1})).
\label{eq:S_markov_0}
\end{equation}
\end{lemma}
%%%%%%%%%%%%%%%%%%%%
\begin{proof}%[Proof of Lemma \ref{lem:S_markov}]
Fix $i=1,\ldots ,n$, $t=2,\ldots T$, and $\tilde{S}\in\mathcal{S}^{nT}$ arbitrarily. Consider the following argument.
\begin{align*}
& P((S_{i,1},\ldots,S_{i,t})=(\tilde{S}_{i,1},\ldots,\tilde{S}_{i,t}))\\
&= \sum_{(a_1,\ldots ,a_{t-1})\in\mathcal{A}^{t-1}}
\left(
\begin{array}{c}P((S_{i,1},A_{i,1},\ldots,S_{i,t-1},A_{i,t-1})=(\tilde{S}_{i,1},a_{1},\ldots,\tilde{S}_{i,t-1},a_{t-1})) \times\\ P(S_{i,t}=\tilde{S}_{i,t}|(S_{i,1},A_{i,1},\ldots,S_{i,t-1},A_{i,t-1})=(\tilde{S}_{i,1},a_{1},\ldots,\tilde{S}_{i,t-1},a_{t-1}))\end{array}\right)\\
& \overset{(1)}{=}
\sum_{(a_1,\ldots ,a_{t-1})\in\mathcal{A}^{t-1}}
\left(
\begin{array}{c}P((S_{i,1},A_{i,1},\ldots,S_{i,t-1},A_{i,t-1})=(\tilde{S}_{i,1},a_{1},\ldots,\tilde{S}_{i,t-1},a_{t-1}))\\ \times P(S_{i,t}=\tilde{S}_{i,t}|(S_{i,t-1},A_{i,t-1})=(\tilde{S}_{i,t-1},a_{t-1}))\end{array}\right)\\
&=
\sum_{a_{t-1}\in\mathcal{A}}\left(
\begin{array}{c}P((S_{i,1},\ldots,S_{i,t-1})=(\tilde{S}_{i,1},\ldots,\tilde{S}_{i,t-1}),A_{i,t-1}=a_{t-1})\\ \times P(S_{i,t}=\tilde{S}_{i,t}|(S_{i,t-1},A_{i,t-1})=(\tilde{S}_{i,t-1},a_{t-1}))\end{array}\right)\\
&=
\sum_{a_{t-1}\in\mathcal{A}}\left(
\begin{array}{c}
P(A_{i,t-1}=a_{t-1}|(S_{i,1},\ldots,S_{i,t-1})=(\tilde{S}_{i,1},\ldots,\tilde{S}_{i,t-1})) \\
\times P((S_{i,1},\ldots,S_{i,t-1})=(\tilde{S}_{i,1},\ldots,\tilde{S}_{i,t-1})) \\
\times P(S_{i,t}=\tilde{S}_{i,t}|(S_{i,t-1},A_{i,t-1})=(\tilde{S}_{i,t-1},a_{t-1}))
\end{array}\right)\\
& \overset{(2)}{=}
\sum_{a_{t-1}\in\mathcal{A}}\left(
\begin{array}{c}
P(A_{i,t-1}=a_{t-1}|S_{i,t-1}=\tilde{S}_{i,1})
P((S_{i,1},\ldots,S_{i,t-1})=(\tilde{S}_{i,1},\ldots,\tilde{S}_{i,t-1}))\\ \times P(S_{i,t}=\tilde{S}_{i,t}|(S_{i,t-1},A_{i,t-1})=(\tilde{S}_{i,t-1},a_{t-1}))
\end{array}\right)\\
&=
P(S_{i,t}=\tilde{S}_{i,t}|S_{i,t-1}=\tilde{S}_{i,t-1})P((S_{i,1},\ldots,S_{i,t-1})=(\tilde{S}_{i,1},\ldots,\tilde{S}_{i,t-1})),
\end{align*}
where (1) holds by Assumption \ref{ass:M}(b) and (2) by Assumption \ref{ass:M}(c). From here, \eqref{eq:S_markov_0} follows from dividing both sides by $P((S_{i,1},\ldots,S_{i,t-1})=(\tilde{S}_{i,1},\ldots,\tilde{S}_{i,t-1}))>0$.
\end{proof}
%%%%%%%%%%%%%%%%%%%%

%%%%%%%%%%%%%%%%%%%%
\begin{lemma} \label{lem:GI_is_group} 
For any $I\in\mathcal{I}$, $\mathbf{G}(I)$ is a transformation group.
\end{lemma}
%%%%%%%%%%%%%%%%%%%%
\begin{proof}
We fix $I\in \mathcal{I}$ arbitrarily. It suffices to verify conditions (i)-(iv) in \citet[Section A.1]{lehmann/romano:2005}. We can verify condition (ii) using the same argument as in \citet[page 693]{lehmann/romano:2005}, so we focus the rest of the proof on conditions (i), (iii), and (iv).

We begin with condition (i). First, for any arbitrary $g_{1},g_{2}\in \mathbf{G}(I)$, we now verify that $g_{2}\circ g_{1}\in \mathbf{G}(I)$. Since $g_{1},g_{2}\in \mathbf{G}(I)$, $g_{1}$ and $g_{2}$ are both onto transformations of $\mathcal{X}$ onto itself, then $g_{2}\circ g_{1}$ is an onto transformation of $\mathcal{X}$ onto itself. Now we will show that, for any $(\breve{S},\breve{A})\in\mathcal{X}$, the data configuration $(\tilde{S},\tilde{A})=(g_{2}\circ g_{1})(\breve{S},\breve{A})$ satisfies $\tilde{S} \in R_{S}(I,\breve{S})$ and $\tilde{A} \in R_{A}(\tilde{S},(\breve{S},\breve{A}))$. 
% Define $(\dot{S},\dot{A})=g_1(\breve{S},\breve{A})$, and so $(\tilde{S},\tilde{A})=g_2(\dot{S},\dot{A})$. 
Since $g_{1},g_{2}\in \mathbf{G}(I)$ and the conditions in Definitions \ref{def:RS} and \ref{def:RA} are expressed as equalities, it is easy to see that $\tilde{S} \in R_{S}(I,\breve{S})$ and $\tilde{A} \in R_{A}(\tilde{S},(\breve{S},\breve{A}))$. By combining these results, we conclude that $g_{2}\circ g_{1}\in \mathbf{G}(I)$, as desired.

Condition (iii) says that the identity transformation belongs to $\mathbf{G}(I)$. To this end, note that the identity transformation is an onto transformation of $\mathcal{X}$ onto itself and that $\breve{S} \in R_{S}(I,\breve{S})$ and $\breve{A} \in R_{A}(\breve{S},(\breve{S},\breve{A}))$.

To verify condition (iv), we now show that for any $g\in \mathbf{G}(I)$, $ g^{-1}\in \mathbf{G}(I)$ holds. Fix $g\in \mathbf{G}(I)$ arbitarily. We first show that $ g^{-1}$ is well defined. Recall that $\mathbf{G}(I)$ is a collection of transformations that map a finite set $\mathcal{X}$ onto itself, and, thus, $g(\mathcal{X})=\mathcal{X}$. To show $ g^{-1}$ is well defined, it suffices to show that $g$ is not many-to-one. If $g$ were many-to-one, then there would be two distinct elements $\tilde{X}_1,\tilde{X}_2 \in \mathcal{X}$ such that $g(\tilde{X}_1)=g(\tilde{X}_2)$. In such case, $|\{\tilde{X}_1,\tilde{X}_2\}|=2>1=|g(\{\tilde{X}_1,\tilde{X}_2\})|$, and so $|\mathcal{X}|>|g(\mathcal{X})|$, contradicting that $g$ is onto.
Having shown that $ g^{-1}$ is well defined, we will  now show that $ g^{-1}\in \mathbf{G}(I)$. For this purpose, pick $\tilde{X}\in \mathcal{X}$ arbitrarily. It suffices to show that, for any $(\breve{S},\breve{A})\in\mathcal{X}$, the data configuration $(\tilde{S},\tilde{A})=g^{-1}(\breve{S},\breve{A})$ satisfies $\tilde{S} \in R_{S}(I,\breve{S})$ and $\tilde{A} \in R_{A}(\tilde{S},(\breve{S},\breve{A}))$.
Since $g\in \mathbf{G}(I)$ and $g(\tilde{S},\tilde{A})=g(g^{-1}(\breve{S},\breve{A}))=(\breve{S},\breve{A})$, we have 
$\breve{S} \in R_{S}(I,\tilde{S})\mbox{ and }\breve{A} \in R_{A}(\breve{S},(\tilde{S},\tilde{A}))$.
Note that all the conditions in Definitions \ref{def:RS} and \ref{def:RA} treat $(\tilde{S},\tilde{A})$ and $(\breve{S},\breve{A})$ symmetrically.
Therefore, we have $\tilde{S} \in R_{S}(I,\breve{S})$ and $\tilde{A} \in R_{A}(\tilde{S},(\breve{S},\breve{A}))$, as desired.
\end{proof}
%%%%%%%%%%%%%%%%%%%%

%%%%%%%%%%%%%%%%%%%%
\begin{lemma} \label{lem:GI_keeps_suff_stat} 
For any $I\in\mathcal{I}$ and any $g\in\mathbf{G}(I)$, 
$\breve{X}$ and $g\breve{X}$ have the same sufficient statistic in \eqref{eq:SuffStatistics}, i.e., $U(\breve{X})= U(g\breve{X})$.
\end{lemma}
%%%%%%%%%%%%%%%%%%%%
\begin{proof}
Let $\breve{X}=(\breve{S},\breve{A})$ and $\tilde{X}=(\tilde{S},\tilde{A})=g( \breve{S},\breve{A})$. By definition \ref{def:GI}, this implies that $\tilde{S}\in R_{S}( \breve{S}) $ and $\tilde{A}\in R_{A}( \tilde{S},(\breve{S},\breve{ A})) $. By \eqref{eq:SuffStatistics}, it then suffices to show the following statements:
\begin{enumerate}
\item $\breve{S}_{i,1}=\tilde{S}_{i,1}$ for all $i=1,\ldots ,n$,
\item $\sum_{i=1}^{n}\sum_{t=1}^{T-1}1\{\tilde{S}_{i,t}=s,\tilde{A}_{i,t}=a, \tilde{S}_{i,t+1}=s^{\prime }\}=\sum_{i=1}^{n}\sum_{t=1}^{T-1}1\{\breve{S} _{i,t}=s,\breve{A}_{i,t}=a,\breve{S}_{i,t+1}=s^{\prime }\}$ for all $ s,s^{\prime }\in \mathcal{S}$ and $a\in \mathcal{A}$,
\item $\sum_{i=1}^{n}1\{\tilde{S}_{i,T}=s,\tilde{A}_{i,T}=a\}= \sum_{i=1}^{n}1\{\breve{S}_{i,T}=s,\breve{A}_{i,T}=a\}$ for all $s\in \mathcal{S}$ and $a\in \mathcal{A}$.
\end{enumerate}
The first statement follows from $\tilde{S}\in R_{S}( \breve{S}) $ and condition (a) in Definition \ref{def:RS}. The second and third statements follow from $\tilde{A}\in R_{A}( \tilde{S},(\breve{S},\breve{A})) $ and conditions (a) and (b) in Definition \ref{def:RA}, respectively.
\end{proof}
%%%%%%%%%%%%%%%%%%%%

Several upcoming results involve the Markov chain of transformations in $\mathbf{G}$, specified in Definition \ref{def:couplingG}.

\begin{definition}\label{def:couplingG}
Let $(G^{(1)},\ldots ,G^{(K)})$ denote a Markov chain of transformations of $\mathcal{X}$ onto itself that is defined as follows: 
\begin{itemize}
    \item $G^{(1)}:\mathcal{X} \to \mathcal{X}$ be equal to the identity transformation, i.e., $x =G^{(1)}x$ for any $x \in  \mathcal{X}$.
    \item For any $k=2,\ldots ,K$ and given $(G^{(1)},\ldots,G^{(k-1)},X)$, $G^{(k)}:\mathcal{X} \to \mathcal{X}$ is a random transformation distributed according to the following transition probability:
\begin{align}
P(G^{(k)}=\tilde{g}\mid G^{(1)},\ldots,G^{(k-1)},X) 
=P(G^{(k)}=\tilde{g}\mid G^{(k-1)})
=\sum_{I\in\mathcal{I}}\sum_{{g}\in\mathbf{G}(I)}\frac{1\{\tilde{g}={g}\circ(G^{(k-1)})\}}{|\mathcal{I}| \times |\mathbf{G}(I)|}. \label{eq:MarkovChain_2}
\end{align}
\end{itemize}
\end{definition}

\begin{lemma}\label{lem:equidist} 
Conditional on $X$, $(X^{(1)},\ldots ,X^{(K)})$ generated by our MCMC algorithm and $(G^{(1)}X,\ldots ,G^{(K)}X)$ with $(G^{(1)},\ldots ,G^{(K)})$ as in Definition \ref{def:couplingG} have the same distribution.
\end{lemma}
\begin{proof}
We condition on $X$ throughout this proof. First, note that our MCMC algorithm and Definition \ref{def:couplingG} imply that $X = X^{(1)} = G^{(1)}X $. Second, note that $(X^{(1)},\ldots ,X^{(K)})$ and $(G^{(1)}X,\ldots ,G^{(K)}X)$ are both Markov chains in $\mathcal{X}$. To complete the proof, it suffices to show that they have the same transition probabilities. As implied by equations \eqref{eq:dist_1} and \eqref{eq:dist_2}, the transition probability of $(X^{(1)},\ldots ,X^{(K)})$ is:
\begin{align}
&P(X^{(k)}=\tilde{X}\mid X^{(1)},\ldots,X^{(k-1)}) ~=~%~=~ P(X^{(k)}=\tilde{X}\mid X^{(k-1)}) 
\notag\\
&
\begin{cases}
\sum_{I\in\mathcal{I}}\dfrac{1\{\tilde{S}\in R_{S}(I,S^{(k-1)}),\tilde{A}\in R_{A}(\tilde{S},X^{(k-1)})\}}{|\mathcal{I}| \times   |R_{S}(I,S^{(k-1)})|\times |R_{A}(\tilde{S}, X^{(k-1)})|} 
&\text{if}~|R_{S}(I,S^{(k-1)})| \times  |R_{A}(\tilde{S}, X^{(k-1)})|>0,\\
0&\text{otherwise.}
\end{cases}\label{eq:dist_X}
\end{align}
It then suffices to show that, for any $k=2,\dots,K$, $\tilde{X}= (\tilde{S},\tilde{A}) \in \mathcal{X}$, and $G^{(k-1)}X = \breve{X}=(\breve{S}, \breve{A})  \in \mathcal{X}$, 
\begin{align}
&P(G^{(k)}X=\tilde{X}\mid G^{(1)}X,\ldots,G^{(k-2)}X,G^{(k-1)}X=\breve{X},X)\notag\\
&~=~ P(G^{(k)}X=\tilde{X}\mid G^{(k-1)}X=\breve{X},X) \notag\\
&~=~
\begin{cases}
\sum_{I\in\mathcal{I}}\dfrac{1\{\tilde{S}\in R_{S}(I,\breve{S}),\tilde{A}\in R_{A}(\tilde{S},\breve{X})\}}{|\mathcal{I}| \times   |R_{S}(I,\breve{S})|\times |R_{A}(\tilde{S}, \breve{X})|} 
&\mbox{ if }|R_{S}(I,\breve{S})| \times  |R_{A}(\tilde{S}, \breve{X})|>0,\\
0&\mbox{ otherwise.}
\end{cases}\label{eq:equidist_1}
\end{align}

For the rest of the proof, we fix $k=2,\ldots ,K$, and $\tilde{X}=(\tilde{S},\tilde{A}),\breve{X}=(\breve{S}, \breve{A})\in\mathcal{X}$ arbitrarily. To show \eqref{eq:equidist_1}, consider the following derivation:
\begin{align}
& P(G^{(k)}X=\tilde{X}\mid G^{(1)}X,\ldots,G^{(k-2)}X,G^{(k-1)}X=\breve{X},X)\nonumber\\
& ~\overset{(1)}{=}~
E[P(G^{(k)}X=\tilde{X}\mid G^{(1)},\ldots,G^{(k-1)},X)\mid G^{(1)}X,\ldots,G^{(k-2)}X,G^{(k-1)}X=\breve{X},X]\nonumber\\
& ~\overset{(2)}{=}~E[P(G^{(k)}X=\tilde{X}\mid G^{(k-1)},X)\mid G^{(1)}X,\ldots,G^{(k-2)}X,G^{(k-1)}X=\breve{X},X], \label{eq:equidist_2}
\end{align}
where (1) holds by the law of total probability and (2) by \eqref{eq:MarkovChain_2}. 
From \eqref{eq:equidist_2}, \eqref{eq:equidist_1} follows if we show that, for $G^{(k-1)}X=\breve{X}$,
\begin{align}
P(G^{(k)}X=\tilde{X}\mid G^{(k-1)},X)=\begin{cases}
\sum_{I\in\mathcal{I}} \dfrac{1\{\tilde{S}\in R_{S}(I,\breve{S}),\tilde{A}\in R_{A}(\tilde{S}, \breve{X})\}}{|\mathcal{I}| \times |R_{S}(I,\breve{S})|\times |R_{A}(\tilde{S},\breve{X})|}&\mbox{ if }|R_{S}(I,\breve{S})|\times |R_{A}(\tilde{S},\breve{X})|>0,\\ 
0&\mbox{ otherwise.}
\end{cases}  \label{eq:equidist_3}
\end{align}
To show \eqref{eq:equidist_3}, consider the following derivation:
\begin{align}
P(G^{(k)}X=\tilde{X}\mid G^{(k-1)},X)
& ~\overset{(1)}{=}~
P(G^{(k)}(G^{(k-1)})^{-1}\breve{X}=\tilde{X}\mid G^{(k-1)},X) \nonumber \\
& ~=~\sum_{g\in \mathbf{G}}P(G^{(k)}=g\mid G^{(k-1)},X)1\{g(G^{(k-1)})^{-1} \breve{X}=\tilde{X}\} \nonumber \\
& ~\overset{(2)}{=}~\sum_{g\in \mathbf{G}}\frac{\sum_{I\in \mathcal{I}}\frac{1 }{|\mathcal{I}|}\sum_{\tilde{g}\in \mathbf{G}(I)}1\{g=\tilde{g}\circ (G^{(k-1)})^{-1}\}}{|\mathbf{G}(I)|}1\{g(G^{(k-1)})^{-1}\breve{X}=\tilde{X}\} \nonumber \\
& ~=~\frac{1}{|\mathcal{I}|}\sum_{I\in \mathcal{I}}\frac{\sum_{\tilde{g}\in \mathbf{G}(I)}\sum_{g\in \mathbf{G}}1\{g=\tilde{g}\circ (G^{(k-1)})\}1\{g(G^{(k-1)})^{-1}\breve{X}=\tilde{X}\}}{|\mathbf{G}(I)|} \nonumber \\
& \overset{(3)}{=}\frac{1}{|\mathcal{I}|}\sum_{I\in \mathcal{I}}\frac{\sum_{ \tilde{g}\in \mathbf{G}(I)}\sum_{g\in \mathbf{G}}1\{g=\tilde{g}\circ (G^{(k-1)})\}1\{\tilde{g}\breve{X}=\tilde{X}\}}{|\mathbf{G}(I)|} \nonumber \\
& ~=~\frac{1}{|\mathcal{I}|}\sum_{I\in \mathcal{I}}\frac{\sum_{\tilde{g}\in \mathbf{G}(I)}1\{\tilde{g}\breve{X}=\tilde{X}\}\sum_{g\in \mathbf{G}}1\{g= \tilde{g}\circ (G^{(k-1)})\}}{|\mathbf{G}(I)|} \nonumber \\
&~ \overset{(4)}{=}~\frac{1}{|\mathcal{I}|}\sum_{I\in \mathcal{I}}\frac{\sum_{ \tilde{g}\in \mathbf{G}(I)}1\{\tilde{g}\breve{X}=\tilde{X}\}}{|\mathbf{G}(I)| }, \label{eq:equidist_4}
\end{align}
where (1) holds by $G^{(k-1)}X=\breve{X}$ and the fact that $ (G^{(k-1)})^{-1}\in \mathbf{G}$ since $\mathbf{G}$ is a transformation group (by Lemma \ref{lemma:Gdefn}), (2) by \eqref{eq:MarkovChain_2}, (3)  by the fact that $\{g=\tilde{g}\circ (G^{(k-1)})\}$ occurs if and only if $ \{g(G^{(k-1)})^{-1}=\tilde{g}\}$, and (4) by the fact that $\sum_{g\in \mathbf{G }}1\{g=\tilde{g}\circ (G^{(k-1)})\}=1$, which is shown in the next paragraph. 

We now show that $\sum_{g\in \mathbf{G}}1\{g=\tilde{g}\circ (G^{(k-1)})\}=1$. Since $g,\tilde{g},G^{(k-1)}\in \mathbf{G}$, and $\mathbf{G}$ is a transformation group, $\tilde{g}\circ (G^{(k-1)})\in \mathbf{G}$, and so  $\exists g\in \mathbf{G}$ s.t.\ $g=\tilde{g}\circ (G^{(k-1)})$, i.e., $\sum_{g\in \mathbf{G}}1\{g=\tilde{g}\circ (G^{(k-1)})\}\geq 1$. Now, suppose that $ \sum_{g\in \mathbf{G}}1\{g=\tilde{g}\circ (G^{(k-1)})\}>1$. This implies that $\exists g_{1},g_{2}\in \mathbf{G}$ with $g_{1}\not=g_{2}$ s.t.\ $g_{1}= \tilde{g}\circ (G^{(k-1)})=g_{2}$. But using again that $\mathbf{G}$ is a transformation group, $\exists g_{1}^{-1}\in \mathbf{G}$ and so $ g_{1}^{-1}g_{2}=g_{1}^{-1}g_{1}$ and $g_{2}g_{1}^{-1}=g_{1}g_{1}^{-1}$ and both equal to the identity transformation. This would imply that $ g_{1}^{-1}=g_{2}^{-1}$, and since the inverse transformation is unique, we reach a contradiction.

Fix $I\in \mathcal{I}$ arbitrarily. By \eqref{eq:equidist_4}, \eqref{eq:equidist_3} then follows from showing that
\begin{equation}
\sum_{\tilde{g}\in \mathbf{G}(I)}\frac{1\{\tilde{g}\breve{X}=\tilde{X}\}}{| \mathbf{G}(I)|}~=~
\begin{cases} 
\dfrac{1\{\tilde{S}\in R_{S}(I,\breve{S}),\tilde{A}\in R_{A}(\tilde{S}, \breve{X})\}}{|R_{S}(I,\breve{S})|\times |R_{A}(\tilde{S},\breve{X})|}&\mbox{ if }|R_{S}(I,\breve{S})|\times|R_{A}(\tilde{S},\breve{X})|>0,\\
0&\mbox{ otherwise.}
\end{cases} \label{eq:equidist_3_new}
\end{equation}
We divide our argument into two cases. First, consider $|R_{S}(I,\breve{S} )|\times |R_{A}(\tilde{S},\breve{X})|=0$. In this case, we have ${\not{\exists}}g\in \mathbf{G}(I)$ s.t.\ $g\breve{X}=\tilde{X}$, and therefore
\[
\sum_{g\in \mathbf{G}(I)}\frac{1\{g\breve{X}=\tilde{X}\}}{|\mathbf{G}(I)|}~=~0,
\]
which verifies \eqref{eq:equidist_3_new}. 

Second, consider $|R_{S}(I,\breve{S} )|\times |R_{A}(\tilde{S},\breve{X})|>0$. Then, consider the following derivation:
\begin{align}
\frac{\sum_{g\in \mathbf{G}(I)}1\{g\breve{X}=\tilde{X}\}}{|\mathbf{G}(I)|}& ~\overset{(1)}{=}~\frac{\sum_{{g}\in \mathbf{G}(I)}1\{g\breve{X}=\tilde{X}\}}{| \mathbf{G}(I)|}1\{\tilde{S}\in R_{S}(I,\breve{S}),\tilde{A}\in R_{A}(\tilde{S },\breve{X})\} \nonumber \\
& ~\overset{(2)}{=}~\frac{\sum_{g\in \mathbf{G}(I)}1\{g\breve{X}=\breve{X}\}}{| \mathbf{G}(I)|}1\{\tilde{S}\in R_{S}(I,\breve{S}),\tilde{A}\in R_{A}(\tilde{S },\breve{X})\} \nonumber \\
& ~\overset{(3)}{=}~\frac{1\{\tilde{S}\in R_{S}(I,\breve{S}),\tilde{A}\in R_{A}(\tilde{S},\breve{X})\}}{|R_{S}(I,\breve{S})|\times|R_{A}(\breve{S},\breve{X} )|} \nonumber \\
& ~\overset{(4)}{=}~\frac{1\{\tilde{S}\in R_{S}(I,\breve{S}),\tilde{A}\in R_{A}(\tilde{S},\breve{X})\}}{|R_{S}(I,\breve{S})|\times|R_{A}(\tilde{S},\breve{X} )|}, \label{eq:equidist_6}
\end{align}
which verifies \eqref{eq:equidist_3_new}, where (1) holds by Definition \ref{def:GI}, as it implies that $\{\tilde{g}\breve{X}=\tilde{X}\}$ with $ \tilde{g}\in \mathbf{G}(I)$, $\breve{X}=(\breve{S},\breve{A})$, and $\tilde{ X}=(\tilde{S},\tilde{A})$ implies that $\{\tilde{S}\in R_{S}(I,\breve{S})\}$ and $\{\tilde{A}\in R_{A}(\tilde{S},\breve{X})\}$, (2) by Lemma \ref{lem:card1}, (3) by \eqref{eq:equidist_5}, and (4) by Lemma \ref{lem:card2} (which applies because the expression is multiplied by $1\{\tilde{S}\in R_{S}(I,\breve{S})\}$). 

To show (3) in \eqref{eq:equidist_6}, consider the following argument.
\begin{align}
|\mathbf{G}(I)|& \overset{(1)}{=}\sum_{\tilde{X}\in \mathcal{X}}\left( \sum_{g\in \mathbf{G}(I)}1\{g\breve{X}=\tilde{X}\}\right) \nonumber \\ & \overset{(2)}{=}\sum_{\tilde{S}\in R_{S}(I,\breve{S})}\sum_{\tilde{A}\in R_{A}(\tilde{S},\breve{X})}\left( \sum_{g\in \mathbf{G}(I)}1\{g\breve{X}= \tilde{X}\}\right) \nonumber \\
& \overset{(3)}{=}
%\left( \sum_{g\in \mathbf{G}(I)}1\{g\breve{X}=\breve{X} \}\right) \sum_{\tilde{S}\in R_{S}(I,\breve{S})}\sum_{\tilde{A}\in R_{A}( \tilde{S},\breve{X})} \nonumber \\
%& =
\left( \sum_{g\in \mathbf{G}(I)}1\{g\breve{X}=\breve{X}\}\right) \sum_{ \tilde{S}\in R_{S}(I,\breve{S})}|R_{A}(\tilde{S},\breve{X})| \nonumber \\
& \overset{(4)}{=}
% \left( \sum_{g\in \mathbf{G}(I)}1\{g\breve{X}=\breve{X} \}\right) |R_{A}(\breve{S},\breve{X})|\sum_{\tilde{S}\in R_{S}(I,\breve{S})} 1 \nonumber \\
%& =
\left( \sum_{g\in \mathbf{G}(I)}1\{g\breve{X}=\breve{X}\}\right)~ |R_{A}( \breve{S},\breve{X})|~|R_{S}(I,\breve{S})|, \label{eq:equidist_5}
\end{align}
where (1) holds by partitioning $\mathbf{G}(I)$ into its possible range of outcomes when applied to $\breve{X}\in \mathcal{X}$, (2) by Definition \ref{def:GI}, as it implies that $\{\tilde{g}\breve{X}=\tilde{X}\}$ with $\tilde{g}\in \mathbf{G}(I)$, $ \breve{X}=(\breve{S},\breve{A})$, and $\tilde{X}=(\tilde{S},\tilde{A})$ if and only if $\{\tilde{S}\in R_{S}(I,\breve{S})\}$ and $\{\tilde{A}\in R_{A}( \tilde{S},\breve{X})\}$, (3) by Lemma \ref{lem:card1}, and (4) by Lemma \ref{lem:card2}.
\end{proof}
%%%%%%%%%%%%%%%%%%%%

%%%%%%%%%%%%%%%%%%
\begin{lemma} \label{lem:card1}
Fix $\breve{X}=(\breve{S},\breve{A})\in\mathcal{X}$, $\tilde{X}=(\tilde{S},\tilde{A})\in\mathcal{X}$, and $I\in\mathcal{I} $ arbitrarily. 
Then, $\tilde{S}\in R_{S}(I,\breve{S})$ and $\tilde{A}\in R_{A}(\tilde{S},\breve{X})$ implies that $\sum_{g\in\mathbf{G}(I)}1\{g\breve{X}=\tilde{X}\}=\sum_{g\in\mathbf{G}(I)}1\{g\breve{X}=\breve{X}\}$.
\end{lemma}
%%%%%%%%%%%%%%%%%%%
\begin{proof}%[Proof of Lemma \ref{lem:card1}]
Fix $\breve{X}=(\breve{S},\breve{A})\in\mathcal{X}$, $\tilde{X}\in\mathcal{X}$, and $I\in\mathcal{I} $ arbitrarily, and assume that $\tilde{S}\in R_{S}(I, \breve{S})$ and $\tilde{A}\in R_{A}(\tilde{S},\breve{X})$. 
By definition of $\mathbf{G}(I)$, $R_{S}(I,\breve{S})$, and $R_{A}(\tilde{S}, \breve{X})$, $\tilde{S}\in R_{S}(I,\breve{S})$ and $\tilde{A}\in R_{A}(\tilde{S},\breve{X})$ implies that $\exists \breve{g}\in\mathbf{G}(I)$ s.t.\ $\breve{g}\breve{X}=\tilde{X}$. Therefore,
$$
\sum_{g\in\mathbf{G}(I)}1\{g\breve{X}=\tilde{X}\} 
~\overset{(1)}{=}~
\sum_{g\in\mathbf{G}(I)}1\{g\breve{X}=\breve{g}\breve{ X}\} 
~\overset{(2)}{=}~
\sum_{g\in\mathbf{G}(I)}1\{\breve{g}^{-1}g\breve{X}= \breve{X}\} 
~\overset{(3)}{=}~
\sum_{g\in\mathbf{G}(I)}1\{g\breve{X}=\breve{X}\},
$$
where (1) holds by $\breve{g}\breve{X}=\tilde{X}$, (2) by $\breve{ g}\in\mathbf{G}(I)$ and the fact that $\mathbf{G}(I)$ is a transformation group (by Lemma \ref{lem:GI_is_group}), (3) by $\{\breve{g}^{-1}g:g\in\mathbf{G}(I)\}=\mathbf{G}(I)$, as $\mathbf{G}(I)$ is a transformation group (again, by Lemma \ref{lem:GI_is_group}).
\end{proof}
%%%%%%%%%%%%%%%%%%%

%%%%%%%%%%%%%%%%%%%
\begin{lemma}\label{lem:card2} 
Fix $\breve{X}=(\breve{S},\breve{A})\in\mathcal{X}$ and $I\in\mathcal{I} $ arbitrarily. Then, $\tilde{S}\in R_{S}(I,\breve{S})$ implies that $|R_{A}(\tilde{S}, \breve{X})|=|R_{A}(\breve{S},\breve{X})|$.
\end{lemma}
%%%%%%%%%%%%%%%%%%%
\begin{proof}%[Proof of Lemma \ref{lem:card2}] 
Fix $\breve{X}=(\breve{S},\breve{A})\in\mathcal{X}$ and $I\in\mathcal{I} $ arbitrarily, and assume that $\tilde{S}\in R_{S}(I,\breve{S})$.

We first show that $|R_{A}(\breve{S},(\breve{S},\breve{A}))|\leq|R_{A}(\tilde{S},(\breve{S},\breve{A}))|$. 
Let $(A^{1},\ldots ,A^{C})$ enumerate the (distinct) elements in $R_{A}(\breve{S},(\breve{S},\breve{A}))$. 
By $\tilde{S}\in R_{S}(I, \breve{S})$ and Lemma \ref{lem:RA_by_transformation}, there is a permutation $ \pi $ s.t\ $\tilde{S}=\breve{S}_{\pi }$ and $A_{\pi }^{c}\in R_{A}(\tilde{S} ,(\breve{S},A^{c}))$ for each $c=1,\ldots ,C$. 
We now show that $(A_{\pi }^{1},\ldots ,A_{\pi }^{C})$ are all distinct elements. 
To this end, suppose that $\exists c_{1},c_{2}\in\{ 1,\ldots ,C\} $ s.t.\ $A_{\pi }^{c_{1}}=A_{\pi }^{c_{2}}$. 
If that were the case, and by the fact that a permutation is a bijective relationship, we conclude that $A^{c_{1}}=A^{c_{2}}$. Since $(A^{1},\ldots ,A^{C})$ are distinct, we conclude that $c_{1}=c_{2}$, as desired. 
To conclude the argument, it suffices to show that $A_{\pi }^{c}\in R_{A}(\tilde{S} ,(\breve{S},\breve{A}))$ for all $c=1,\ldots ,C$. To this end, choose $c=1,\ldots ,C$ arbitrarily. 
Since $\breve{S}\in R_{S}(I,\breve{S})$ (trivially) and $A^{c}\in R_{A}(\breve{S},(\breve{S},\breve{A}))$, Definition \ref{def:GI} implies that $\exists g_{1}\in\mathbf{G}(I)$ s.t.\ $g_{1}(\breve{S},\breve{A})=(\breve{S} ,A^{c})$. 
Since $\tilde{S}\in R_{S}(I,\breve{S})$ and $A_{\pi }^{c}\in R_{A}(\tilde{S},(\breve{S},A^{c}))$, Definition \ref{def:GI} implies that $\exists g_{2}\in\mathbf{G}(I)$ s.t.\ $ g_{2}(\breve{S},A^{c})=(\tilde{S},A_{\pi }^{c})$. 
Since $\mathbf{G}(I)$ is a transformation group (by Lemma \ref{lem:GI_is_group}), we conclude that $g_{3}=g_{2}\circ g_{1}\in\mathbf{G}(I)$. 
Since $g_{3}(\breve{S},\breve{A})=(\tilde{S},A_{\pi }^{c})$ and $ g_{3}\in\mathbf{G}(I)$, Definition \ref{def:GI} implies that $ A_{\pi }^{c}\in R_{A}(\tilde{S},(\breve{S},\breve{A}))$, as desired.

We next show that $|R_{A}(\breve{S},(\breve{S},\breve{A}))|\geq |R_{A}(\tilde{S},(\breve{S},\breve{A}))|$. Let $(A^{1},\ldots ,A^{C})$ enumerate the (distinct) elements in $R_{A}(\tilde{S},(\breve{S},\breve{A}))$. Since $\tilde{S}\in R_{S}(I, \breve{S})$ and by the fact that the Definition \ref{def:RS} treats $\tilde{S}$ and $\breve{S}$ symmetrically, we conclude that $\breve{S}\in R_{S}(I, \tilde{S})$. In turn, by $\breve{S}\in R_{S}(I,\tilde{S})$ and Lemma \ref{lem:RA_by_transformation}, there is a permutation $\pi $ s.t.\ $ \breve{S}=\tilde{S}_{\pi }$ and $A_{\pi }^{c}\in R_{A}(\breve{S},(\tilde{S},A^{c}))$ for each $c=1,\ldots ,C$. By repeating the previous argument, we can show that $(A_{\pi }^{1},\ldots ,A_{\pi }^{C})$ are all distinct elements. To conclude the proof, it suffices to show that $A_{\pi }^{c}\in R_{A}(\breve{S},(\breve{S},\breve{A}))$ for all $c=1,\ldots ,C$. To this end, choose $c=1,\ldots ,C$ arbitrarily. Since $\tilde{S}\in R_{S}(I,\breve{S})$ and $A^{c}\in R_{A}(\tilde{S},(\breve{S},\breve{A}))$, Definition \ref{def:GI} implies that $\exists g_{1}\in\mathbf{G}(I)$ s.t.\ $g_{1}(\breve{S} ,\breve{A})=(\tilde{S},A^{c})$. Since $\breve{S}\in R_{S}(I,\tilde{S})$ and $A_{\pi }^{c}\in R_{A}(\breve{S},(\tilde{S} ,A^{c}))$, Definition \ref{def:GI} implies that $\exists g_{2}\in\mathbf{G}(I)$ s.t.\ $g_{2}(\tilde{S},A^{c})=(\breve{S} ,A_{\pi }^{c})$. Since $\mathbf{G}(I)$ is a transformation group (by Lemma \ref{lem:GI_is_group}), we conclude that $g_{3}=g_{2}\circ g_{1}\in\mathbf{G}(I)$. Since $g_{3}(\breve{S},\breve{A})=(\breve{S},A_{\pi }^{c})$ and $g_{3}\in\mathbf{G}(I)$, Definition \ref{def:GI} implies that $A_{\pi }^{c}\in R_{A}(\breve{S},(\breve{S},\breve{A}))$, as desired.
\end{proof}
%%%%%%%%%%%%%%%%%%%%

%%%%%%%%%%%%%%%%%%%%
\begin{lemma}
\label{lem:RA_by_transformation} 
For any $\breve{S}\in\mathcal{S}^{nT}$, $I\in\mathcal{I}$ and $\tilde{S}\in R_{S}(I, \breve{S})$, there exists a permutation $\pi :\{1,\dots ,n\}\times \{1,\dots ,T\}\to\{1,\dots ,n\}\times \{1,\dots ,T\}$ such that $ \tilde{S}=\breve{S}_{\pi }$ and $\breve{A}_{\pi }\in R_{A}(\tilde{S},(\breve{S},\breve{A}))$ for every $\breve{A}\in\mathcal{A}^{nT}$.
\end{lemma}
%%%%%%%%%%%%%%%%%%%%
\begin{proof}%[Proof of Lemma \ref{lem:RA_by_transformation}]
Fix $\breve{S},\in\mathcal{S}^{nT}$ and $I\in\mathcal{I}$ arbitrarily and assume that $\tilde{S}\in R_{S}(I,\breve{S})$. For every $s,s^{\prime}\in\mathcal{S}$, let
\begin{align*}
\mathrm{Index}_{1}(s,s^{\prime})& ~=~\{(i,t)\in\{1,\dots,n\}\times \{1,\dots,T-1\}:(\breve{S}_{i,t},\breve{S}_{i,t+1})=(s,s^{\prime})\}, \\
\mathrm{Index}_{2}(s,s^{\prime})& ~=~\{(i,t)\in\{1,\dots,n\}\times \{1,\dots,T-1\}:(\tilde{S}_{i,t},\tilde{S}_{i,t+1})=(s,s^{\prime})\}, \\
\mathrm{Index}_{1}(s)&~=~\{(i,T):i\in\{1,\dots,n\},\breve{S}_{i,T}=s\}, \\
\mathrm{Index}_{2}(s)& ~=~\{(i,T):i\in\{1,\dots,n\},\tilde{S}_{i,T}=s\}.
\end{align*}
We use
\begin{align*}
C(s,s^{\prime})& ~\equiv~ |\mathrm{Index}_{1}(s,s^{\prime})|~\overset{(1)}{=}~| \mathrm{Index}_{2}(s,s^{\prime})|, \\
C(s)& ~\equiv~ |\mathrm{Index}_{1}(s)|~\overset{(2)}{=}~|\mathrm{Index}_{2}(s)|,
\end{align*}
where (1) and (2) hold by $\tilde{S}\in R_{S}(I,\breve{S})$. 
 
For every $s,s^{\prime}\in\mathcal{S}$, we can enumerate $\mathrm{Index}_{1}(s,s^{\prime})$ by $(\nu_{1}(1,s,s^{\prime}),\ldots ,\nu_{1}(C(s,s^{\prime}),s,s^{\prime}))$, $ \mathrm{Index}_{2}(s,s^{\prime})$ by $(\nu_{2}(1,s,s^{\prime}),\ldots ,\nu_{2}(C(s,s^{\prime}),s,s^{\prime}))$, $\mathrm{Index}_{1}(s)$ by $(\nu_{1}(1,s),\dots,\nu_{1}(C(s),s))$, and $ \mathrm{Index}_{2}(s)$ by $(\nu_{2}(1,s),\dots,\nu_{2}(C(s),s))$. By definition, $(\nu_{1}(1,s,s^{\prime}),\ldots,\nu_{1}(C(s,s^{\prime}),s,s^{\prime}))$ represent the $(i,t)$ indices that satisfy $(\breve{S}_{i,t},\breve{S}_{i,t+1})=(s,s^{\prime})$ and $(\nu_{1}(1,s),\ldots ,\nu_{1}(C(s),s))$ represent the $(i,T)$ indices that satisfy $\breve{S}_{i,T}=s$, $(\nu_{2}(1,s,s^{\prime}),\ldots ,\nu_{2}(C(s,s^{\prime}),s,s^{\prime}))$ represent the $(i,t)$ indices that satisfy $(\tilde{S}_{i,t},\tilde{S}_{i,t+1})=(s,s^{\prime})$, and $(\nu_{2}(1,s),\dots ,\nu_{2}(C(s),s))$ represents the $(i,T)$ indices that satisfy $\tilde{S}_{i,T}=s$.

These enumerations allow us to interpret $\breve{S}$ as a permutation of the values of $\tilde{S}$. 
We denote this permutation by $\pi :\{1,\dots ,n\}\times \{1,\dots ,T\}\to \{1,\dots ,n\}\times \{1,\dots ,T\}$, and characterize it next. 
For any $(i,t)\in\{1,\dots ,n\}\times \{1,\dots ,T-1\}$, there exists $(s,s^{\prime})\in\mathcal{S}$ and $c=1,\ldots ,C(s,s^{\prime})$ s.t.\ $(i,t)=\nu_{1}(c,s,s^{\prime})\in\mathrm{Index}_{1}(s,s^{\prime})$. 
In this case, set $\pi(i,t)=\nu_{2}(c,s,s^{\prime})$. By this construction,
$$
\breve{S}_{i,t}~=~\breve{S}_{\nu_{1}(c,s,s^{\prime})}~=~\tilde{S}_{\nu_{2}(c,s,s^{\prime})}~=~\tilde{S}_{\pi(i,t)},
$$
Similarly, for any $i\in\{1,\dots ,n\}$, there exists $s\in\mathcal{S}$ and $c=1,\ldots ,C(s)$ s.t.\ $(i,T)=\nu_{1}(c,s)\in\mathrm{Index}_{1}(s)$. In this case, set $\pi(i,T)=\nu_{2}(c,s)$. By this construction,
$$
\breve{S}_{i,T}~=~\breve{S}_{\nu_{2}(c,s,)}~=~\tilde{S}_{\nu_{2}(c,s)}~=~\tilde{S }_{\pi(i,T)}.
$$

To show the second part, for any $\breve{A}\in\mathcal{A}^{nT}$, consider $\tilde{A}=\breve{A}_{\pi }$. For each $s,s^{\prime}\in\mathcal{S}$, note that
\begin{align}
\tilde{A}_{\nu_{2}(c,s,s^{\prime})}& ~=~\breve{A}_{\nu_{1}(c,s,s^{\prime})}~~ \text{ for }c=1,\ldots ,C(s,s^{\prime})\nonumber\\
\tilde{A}_{\nu_{2}(c,s)}& ~=~\breve{A}_{\nu_{1}(c,s)}~~\text{ for }c=1,\ldots ,C(s). \label{eq:permut1}
\end{align}
To complete the proof, it suffices to show that $\tilde{A}\in R_{A}(\tilde{S} ,\breve{X})$. To this end, it suffices to verify conditions (a)-(b) in Definition \ref{def:RA}.
We only show condition (a), as condition (b) can be shown using an analogous argument. For any $s,s^{\prime}\in\mathcal{S}$ and $a\in\mathcal{A}$, consider the following derivation:
\begin{align}
\sum_{i=1}^{n}\sum_{t=1}^{T-1}1\{\breve{S}_{i,t}=s,\breve{A}_{i,t}=a,\breve{S}_{i,t+1}=s^{\prime}\}
& ~\overset{(1)}{=}~
\sum_{(i,t)\in\mathrm{Index}_{1}(s,s^{\prime})}1\{\breve{A}_{i,t}=a\}\nonumber\\
& ~=~\sum_{c=1}^{C(s,s^{\prime})}1\{\breve{A}_{\nu_{1}(c,s,s^{\prime})}=a\}\nonumber\\
& ~\overset{(2)}{=}~
\sum_{c=1}^{C(s,s^{\prime})}1\{\tilde{A}_{\nu_{2}(c,s,s^{\prime})}=a\}\nonumber\\
& ~=\sum_{(i,t)\in\mathrm{Index}_{2}(s,s^{\prime})}1\{\tilde{A}_{i,t}=a\}\nonumber\\
& ~\overset{(3)}{=}~
\sum_{i=1}^{n}\sum_{t=1}^{T-1}1\{\tilde{S}_{i,t}=s,\tilde{A}_{i,t}=a,\tilde{S}_{i,t+1}=s^{\prime}\}, \label{eq:permut2}
\end{align}
where (1) holds by fact that $\mathrm{Index}_{1}(s,s^{\prime})$ is the collection of all indices $(i,t)\in\{1,\dots ,n\}\times \{1,\dots ,T-1\}$ s.t.\ $(\breve{S}_{i,t},\breve{S}_{i,t+1})=(s,s^{\prime})$, (2) by \eqref{eq:permut1}, and (3) by the fact that $\mathrm{Index}_{2}(s,s^{\prime})$ is the collection of all indices $(i,t)\in\{1,\dots ,n\}\times \{1,\dots ,T-1\}$ s.t.\ $(\breve{S}_{i,t},\breve{S}_{i,t+1})=(s,s^{\prime})$.
\end{proof}
%%%%%%%%%%%%%%%%%%%%

%%%%%%%%%%%%%%%%%%%%
\begin{lemma}\label{lem:q_symmetric}
The transition probability in \eqref{eq:MarkovChain_2} is symmetric, i.e., for any $g,\breve{g}\in\mathbf{G}$, $P(G^{(k+1)}=\breve{g}|G^{(k)}=g)=P(G^{(k+1)}=g|G^{(k)}=\breve{g})$.
\end{lemma}
%%%%%%%%%%%%%%%%%%%%
\begin{proof} Fix $g,\breve{g}\in\mathbf{G }$ arbitrarily and consider the following argument.
\begin{align*}
P(G^{(k+1)}=\breve{g}|G^{(k)}=g)
&~=~\sum_{I\in\mathcal{I}}\frac{1}{|\mathcal{I}|}\sum_{\tilde{g}\in\mathbf{G}(I)}\frac{1\{\breve{g}=\tilde{g}\circ g\}}{|\mathbf{G}(I)|} \\
&~\overset{(1)}{=}~\sum_{I\in\mathcal{I}}\frac{1}{|\mathcal{I}|}\sum_{\tilde{g}\in\mathbf{G}(I)}\frac{1\{g=\tilde{g} ^{-1}\circ \breve{g}\}}{|\mathbf{G}(I)|} \\
&~\overset{(2)}{=}~\sum_{I\in\mathcal{I}}\frac{1}{|\mathcal{I}|}\sum_{\tilde{g}\in\mathbf{G}(I)}\frac{1\{g=\tilde{g}\circ \breve{g}\}}{|\mathbf{ G}(I)|}\\
&~=~P(G^{(k+1)}=g|G^{(k)}=\breve{g}),
\end{align*}
where (1) holds by the fact that $\mathbf{G}(I)$ is a transformation group (by Lemma \ref{lem:GI_is_group}), and so $\exists \tilde{g}^{-1}\in\mathbf{G }(I)$ for any $\tilde{g}\in\mathbf{G}(I)$, and that $1\{\breve{g}=\tilde{g}\circ g\}=1\{\tilde{g} ^{-1}\circ \breve{g}=g\}$, and (2) by defining $\mathbf{G}(I)=\{ \tilde{g}^{-1} :\tilde{g}\in\mathbf{G}(I)\} $, which holds because $\mathbf{G}(I)$ is a transformation group (again, by Lemma \ref{lem:GI_is_group}).
\end{proof}
%%%%%%%%%%%%%%%%%%%%

%%%%%%%%%%%%%%%%%%%%
\end{small}
\bibliography{BIBLIOGRAPHY}
%%%%%%%%%%%%%%%%%%%%
\end{document}